%% file: fw-jams-revision1.tex
\documentclass{siamltex}
\usepackage{amsmath,amssymb,mathabx}
\usepackage{graphicx,esint,ulem}
\include{abbrevs3}
\begin{document}

\title{ Honeycomb Lattice Potentials and Dirac Points }

\author{Charles L. Fefferman\footnotemark[1] and  Michael I. Weinstein\footnotemark[2]}

\maketitle

\renewcommand{\thefootnote}{\fnsymbol{footnote}}
\footnotetext[1]{Department of Mathematics, Princeton University; cf@math.princeton.edu}
\footnotetext[2]{Department of Applied Physics and Applied
  Mathematics, Columbia University; miw2103@columbia.edu}

\maketitle

\begin{abstract}
We prove that the two-dimensional Schr\"odinger operator with a potential having the symmetry of a honeycomb structure has dispersion surfaces with conical singularities ({\it Dirac points}) at the vertices of its Brillouin zone. No assumptions are made on the size of the potential. We then prove the robustness of such conical singularities to a restrictive class of perturbations, which break the honeycomb lattice symmetry. General small perturbations of potentials with Dirac points do not have Dirac points; their dispersion surfaces are smooth. The presence of Dirac points in honeycomb structures is associated with many novel electronic and optical properties of materials such as graphene.
\end{abstract}

%
%
\pagestyle{myheadings}
\thispagestyle{plain}
\markboth{Honeycomb Lattice Potentials and Dirac Points}{Charles L. Fefferman and Michael I. Weinstein}
\section{Introduction and Outline}
\label{sec:introduction-outline}

In this article we study the spectral properties of the~Schr\"odinger operator 
$
H_V=-\Delta+V(\bx),\ \ \ \bx\in\R^2,
$
where the potential, $V$, is periodic and has honeycomb structure symmetry. For general periodic potentials  the spectrum of $H_V$, considered as an operator on $L^2(\R^2)$, is the union of closed intervals of continuous spectrum called the spectral bands. Associated with each spectral band are a band dispersion function, $\mu(\bk)$, and  Floquet-Bloch states, $u(\bx;\bk)=p(\bx;\bk)e^{i\bk\cdot\bx}$, where $Hu(\bx;\bk)=\mu(\bk)u(\bx;\bk)$ and 
 $p(\bx;\bk)$ is periodic with the periodicity of $V(\bx)$. The quasi-momentum, $\bk$, varies over $\mathcal{B}$, the first Brillouin zone \cite{Kittel}. Therefore, the time-dependent Schr\"odinger equation has solutions of the form $e^{i\left(\bk\cdot\bx-\mu(\bk)t\right)}\ p(\bx;\bk)$. Furthermore, any finite energy solution of the initial value problem for the time-dependent Schr\"odinger equation is a continuum weighted superposition, an integral $d\bk$,  over such states. Thus, the  time-dynamics are strongly influenced by the character of $\mu(\bk)$ on the spectral support of the initial data. 

We investigate the properties of $\mu(\bk)$ in the case where $V=V_h$ is a {\it honeycomb lattice potential}, {\it i.e.} $V_h$ is periodic with respect to a particular lattice, $\Lambda_h$, and has honeycomb structure symmetry; see Definition \ref{honeyV}. There has been intense interest within the fundamental and applied  physics communities in such structures; see, for example, the survey articles \cite{RMP-Graphene:09,Novoselov:11}. Graphene, a single atomic layer of carbon atoms, is a two-dimensional structure with carbon atoms located at the sites of a honeycomb structure.  Most remarkable is that the associated dispersion surfaces are observed to have  conical singularities at the vertices of $\brill_h$, which in this case is a regular hexagon.
That is, locally about any such quasi-momentum vertex, $\bk\approx\bK_\star$, one has 
\begin{equation}
\mu(\bk)-\mu(\bK_\star)\ \approx\  \pm\ \left|\lambda_\sharp\right|\cdot |\bk-\bK_\star|\ , 
\label{intro-cone}
\end{equation}
for some complex constant $\lambda_\sharp\ne0$.
 A consequence is that for wave-packet initial conditions with spectral components which are concentrated near these vertices, the effective evolution equation governing the wave-packet envelope is the two-dimensional Dirac wave equation, the equation of evolution for massless relativistic fermions
  \cite{RMP-Graphene:09,Ablowitz-Zhu:11}. Hence, these special vertex quasi-momenta  associated with the hexagonal lattice  are often called {\it Dirac points}. In contrast, wave-packets concentrated at spectral band edges,  bordering a spectral gap where the dispersion relation is typically quadratic, 
behave as massive non-relativistic particles; the effective wave-packet envelope equation is the Schr\"odinger equation with inverse effective mass related to the local curvature of the band dispersion relation at the band edge.   The presence of Dirac points has many physical implications with great potential for technological applications \cite{Wong-Akinwanke:10}. Refractive index profiles with honeycomb lattice symmetry and their applications are also considered in the context of electro-magnetics 
  \cite{Haldane:08,Soljacic:08}. Also, linear and nonlinear propagation of light in a two-dimensional refractive index profile with honeycomb lattice symmetry, generated via the interference pattern of plane waves incident on a photorefractive crystal, has been investigated in \cite{Segev-etal:07,Segev-etal:08} .  In such structures, wave-packets of light with spectral components concentrated near Dirac points,  evolve diffractively (rather than dispersively) with increasing propagation distance into the crystal.
\medskip

Previous mathematical analyses of such honeycomb lattice structures are based upon extreme limit models:
\begin{enumerate}
\item the tight-binding / infinite contrast limit (see, for example, \cite{Wallace:47,RMP-Graphene:09,Kuchment-Post:07}) in which the potential is  taken to be concentrated at lattice points or edges of a graph; in this limit, the dispersion relation has an explicit analytical expression, or 
\item the weak-potential limit, treated by formal perturbation theory in
\cite{Haldane:08,Ablowitz-Zhu:11} and rigorously in \cite{Grushin:09}. 
\end{enumerate}
\medskip

  The goal of the present paper is to provide a rigorous construction of conical singularities (Dirac points) for essentially {\it any} potential with a honeycomb structure. No assumptions on smallness or largeness of the potential are made. More precisely, consider the Schr\"odinger operator
\begin{equation}
 H^{(\eps)}\ \equiv\ -\Delta\ +\ \eps V_h\ \ (\eps\ {\rm real}) \label{honey-ham}
\end{equation}
where $V_h(\bx)$ denotes a {\it honeycomb lattice potential}. These potentials are real-valued, smooth,   $\Lambda_h$- periodic and, with respect to some origin of coordinates,  inversion symmetric $(\bx\to-\bx)$ and invariant under a $2\pi/3$- rotation 
 ($\mathcal{R}$- invariance);
 see Def. \ref{honeyV}.  We also make a simple, explicit genericity assumption on $V_h(x)$; see equation \eqref{V11eq0}.\smallskip

\noindent  Our main results are:
\begin{enumerate}
\item {\it Theorem \ref{main-thm}, which states that for fixed honeycomb lattice potential $V_h$, the dispersion surface of  $H^{(\eps)}$ has conical singularities at each vertex of the hexagonal Brillouin zone, except possibly for $\eps$ in a countable and closed set, $\tilde{\mathcal{C}}$.} We do not know whether exceptional non-zero $\eps$ can occur, {\it i.e.} whether the above countable closed set can be taken to be $\{0\}$. However our proof excludes exceptional $\eps$  
from $(-\eps_0,\eps_0)\setminus\{0\}$, for some $\eps_0>0$. Moreover, for small $\eps$ these conical singularities occur either as intersections between the first and second band dispersion surfaces or between the second and third dispersion surfaces. As $\eps$ increases, there continue to be such conical intersections of dispersion surfaces, but we do not control  which dispersion surfaces intersect.
\medskip

\item {\it Theorem \ref{deformed}, which states that the conical singularities of the dispersion surface of  $H^{(\eps)}$ for $\eps\notin\tilde{\mathcal{C}}$, are robust in the  following sense: Let $W(\bx)$ be real-valued, $\Lambda_h$- periodic and inversion-symmetric (even), but not necessarily $\mathcal{R}$- invariant. Then, for all sufficiently small real $\eta$, the operator $H(\eta)=H^{(\eps)}+\eta W$ has a dispersion surface with conical-type singularities. Furthermore, these conical singularities will typically not occur at the vertices of the Brillouin zone, $\brill_h$; see also the numerical results in \cite{Segev-etal:08}}.
In Remark 
\ref{instability} we show instability of Dirac points to certain perturbations, {\it e.g.} perturbations $W$  which are $\Lambda_h$- periodic but not inversion-symmetric. The dispersion surface is locally smooth in this case.

In a forthcoming paper we prove that Dirac points persist if the honeycomb lattice is subjected to a small uniform strain.

\end{enumerate}
\smallskip

\noindent 

The paper is structured as follows. In section \ref{sec:spectr-theory-peri} we briefly outline the spectral theory of general periodic potentials. We then introduce $\Lambda_h$, the particular lattice (Bravais lattice) used to generate a honeycomb structure or ``honeycomb lattice'', the union of two interpenetrating triangular lattices.  Section \ref{sec:spectr-theory-peri} concludes with implications for Fourier analysis in this setting. Section \ref{Veq0} contains a discussion of the spectrum of the Laplacian on $L^2_\bk$, the subspace of $L^2$ satisfying pseudo-periodic boundary conditions with quasi-momentum $\bk\in\brill_h$, the Brillouin zone. We observe that degenerate eigenvalues of multiplicity three occur at the vertices of  $\brill_h$.
In section   \ref{sec:doubleimpliescone} we state and prove Theorem \ref{prop:2impliescone} which reduces the construction of conical singularities of the dispersion surface at the vertices of $\brill_h$ to  establishing the existence of  two-dimensional $\mathcal{R}-$ invariant  eigenspaces of $H^{(\eps)}$ for quasi-momenta at the vertices of  $\brill_h$. 
 In section \ref{sec:MainTheorem} we give a precise statement of our main result, Theorem 
\ref{main-thm}, on conical singularities of dispersion surfaces at the vertices of $\brill_h$.   In section \ref{sec:pfeps-small} we prove for all $\eps$ sufficiently small and non-zero, by a Lyapunov-Schmidt reduction, that the degenerate, multiplicity three eigenvalue of the Laplacian splits into a multiplicity two eigenvalue and a multiplicity one eigenvalue, with  associated $\mathcal{R}$- invariant eigenspaces. 
  In order to  continue this result to $\eps$ large we introduce, in section 
  \ref{sec:det2}, a globally-defined analytic function, $\mathcal{E}(\mu,\eps)$, whose zeros, counting multiplicity, are the eigenvalues of $H^{(\eps)}$.  Eigenvalues occur where an operator 
  $I+ \mathcal{C}(\mu,\eps),\ \ \mathcal{C}(\mu,\eps)$  compact, is singular. 
 Since $\mathcal{C}(\mu,\eps)$ is not trace-class but is Hilbert-Schmidt, we work with $\mathcal{E}(\mu,\eps)=
 \det_2(I+\mathcal{C}(\mu,\eps))$, a renormalized determinant. 
 In section \ref{sec:continuation}, $\mathcal{E}(\mu,\eps)$ and $\lambda_\sharp^\eps$ (see \eqref{intro-cone}) are studied using techniques of complex function theory to establish the existence of Dirac points for arbitrary real values of $\eps$, except possibly for  a countable closed subset of $\mathbb{R}$. In section \ref{sec:deformed} we prove Theorem \ref{deformed}, which gives conditions for the local persistence of the conical singularities.  Remark \ref{instability} discusses perturbations which break the conical singularity and for which the dispersion surface is smooth. Appendix \ref{sec:topology} contains a counterexample, illustrating the topological obstruction discussed in section \ref{sec:linearalgebra}. 
 \smallskip
 
 Finally we remark that conical singularities have long been known to occur in Maxwell equations with constant anisotropic dielectric tensor;
  see \cite{Berry-Jeffrey:07} and references cited therein. \medskip

\textbf{Acknowledgments:} CLF was supported by US-NSF Grant DMS-09-01040.
 MIW was supported in part by US-NSF Grant DMS-10-08855. 
The authors would like to thank Z.H. Musslimani for stimulating discussions early in this work.
 We are also grateful to  M.J. Ablowitz, B. Altshuler, J. Conway, W. E, P. Kuchment, J. Lu, C. Marianetti,  A. Millis and G. Uhlmann 
 for discussions, and  K. Pankrashkin for bringing reference \cite{Grushin:09} to our attention.

\subsection{Notation} 
\label{sec:notation}
\begin{enumerate}
\item $z\in\mathbb{C}\ \implies\ \overline{z}$ denotes the complex conjugate of $z$.
\item $A$, a $d\times d$ matrix $\implies$ $A^t$ is its transpose and $A^*$ is its conjugate-transpose.
\item  $\langle \bj \rangle =\sqrt{1+|\bj|^2}.$
\item $\bK^\bm=\bK^{m_1,m_2}=\bK+\bm\bk=\bK+m_1\bk_1+m_2\bk_2$.\\  $\bK, \bk_1$ and $\bk_2$ are defined in section \ref{sec:honeycomb}.
\item $\nabla_\bk=e^{-i\bk\cdot\bx}\nabla_\bx e^{i\bk\cdot\bx}=\nabla_\bx+i\bk$,\ $\Delta_\bk=\nabla_\bk\cdot\nabla_\bk$.
\item ${\bf x}, {\bf y}\in\C^n,\ \ \left\langle {\bf x}, {\bf y}\right\rangle=\overline{\bf x}\cdot{\bf y}$,\ 
$\bx\cdot \by=x_1y_1+\dots +x_ny_n$.
\item For $\bq=(q_1,q_2)\in\mathbb{Z}^2$, $\bq\bk=q_1\bk_1+q_2\bk_2$\ .
\item $\langle f,g\rangle=\int \overline{f}g$
\item $l^2_s(\Z^2)=\left\{{\bf \xi}=\{\xi_\bj\}_{\bj\in\Z^2}\ :\ \sum_{\bj\in\Z^2}\ \langle \bj \rangle^{2s}|\xi_\bj|^2<\infty\right\}$
\end{enumerate}

\section{Periodic Potentials and Honeycomb Lattice Potentials}
\label{sec:spectr-theory-peri}

We begin this section with a  review of Floquet-Bloch theory of periodic potentials \cite{Eastham:73}, \cite{Kuchment-01}, \cite{RS4}. We then turn to the definition of honeycomb structures and their Fourier analysis.

\subsection{Floquet-Bloch Theory} 

Let $\{\bv_1,\bv_2\}$ be a linearly independent set in $\mathbb{R}^2$. Consider the lattice 
\begin{equation}
\Lambda=\{m_1\bv_1+m_2\bv_2: m_1,m_2\in\Z\ \}=\Z\bv_1\oplus\Z\bv_2,
\label{Lambda-def}
\end{equation}
The fundamental period cell is denoted 
\begin{equation}
\Omega=\{\ \theta_1\bv_1+\theta_2\bv_2: 0\le\theta_j\le1,\ j=1,2\ \}\ .
\label{Omega-def}
\end{equation}
Denote by $L^2_{per,\Lambda}\ =\ L^2(\R^2/\Lambda)$, the space of  $L^2_{\rm loc}$ functions which are periodic with the respect to the lattice $\Lambda$, or equivalently functions in $L^2$ on the torus
 $\mathbb{R}^2/\Lambda=\mathbb{T}^2$:
\begin{equation}
f\in L^2_{per,\Lambda} \ \textrm{if and only if}\  \ f(\bx+\bv)=f(\bx), \ {\rm for}\ \bx\in\R^2 ,\ \ \bv\in\Lambda\ .
\nn\end{equation}
More generally, we consider functions satisfying a pseudo-periodic boundary condition:
\begin{equation}
f\in L^2_{\bk,\Lambda}  \ \textrm{if and only if}\  \ f(\bx+\bv)=f(\bx)e^{i\bk\cdot\bv}, \ {\rm for}\ \bx\in\R^2 ,\ \ \bv\in\Lambda .
\end{equation}
We shall suppress the dependence on the period-lattice, $\Lambda$, and write $L^2_\bk$,  if the choice of lattice is clear from context. For $f$ and $g$ in $L^2_{\bk,\Lambda}$, $\overline{f}g$ is locally integrable and $\Lambda$- periodic and we define their inner product by:
\begin{equation}
\left\langle f,g\right\rangle\ =\ \int_\Omega\ \overline{f(\bx)}\ g(\bx)\ d\bx\ .
\label{inner-product}
\end{equation}
In a standard way, one can introduce the Sobolev spaces $H^s_{\bk,\Lambda}$.

The dual lattice, $\Lambda^*$, is defined to be
\begin{equation}
\Lambda^*\ =\ \{m_1\bk_1+m_2\bk_2 : m_1, m_2\in\mathbb{Z}\}=\Z\bk_1\oplus\Z\bk_2 \ ,
\label{eq:dual-lattice}
\end{equation}
where $\bk_1$ and $\bk_2$ are dual lattice vectors, satisfying the relations:
\begin{equation}
\bk_i\cdot \bv_j = 2\pi \delta_{ij}\ .
\nn\end{equation}

If $f\in L^2_{per,\Lambda}$ then  $f$ can be expanded in a Fourier series with Fourier coefficients
 $\hat{f} = \{ f_\bm \}_{\bm\in\mathbb{Z}^2}$:
\begin{align}
f(\bx)\ &=\ \sum_{\bm\in\Z^2} f_\bm\ e^{i\bm\bk\cdot\bx}\ =\ \sum_{(m_1,m_2)\in\Z^2} f_{m_1,m_2}\ e^{i(m_1\bk_1+m_2\bk_2)\cdot \bx}\ \ \ \ ,\\
f_\bm\ & \equiv\ \ \frac{1}{|\cell|}\ \int_{\cell}\ e^{-i\bm\bk
\cdot \by}\ f(\by)\ d\by\ =\ \frac{1}{|\cell|}\ \int_{\cell}\ e^{-i(m_1\bk_1+m_2\bk_2)
\cdot \by}\ f(\by)\ d\by.
\label{Fourier-coeff}\end{align}

 \bigskip
 
Let $V(\bx)$ denote a real-valued potential which is periodic
 with respect to $\Lambda$, {\it i.e.}
\begin{equation} V(\bx+\bv)=V(\bx), \ {\rm for}\ \bx\in\R^2 ,\ \ \bv\in\Lambda
\nn\end{equation}
Throughout this paper we shall also assume the potential, $V(\bx)$, under consideration is $C^\infty$.
Thus,
\begin{equation}
V\in C^\infty(\mathbb{R}^2/\Lambda)\ .
\label{Vassumptions}
\end{equation}
We expect that this smoothness assumption can be relaxed considerably without much extra work.

For each $\bk\in\R^2$ we consider the  {\it Floquet-Bloch eigenvalue problem}
\begin{align}
H_V\  \phi(\bx;\bk) &= \mu(\bk)\ \phi(\bx;\bk),\ \ \bx\in\R^2\label{phi-eqn}\\
 \phi(\x+\bv;\bk) &= e^{i\bk\cdot \bv}\ \phi(\bx;\bk),\ \ \bv\in \Lambda,\   \label{pseudo-per}
 \end{align}
 where
 \begin{equation}
 H_V\ \equiv\  -\Delta + V(\bx)\ .
 \label{HV-def}
 \end{equation}
 An $L^2_\bk$- solution of \eqref{phi-eqn}-\eqref{pseudo-per} is called a {\it Floquet-Bloch} state.

 Since the eigenvalue problem \eqref{phi-eqn}-\eqref{pseudo-per} is invariant under the change $\bk\mapsto \bk+\tilde{\bk}$, where $\tilde{\bk}\in\Lambda^*$, the dual period lattice, the eigenvalues and eigenfunctions of  \eqref{phi-eqn}-\eqref{pseudo-per} can be regarded as  $\Lambda^*-$ periodic functions of $\bk$, or functions on $\mathbb{T}_\bk^2=\mathbb{R}^2_\bk/\Lambda^*$.
  Therefore, it suffices to restrict our attention to $\bk$ varying over any primitive cell.
  It is standard to work with the first Brillouin zone, $\mathcal{B}$, the closure of the set of  points $\bk\in\mathbb{R}^2$, which are closer to the origin than to any other lattice point.

 An alternative formulation is obtained as follows. For every $\bk\in\mathcal{B}$ we set 
 \begin{equation}
 \phi(\bx;\bk)=e^{i\bk\cdot\bx}p(\bx;\bk)\label{phi-p}
 \end{equation}
Then $p(\bx;\bk)$  satisfies the periodic elliptic boundary
value problem: 
\begin{align}
H_V(\bk) p(\bx;\bk)\ &=\   \mu(\bk)\ p(\bx;\bk),\ \  \bx\in\R^2\label{p-eqn}\\
  p(\bx+\bv;\bk) &= p(\bx;\bk),\ \ \bv\in \Lambda \label{p-periodic}, 
\end{align}
where
\begin{equation}
H_V(\bk)\equiv -\left(\nabla+ i\bk\right)^2 + V(\bx)\ \equiv\ -\Delta_\bk+V(\bx)\ \ \ \ .
\label{HVdef}
\end{equation}
The eigenvalue problem \eqref{phi-eqn}-\eqref{pseudo-per}, or equivalently \eqref{p-eqn}-\eqref{p-periodic},  has a discrete spectrum:
\begin{equation}
\mu_1(\bk)\ \le\ \mu_2(\bk)\ \le\ \mu_3(\bk)\ \le\ \dots
\label{eig-ordering}\end{equation}
with eigenpairs
$
p_b(\bx;\bk),\ \mu_b(\bk):\ b=1,2,3,\dots .
$
The set $\{p_b(\bx;\bk)\}_{b\ge1}$ can be taken to be a complete orthonormal set in 
$L^2_{\rm per}(\R^2/\Lambda)$.  

The functions $\mu_b(\bk)$ are called band dispersion functions. Some general results on their regularity appear in \cite{avron-simon:78}.
As $\bk$ varies over $\mathcal{B}$, $\mu_b(\bk)$ sweeps out a closed real interval. The spectrum of $-\Delta + V(\bx)$ in $L^2(\R^2)$ is the union of these closed
intervals:
\begin{equation}
  \label{L2-spectrum}
  \textrm{spec}(H_V) = \bigcup_{ \bk \in \mathcal{B}} \textrm{spec}\left(H_V(\bk)\right) \ \ \ .
\end{equation}
Moreover, the set $\bigcup_{b\ge1}\bigcup_{\bk\in\mathcal{B}}\{\phi_b(\bx;\bk)\}, \phi_b(\bx;\bk)\equiv e^{i\bk\cdot\bx}p_b(\bx;\bk)$,\  suitably normalized, is complete in $L^2(\R^2)$:
\begin{equation}
f\in L^2(\R^2)\ \implies\ f(\bx)=\sum_{b\ge1}\ \int_\brill\langle\phi_b(\cdot,\bk),f\rangle\phi_b(\bx;\bk)\ d\bk\ ,
\nn\end{equation}
where the sum converges in the $L^2$ norm.
\bigskip

\subsection{The period lattice, $\Lambda_h$ , and its dual, $\Lambda_h^*$}
\label{sec:honeycomb}
{\ \ \ \ }\bigskip

Consider $\Lambda_h=\Z{\bf v}_1 \oplus \Z{\bf v}_2$, the lattice generated by the basis vectors:
\begin{align}
 {\bf v}_1 &=\ a\left( \begin{array}{c} \frac{\sqrt{3}}{2} \\ {}\\  \frac{1}{2}\end{array} \right),\ \ 
{\bf v}_2 =\ a\left(\begin{array}{c} \frac{\sqrt{3}}{2} \\ {}\\ -\frac{1}{2} \end{array}\right),\ \ a>0.\label{v12-def}
\end{align}
Note: $\Lambda_h$ (``$h$'' for honeycomb) is a triangular lattice, that arises naturally in connection with honeycomb structures; see Figure \ref{fig:honeyAB-1}.

The dual lattice $\Lambda_h^* =\ \Z {\bf k}_1\oplus \Z{\bf k}_2$ is spanned by the dual basis vectors:
\begin{align}
&  {\bf k}_1=\ q\left(\begin{array}{c} \frac{1}{2}\\ {}\\ \frac{\sqrt{3}}{2}\end{array}\right),\ \ \ {\bf k}_2 = q\ \left(\begin{array}{c}\frac{1}{2}\\ {}\\ -\frac{\sqrt{3}}{2}\end{array}\right),\ \  \ q\equiv \frac{4\pi}{a\sqrt{3}}\ ,\label{q-def}
\end{align}
where
\begin{align}
&{\bf k}_{\ell}\cdot {\bf v}_{{\ell'}}=2\pi\delta_{\ell\ell'}\ , \label{orthog-kv}\\
&|\bv_1|=|\bv_2|=a,\ \ \bv_1\cdot\bv_2=\frac{a^2}{2}\ , \label{bv12}\\
&|\bk_1|=|\bk_2|=q,\ \ \bk_1\cdot\bk_2=-\frac{1}{2}q^2\ .\label{bk12}
\end{align}

The Brillouin zone, $\brill_h$, is a hexagon in $\R^2$; see figure \ref{fig:honeyAB-2}. Denote by $\bK$ and $\bKp$ the  vertices of  $\brill_h$ given by:
\begin{equation}
\bK\equiv\frac{1}{3}\left(\bk_1-\bk_2\right),\ \ \bKp\equiv-\bK=\frac{1}{3}\left(\bk_2-\bk_1\right)\ .
\label{KKprime}
\end{equation}
All six  vertices of $\brill_h$ can be generated by application of the rotation matrix, $R$,
 which rotates a vector in $\mathbb{R}^2$ clockwise by $2\pi/3$. $R$ is given by
\begin{equation}
R\ =\ \left(\begin{array}{cc} -\frac{1}{2} & \frac{\sqrt{3}}{2}\\
                                                 {} & {}\\
                                           -\frac{\sqrt{3}}{2} & -\frac{1}{2}
                                           \end{array}\right)
  \label{Rdef}\end{equation}
  and the vertices of $\brill_h$ fall into to groups, generated by action of $R$ on $\bK$ and $\bK'$:
\begin{align}
&\bK\ {\rm type-points:}\ \bK,\ R\bK=\bK+\bk_2,\ R^2\bK=\bK-\bk_1\nn\\
&\bKp\ {\rm type-points:}\ \bKp,\ R\bKp=\bK'-\bk_2,\ R^2\bKp=\bKp+\bk_1\ .
\label{Bvertices}
\end{align}
\begin{remark}[Symmetry Reduction]\label{symmetry-reduction}
Let $(\ \phi(\bx;\bk), \mu(\bk)\ )$ denote a Floquet-Bloch eigenpair for the eigenvalue problem \eqref{phi-eqn}-\eqref{pseudo-per} with quasi-momentum $\bk$. Since $V$ is real, 
$(\ \tilde{\phi}(\bx;\bk)\equiv\overline{\phi(\bx;\bk)}, \mu(\bk)\ )$ is a Floquet-Bloch eigenpair for the eigenvalue problem with quasi-momentum $-\bk$. Recall the relations \eqref{Bvertices}  and the $\Lambda_h^*$- periodicity of: $\bk\mapsto\mu(\bk)$ and $\bk\mapsto \phi(\bx;\bk)$.    It follows that  the local character of the dispersion surfaces in a neighborhood of any vertex of $\brill_h$ is determined by its character about any other vertex of $\brill_h$.
\end{remark}
\bigskip

In our computations using Fourier series, we shall frequently make use of the following relations:
\begin{align}
&R\ \kv_1=\bk_2,\ \ R\ \bk_2=-\left(\bk_1+\bk_2\right),\ \ R\ \left(\bk_1+\bk_2\right)=-\bk_1
\label{Rk-identities}\\
& R\ \zeta =\ \tau\ \zeta, \ \ R\ \bar\zeta\ =\ \bar\tau\ \bar\zeta,\ {\rm where}\\  
 &\zeta\equiv\frac{1}{\sqrt{2}}\left(\begin{array}{c}1\\ i \end{array}\right),\ \ \ \tau=e^{\frac{2\pi i}{3}}=-\frac{1}{2}+i\frac{\sqrt{3}}{2},\ \ \tau^3=1.
\label{Reigs} \end{align}

Moreover, $R^*$ maps the period lattice $\Lambda_h$ to itself and, in particular,
\begin{equation}
R^*\bv_1=-\bv_2,\ R^*\bv_2=\bv_1-\bv_2
\label{Rstar-vj}
\end{equation}

\subsection{Honeycomb lattice potentials}\label{sec:HLP}
{\ \ \ }

For any function $f$, defined on $\mathbb{R}^2$,  introduce
\begin{equation}
\mathcal{R}[f](\bx)=f(R^*\bx),\label{calRdef}
\end{equation}
where $R$ is the $2\times2$ rotation matrix displayed in \eqref{Rdef}.
\medskip

\begin{definition}\label{honeyV}[Honeycomb lattice potentials]
{ }

Let $V$ be  real-valued and  $V\in C^\infty(\R^2)$.
$V$ is  a  \underline{ honeycomb lattice potential} 
 if there exists $\bx_0\in\mathbb{R}^2$ such that $\tilde{V}(\bx)=V(\bx-\bx_0)$ has the following properties:
\begin{enumerate}
\item $\tilde{V}$ is $\Lambda_h-$ periodic, {\it i.e.}  $\tilde{V}(\bx+\bv)=\tilde{V}(\bx)$ for all $\bx\in\mathbb{R}^2$ and $\bv\in\Lambda_h$.  
\item $\tilde{V}$ is even or inversion-symmetric, {\it i.e.} $\tilde{V}(-\bx)=\tilde{V}(\bx)$.
\item  $\tilde{V}$  is $\mathcal{R}$- invariant, {\it i.e.}
 \begin{equation}
 \mathcal{R}[\tilde{V}](\bx)\ \equiv\ \tilde{V}(R^*\bx)\ =\ \tilde{V}(\bx),
 \nn\end{equation}
  where, $R^*$ is the counter-clockwise rotation matrix by $2\pi/3$, {\it i.e.} $R^*=R^{-1}$, where $R$
 is given by \eqref{Rdef}. 
 
 \end{enumerate}\smallskip
 
 Thus, a honeycomb lattice potential is smooth, $\Lambda_h$- periodic and, with respect to some origin of coordinates, both inversion symmetric and $\mathcal{R}$- invariant.
 \end{definition}\medskip
 
 \begin{remark}\label{takex0eq0}
 As the spectral properties are independent of translation of the potential we shall assume in the proofs, without any loss of generality,  that $\bx_0=0$.
 \end{remark}\medskip
 
 \begin{remark}\label{Vreal-even}
 A consequence of a honeycomb lattice potential being real-valued and even is that if 
 $\left(\phi(\bx;\bk),\mu\right)$ is an eigenpair with quasimomentum $\bk$ of the Floquet-Bloch eigenvalue problem, then $\left(\overline{\phi(-\bx;\bk)},\mu\right)$ is also an eigenpair with quasimomentum $\bk$.
 \end{remark} 
 \medskip

\begin{remark}\label{honey-remark} 
We present two constructions of honeycomb lattice potentials.\bigskip

\begin{itemize}
\item[Example 1:] {\bf ``Atomic'' honeycomb lattice potentials:}\ Start with the two points
\begin{equation} {\bf A}=(0,0),\ \ {\rm and}\ \  {\bf B}=a\left(\frac{1}{\sqrt{3}},0\right)\ ,
\label{ABdef}
\end{equation}
which lie within the unit period cell of $\Lambda_h$; see \eqref{v12-def}.
 Define the triangular lattices of {\bf A-} type and {\bf B-} type points:
\begin{equation}
\Lambda_{\bf A}\ =\ {\bf A}\ +\ \Lambda_h\ \ {\rm and} \ \ \Lambda_{\bf B}\ =\ {\bf B}\ +\ \Lambda_h
\label{LambdaAB}
\end{equation}
We define the honeycomb structure, ${\bf H}$, to be the union of these two triangular lattices:
\begin{equation}
{\bf H}\ =\ \Lambda_{\bf A}\ \cup\  \Lambda_{\bf B}\ ;
\label{Hdef}
\end{equation}
see Figure \ref{fig:honeyAB-1}.
Let $V_0$ be a smooth, radial and rapidly decreasing function, which we think of as an ``atomic potential''.  Then,  
\[V(\bx) = \sum_{{\bf a}\in{\bf H}} V_0(\bx-{\bf a}),\]
is a potential associated with ``atoms'' at each site of the honeycomb structure ${\bf H}$. 
Moreover, $V(\bx)$ 
 is a honeycomb lattice potential in the sense of Definition \ref{honeyV} with $\bx_0= -{\bf B}$.\smallskip
 
 Note that a potential of the form  
 \begin{equation}
 V(\bx)\ =\ \sum_{{\bf a}\in\Lambda_h} V_0(\bx-{\bf a}),
 \nn\end{equation}
 a {\it ``triangular lattice potential''}, also satisfies the properties listed in Definition \ref{honeyV}. 

\item[Example 2:] {\bf Optical honeycomb lattice potentials:}\  The electric field envelope of a nearly monochromatic beam of light propagating through a dielectric medium with two-dimensional refractive index profile satisfies a linear Schr\"odinger equation $i\partial_z\psi=-\Delta_{x,y}\psi + V(x,y)\psi=0$.  Here, $z$ denotes the direction of propagation of the beam
 and $(x,y)$ the transverse directions. 
  Honeycomb lattice potentials have been generated taking advantage of nonlinear  optical phenomena. It was demonstrated in \cite{Segev-etal:07} that a honeycomb lattice potential (a honeycomb ``photonic lattice''), $V(x,y)$,  can be generated through  an optical induction technique based on the  interference of three plane wave beams of light within a photorefractive crystal, exhibiting the defocusing (nonlinear) optical Kerr effect. The refractive-index variations are governed by a  potential of the approximate form:
 \begin{equation}
 V(\bx) \approx  V_0\left(\ \cos(\bk_1\cdot\bx) + \cos(\bk_2\cdot\bx) + \cos\left( (\bk_1+\bk_2)\cdot\bx\right)\right)
 \end{equation}
 It is straightforward to check, in view of \eqref{Rk-identities}, that a potential of this type is a honeycomb lattice potential in the sense of Definition \ref{honeyV} with $\bx_0=0$. In fact, in Proposition 
 \ref{proposition:V-honey}  below we assert that with respect to some origin of coordinates, any honeycomb lattice potential can be expressed as a Fourier series of terms of this type.
\end{itemize}
\end{remark}

 \begin{figure}[ht!]
\centering \includegraphics[width=4.75in]{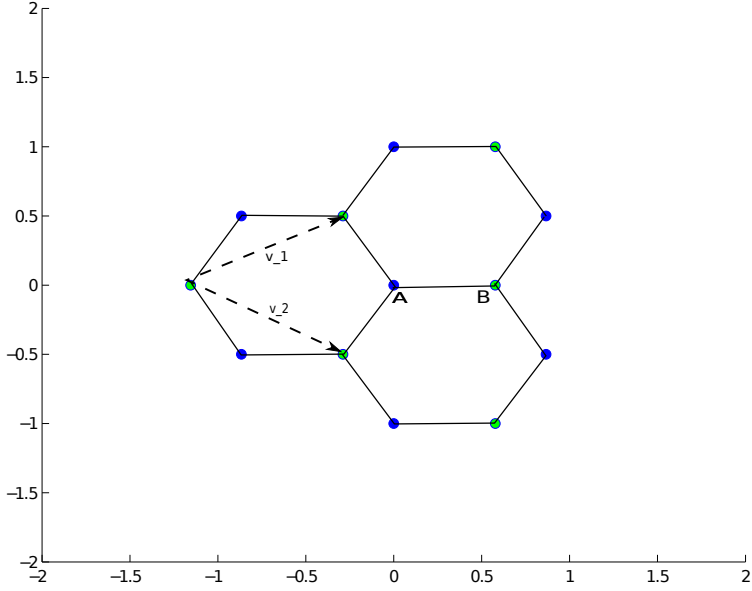}
  \caption{ Part of the honeycomb structure, ${\bf H}$.
  ${\bf H}$ is the union of two sub-lattices  $\Lambda_{\bf A}={\bf A}+\Lambda_h$ (blue) 
   and $\Lambda_{\bf B}={\bf B}+\Lambda_h$ (green). The lattice vectors 
    $\{\bv_1,\bv_2\}$ generate $\Lambda_h$. See  Remark \ref{honey-remark}. 
   } .  \label{fig:honeyAB-1}
\end{figure}

 \begin{figure}[ht!]
\centering 
\includegraphics[width=4.75in]{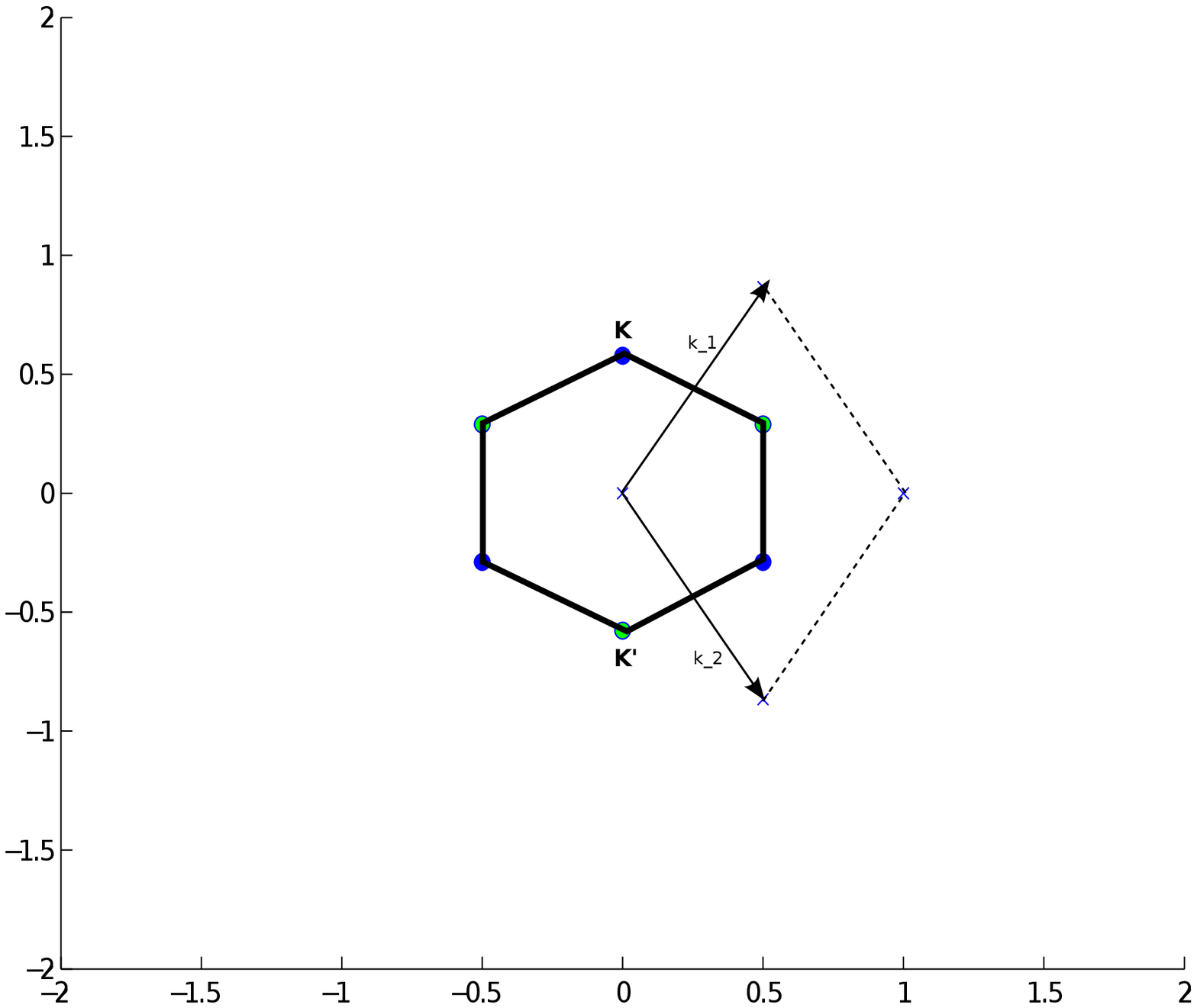}
  \caption{  Brillouin zone, $\brill_h$, and dual basis $\{\bk_1,\bk_2\}$. $\bK$ and $\bK'$ are labeled. Other vertices of $\brill_h$ obtained via application of $R$, rotation by $2\pi/3$; see equation \eqref{Bvertices}.} .  \label{fig:honeyAB-2}
\end{figure}
 
\bigskip

The following proposition plays a key role. It states that at distinguished points 
in $\bk-$ space, namely the $\bK$ and $\bKp$ type points, $H_V$ with quasi-momentum dependent boundary conditions \eqref{pseudo-per} or equivalently, $H_V(\bk)$, with $\Lambda_h$ periodic boundary conditions,  has an extra rotational invariance property.\medskip

\begin{proposition}\label{SR-invariance}
Assume $V$ is a honeycomb lattice potential, as in Definition \ref{honeyV}. 
Assume  $\bK_\star$ is a point of $\bK$ or $\bK'$ type; see \eqref{Bvertices}. Then, $H_V$ and $\mathcal{R}$ map a dense subspace of $L^2_{\bK_\star}$ to itself. Furthermore, restricted to this dense subspace of
$L^2_{\bK_\star}$, 
 the commutator $\left[H,\mathcal{R}\right]\equiv H_V\ \mathcal{R}- \mathcal{R}\ H_V$ vanishes. In particular, 
  if $\phi(\bx;\bk)$  is a solution of the Floquet-Bloch eigenvalue problem \eqref{phi-eqn}-\eqref{pseudo-per} 
with $\bk=\bK_\star$, then   $\mathcal{R}[\phi(\cdot,\bk)](\bx)$ is also a solution 
of \eqref{phi-eqn}-\eqref{pseudo-per} with $\bk=\bK_\star$.
\end{proposition}\bigskip

\begin{proof} Take as a dense subspace $C^\infty_{\bK_\star}$, the space of $C^\infty$  functions satisfying
 $f(\bx+\bv)=e^{i\bK_\star\cdot\bv}f(\bx)$ for all $\bx\in\R^2$ and $\bv\in\Lambda_h$. Clearly, $H$ maps $C^\infty_{\bK_\star}$  to itself.
 Define $\phi_R(\bx)=\mathcal{R}[\phi(\cdot,\bK_\star)](\bx)= \phi(R^*\bx,\bK_\star)$. Without loss of generality, assume $\bK_\star=\bK$. By \eqref{Rstar-vj}, if $\bv\in\Lambda_h$ then $R^*\bv\in\Lambda_h$. We have 
\begin{align*}
\phi_R(\bx+\bv)&=\phi(R^*\bx+R^*\bv,\bK)\ =\ e^{i\bK\cdot R^*\bv}\phi(R^*\bx,\bK)\nn\\
&=e^{iR\bK\cdot\bv}\ \phi(R^*\bx,\bK)\ =\ e^{i(\bK+\bk_2)\cdot\bv}\ \phi(R^*\bx,\bK)\nn\\
& =\  e^{i\bK\cdot\bv}\ \phi(R^*\bx,\bK)\ =\  e^{i\bK\cdot\bv}\ \phi_R(\bx). 
\end{align*}
Thus, we have $\mathcal{R}$ maps $C^\infty_{\bK_\star}$ to itself.

Next note that by invariance of the Laplacian under rotations, $-\Delta_\bx\phi_R(\bx)=-\left.\Delta_\by\phi(\by,\bK_\star)\right|_{\by=R^*\bx}$\ . Furthermore, by $\mathcal{R}-$ invariance of $V(\bx)$, that $V(\bx)\phi_R(\bx)=V(R^*\bx)\phi(R^*\bx,\bK_\star)=\left.V(\by)\phi(\by,\bK_\star)\right|_{\by=R^*\bx}$\ . Therefore,  
$\left[H,\mathcal{R}\right]$ vanishes on  on $C^\infty_{\bK_\star}$. In particular, we have that
\begin{equation}
H_V\phi(\bx,\bK_\star)=\mu\phi(\bx,\bK_\star)\ \implies\ H_V\phi_R(\bx)=\mu\phi_R(\bx)\ .
\nn\end{equation}
 This completes the proof of the proposition. \end{proof}
 \bigskip
 
  We conclude this section with a discussion of the  Fourier representation of honeycomb lattice potentials in the sense of Definition \ref{honeyV}. Let $V(\bx)$ be such a potential with Fourier series:
 \[ V(\bx) = \sum_{\bm\in\mathbb{Z}^2}V_\bm e^{i\bm\bk\cdot\bx}=\sum_{(m_1,m_2)\in\mathbb{Z}^2}
 V_{m_1,m_2} e^{i (m_1\bk_1+m_2\bk_2)\cdot\bx}.\]
Since $V(\bx)=\mathcal{R}[V](\bx)$, we have 
\[ V(R^*\bx) = \sum_{(m_1,m_2)}
 V_{m_1,m_2} e^{i (m_1R\bk_1+m_2R\bk_2)\cdot\bx} =
  \sum_{(m_1,m_2)}
 V_{m_1,m_2} e^{i ( (-m_2)\bk_1+(m_1-m_2)\bk_2)\cdot\bx} \]
Therefore, $V_{m_1,m_2}=V_{-m_2,m_1-m_2}$.
Similarly, $V(\bx)=\mathcal{R}^2[V](\bx)$ implies that $V_{m_1,m_2}=V_{m_2-m_1,-m_1}$.
Introduce the mapping $\tilde{R}:\mathbb{Z}^2\to\mathbb{Z}^2$ acting on the indices of the Fourier coefficients of 
 $V$:
\begin{align} \tilde{R}(m_1,m_2)&=(-m_2,m_1-m_2)\ \ {\rm and\ therefore}\nn\\
 \ \tilde{R}^2(m_1,m_2)&=(m_2-m_1,-m_1),\ \ {\rm and}\ \ \tilde{R}^3(m_1,m_2)= (m_1,m_2)\ .
 \label{tRdef}\end{align}
Then we have
\begin{equation}
V_\bm\ =\ V_{\tilde{R}\bm}\ =\ V_{\tilde{R}^2\bm}\label{tildeRbm}
\end{equation}
Note that $\tilde{R}{\bf 0}={\bf 0}$ and that ${\bf 0}$ is the unique element of the kernel of $\tilde{R}$. Furthermore, any $\bm\ne0$ lies on an $\tilde{R}-$ orbit of length exactly three. Indeed, 
\begin{align}
\bm&=\tilde{R}\bm\ \leftrightarrow\ (m_1,m_2)=(-m_2,m_1-m_2)\ \implies\ m_1=m_2=0\ \ {\rm and}\nn\\
\bm&=\tilde{R}^2\bm\ \leftrightarrow\ (m_1,m_2)=(-m_1+m_2,-m_1)\ \implies m_1=m_2=0\ .\nn
\end{align}

Suppose $\bm$ and $\bn$ are non-zero. We say that $\bm\sim\bn$ if $\bm$ and $\bn$ lie on the same $3-$ cycle.
The relation $\sim$ is an equivalence relation, which partitions $\mathbb{Z}^2\setminus\{{\bf 0}\}$ into equivalence classes, $ \left(\mathbb{Z}^2\setminus\{{\bf 0}\}\right) / \sim$. Let $\tilde{S}$ denote a set consisting of exactly one representative from each equivalence class. We now have the following characterization of Fourier series of honeycomb lattice potentials:\smallskip

\begin{proposition}\label{proposition:V-honey} 
Let $V(\bx)$ denote a honeycomb lattice potential. Then, 
\begin{equation}
V(\bx) = \hat{V}({\bf 0})+\sum_{\bm\in\tilde{\mathcal{S}}}\ V_\bm\ \left[ \cos(\bm\bk\cdot\bx)\ +\ \cos((\tilde{\mathcal{R}}\bm)\bk\cdot\bx)\ +\ \cos((\tilde{\mathcal{R}^2}\bm)\bk\cdot\bx)\ \right]\ ,
\label{honey-fourier}
\end{equation}
where $V_\bm$ are real and $\tilde{R}$ is defined in \eqref{tRdef}.
\end{proposition}
\bigskip

\noindent{\it Proof of Proposition \ref{proposition:V-honey}:} From \eqref{tildeRbm} we have
\begin{equation}
V(\bx) = \hat{V}({\bf 0})+\sum_{\bm\in\tilde{S}}V_\bm\left( e^{i\bm\bk\cdot\bx}+
e^{i(\tilde{R}\bm)\bk\cdot\bx}+e^{i(\tilde{R}^2\bm)\bk\cdot\bx}\right)
\label{fs}\end{equation}
The relation $V(\bx)=(V(\bx)+V(-\bx))/2$ and  \eqref{fs} imply
\begin{equation}
V(\bx) = \hat{V}({\bf 0})+\sum_{\bm\in\tilde{S}}V_\bm\left( \cos(\bm\bk\cdot\bx)\ +\ 
\cos\left((\tilde{R}\bm)\bk\cdot\bx\right)\ +\ \cos\left((\tilde{R}^2\bm)\bk\cdot\bx\right)\ \right)
\label{cosfs}\end{equation}
Moreover, since $V$ is real and even, $V_{\bm}$ is real for $\bm\in\mathbb{Z}^2$. This completes the proof.

\subsection{Fourier analysis  in $L^2_{\bK_\star}$}\label{sec:Fourier-on-H}
{\ \ \ }

We characterize the Fourier series of functions $\phi\in L^2_{\bK_\star}$, {\it i.e.} functions $\phi(\bx;\bK_\star)$, satisfying 
the quasiperiodic boundary condition:
\begin{equation}
\phi(\bx+\bv)\ =\ e^{i\bK_\star\cdot\bv}\ \phi(\bx)
\label{phi-quasi}
\end{equation}
The discussion is analogous to that preceding Proposition \ref{proposition:V-honey}.\medskip

If \eqref{phi-quasi} holds then  $\phi(\bx) = e^{i\bK_\star\cdot\bx}\ p(\bx)$, where $p(\bx)$ is $\Lambda-$ periodic.  It follows that $\phi$ has a Fourier representation:
 \begin{equation}
 \phi(\bx)\ =\ e^{i\bK_\star\cdot\bx}\sum_{(m_1,m_2)\in\Z^2}\ c(m_1,m_2)\ e^{i(m_1\bk_1+m_2\bk_2)\cdot\bx}, 
 \label{phi-expand}
 \end{equation} 
 which we rewrite  as 
 \begin{align}
 \phi(\bx)\ &=\ \sum_{(m_1,m_2)\in\Z^2}\ c(m_1,m_2)\ 
 e^{i(\bK_\star + m_1\bk_1+m_2\bk_2)\cdot\bx}\nn\\
  &=\ \sum_{(m_1,m_2)\in\Z^2}\ c(m_1,m_2)\
  e^{i \bK^{m_1,m_2}_\star \cdot\bx}\ =\ \sum_{\bm\in\Z^2}\ c(\bm)\
  e^{i \bK^{\bm}_\star \cdot\bx},\
 \label{phi-expand1}\\
 &\ \ {\rm where}\ \ \bK^{\bm}_\star\ =\ \bK_\star+m_1\bk_1+m_2\bk_2.\nn
 \end{align} 
 {\it Usually, we denote by $c_\phi(\bm)$ or $c(\bm;\phi)$ the Fourier coefficients,
  as in \eqref{phi-expand1}, of $\phi\in L^2_{\bK_\star}$.}
 
 Note that the transformation $\mathcal{R}$, defined in \eqref{calRdef}, is unitary on $L^2$ and so its eigenvalues lie on the unit circle in $\mathbb{C}$. Furthermore, if $\mathcal{R}\phi=\lambda\phi$ and  $\phi\ne0$, then since $\mathcal{R}^3=Id$, $\phi= \mathcal{R}^3\phi =\lambda^3\phi$, we have $\lambda^3=1$. Therefore $\lambda\in\{1,\tau,\bar\tau\}$, where $\tau=\exp(2\pi i/3)$.
 
We are interested in the general Fourier expansion of functions in 
  each of the eigenspaces of $\mathcal{R}$: 
  \begin{align}
  L^2_{\bK_\star,1} &\equiv \{f\in L^2_{\bK_\star}: \mathcal{R}f=f\}\label{L2K1}\\
   L^2_{\bK_\star,\tau} &\equiv \{f\in L^2_{\bK_\star}: \mathcal{R}f=\tau f\}\label{L2Ktau1}, \\
    L^2_{\bK_\star,\bar{\tau}} &\equiv \{f\in L^2_{\bK_\star}: \mathcal{R}f=\bar\tau f\}\label{L2Ktaubar1}
    \end{align}
   Since $\mathcal{R}$ is unitary these subspaces are pairwise orthogonal.
    \medskip
 
 Fix, without loss of generality, $\bK_\star=\bK$. We first consider the action of $\mathcal{R}$ on 
 general $\phi\in L^2_\bK$. 
   Applying $\mathcal{R}$ to  $\phi$, given by \eqref{phi-expand1}, we obtain:
\begin{align*}
\mathcal{R}[\phi](\bx) &= \sum_{(m_1,m_2)\in\Z^2}\ c_\phi(m_1,m_2)\ 
 e^{i\bK^\bm\cdot R^*\bx}\\
 &= \sum_{(m_1,m_2)\in\Z^2}\ c_\phi(m_1,m_2)\ 
 e^{iR\bK^\bm\cdot\bx}\\
 &=\sum_{(m_1,m_2)\in\Z^2}\ c_\phi(m_1,m_2)\ 
 e^{i(\bK +  (-m_2)\bk_1+(m_1-m_2+1)\bk_2)\cdot\bx},
 \end{align*}
 since 
 \begin{equation}
R\bK^\bm = R\bK^{m_1,m_2}=\bK +  (-m_2)\bk_1+(m_1-m_2+1)\bk_2=\bK^{-m_2,m_1-m_2+1} .
 \label{RbKm}
 \end{equation}
 Thus, 
 \begin{align}
 c_{\mathcal{R}\phi}(-m_2,m_1-m_2+1)&=c_\phi(m_1,m_2),\ \ \textrm{or equivalently}\nn\\
 c_{\mathcal{R}\phi}(m_1,m_2) &= c_\phi(m_2-m_1-1,-m_1)\ .\label{cSRphi}
 \end{align}
 Similarly, by a second application of $\mathcal{R}$, and using the relation
 \begin{equation}
 R^2\bK^{m_1,m_2}=\bK^{m_2-m_1-1,-m_1},
 \label{R2bKm}\end{equation}
 we have
 \begin{align}
 c_{\cR^2\phi}(m_2-m_1-1,-m_1) &=  c_\phi(m_1,m_2),\ \ \textrm{or equivalently}\nn\\
c_{\cR^2\phi}(m_1,m_2) &= c_\phi(-m_2,m_1-m_2+1)\label{cSR2phi}\end{align}
Finally, since $R^3=I$, $c_{\cR^3\phi}(m_1,m_2)=c_\phi(m_1,m_2)$.
  
  $\mathcal{R}$ acting in $L^2_{\bK}$ induces a decomposition of $\Z^2$ into orbits of length three:
  \begin{equation}
  (m_1,m_2)\ ^{\mathcal{R}}\mapsto\ (-m_2,m_1-m_2+1)\ ^{\mathcal{R}}\mapsto\ (m_2-m_1-1,-m_1)\ ^{\mathcal{R}}\mapsto\ (m_1,m_2)
  \label{3cycle}
  \end{equation}
   
 \nit {\bf For convenience we shall abuse notation and write}
 \begin{align}
\cR\bm &= \cR(m_1,m_2) = (-m_2,m_1-m_2+1)\nn\\
\cR^2\bm &= \cR^2(m_1,m_2) = (m_2-m_1-1,-m_1),\nn\\
\cR^3\bm &={\rm Id}\ (m_1,m_2) = (m_1,m_2)
  \label{Rbm}\end{align}
  
Using the notation \eqref{Rbm},  relations \eqref{cSRphi}, \eqref{cSR2phi} and \eqref{Rbm} can be expressed as:
\begin{align}
c_{\mathcal{R}\phi}(\bm)\ &=\ c_\phi(\mathcal{R}^2\bm) = c_\phi(m_2-m_1-1,-m_1)\nn\\
c_{\cR^2\phi}(\bm) &= c_\phi(\mathcal{R}\bm) = c_\phi(-m_2,m_1-m_2+1)
\label{cSRcSR2phi}
\end{align}

Furthermore, by \eqref{RbKm} and \eqref{R2bKm} 
\begin{equation}
R\bK^\bm = \bK^{\mathcal{R}\bm}\ \ {\rm and}\ \ R^2\bK^\bm = \bK^{\mathcal{R}^2\bm} .
\label{RbKmR2bKm}
\end{equation}

Each point in $\mathbb{Z}^2$ lies on an orbit of $\mathcal{R}$ of precisely length $3$, a 3-cycle. To see this, note that by \eqref{Rbm}
  $\mathcal{R}^3\bm=\bm$ for all $\bm\in\mathbb{Z}^2$. So we need only check that there are no solutions to either $\mathcal{R}\bm=\bm$ or to $\mathcal{R}^2\bm=\bm$. First, suppose $\mathcal{R}\bm=\bm$. Then, $\mathcal{R}^2\bm=\bm$ as well. So, on the one hand the centroid of  $\bm, \mathcal{R}\bm$ and $\mathcal{R}^2\bm$ is equal to $\bm\in\mathbb{Z}^2$. On the other hand, by \eqref{Rbm}
  their centroid is $(-1/3,1/3)\notin\mathbb{Z}^2$, a contradiction.  Therefore, there are no $\mathbb{Z}^2$ solutions of $\mathcal{R}\bm=\bm$. Now if $\mathcal{R}^2\bm=\bm$, then applying $\mathcal{R}$
   to this relation yields $\bm=\mathcal{R}^3\bm=\mathcal{R}\bm$, and we're back to the previous case.
  
We shall say that two points in $\mathbb{Z}^2$, $\bm$ and $\bn$ are equivalent, $\bm\approx \bn$, if 
they lie on the same 3-cycle of $\mathcal{R}$. We identify all equivalent points by introducing the set of equivalence classes, $\mathbb{Z}^2/\approx$\ .
\begin{definition}\label{Sdef} We denote by $\mathcal{S}$ a set consisting of exactly one representative of each equivalence class in $\mathbb{Z}^2/\approx$\ . For example, $\{(0,0),(0,1),(-1,0)\}\in \mathbb{Z}^2/\approx$,  from which we choose $(0,1)$ as its representative in $\mathcal{S}$.
\end{definition}

Using the relations \eqref{Rbm}, we can express the Fourier series of an arbitrary $\phi\in L^2_\bK$ as
   a sum over 3-cycles of $\mathcal{R}$:
  \begin{align}
\phi(\bx)\ =&\   \sum_{(m_1,m_2)\in\Z^2}\ c_\phi(m_1,m_2)\ 
 e^{i(\bK + m_1\bk_1+m_2\bk_2)\cdot\bx} \nn\\
=& \sum_{\bm\in\mathcal{S}}\ \left( c_\phi(\bm)\ 
 e^{i\bK^\bm\cdot\bx}\ 
 +\ c_\phi(\cR\bm) e^{iR\bK^\bm\cdot\bx}  +\ c_\phi(\cR^2\bm) e^{i R^2\bK^\bm\cdot\bx}\ \right),
 \label{Fourier-S}\end{align}
 where $\mathcal{R}^j\bm,\ j=1,2$ is given in \eqref{Rbm}.

  We now turn to the Fourier representation of elements of the subspaces $L^2_{\bK,1},\ L^2_{\bK,\tau}$
  and $L^2_{\bK,\bar\tau}$. 
 
\begin{proposition}\label{coeff-eigs}
Let $\phi\in L^2_{\bK}$. 
\begin{align}
\mathcal{R}\phi\ &=\  \phi\ \Leftrightarrow\ \ c_\phi(\bm) = c_\phi(\mathcal{R}\bm) = c_\phi(\mathcal{R}^2\bm) 
\label{c-eig-1}\\
\mathcal{R}\phi\ &=\ \tau\ \phi\ \Leftrightarrow\ \   c_\phi(\mathcal{R}^2\bm) = \tau c_\phi(\bm)\ {\rm and}\ 
c_\phi(\mathcal{R}\bm) = \bar\tau c_\phi(\bm) \label{c-eig-tau}\\
\mathcal{R}\phi\ &=\ \bar\tau\ \phi\ \Leftrightarrow\ \   c_\phi(\mathcal{R}^2\bm) = \bar\tau c_\phi(\bm)\ {\rm and}\ 
c_\phi(\mathcal{R}\bm) = \tau c_\phi(\bm) .\label{c-eig-taubar}
\end{align}
Moreover,
\begin{align}
\mathcal{R}^2\phi\ &=\ \bar{\tau}\ \phi\ \Leftrightarrow\ \ c_{\mathcal{R}^2\phi}(\bm)=\ c_\phi(\mathcal{R}\bm) = \bar\tau c_\phi(\bm),\label{c-eig-taubar2}
\end{align}
where $\mathcal{R}\bm$ and $\mathcal{R}^2\bm$ are defined in \eqref{Rbm}.
\end{proposition}
\begin{proof} Assume $\mathcal{R}\phi=\sigma\phi$. Then,   $c_{\mathcal{R}\phi}(\bm)=\sigma c_{\phi}(\bm)$.
By  \eqref{cSRcSR2phi} $c_{\mathcal{R}\phi}(\bm)=c_\phi(\mathcal{R}^2\bm)$ and therefore
\begin{equation}
c_\phi(\mathcal{R}^2\bm)\ =\ \sigma\ c_\phi(\bm)
\label{c_phi1}\end{equation}
Furthermore, $\mathcal{R}^2\phi=\sigma^2\phi$ and therefore  $c_{\mathcal{R}^2\phi}(\bm)=\sigma^2 c_{\phi}(\bm)$. By  \eqref{cSRcSR2phi} $c_{\mathcal{R}^2\phi}(\bm)=c_\phi(\mathcal{R}\bm)$ and therefore
\begin{equation}
c_\phi(\mathcal{R}\bm)\ =\ \sigma^2\ c_\phi(\bm)
\label{c_phi2}\end{equation}
We next apply relations \eqref{c_phi1} and \eqref{c_phi2} to the cases: $\sigma=1,\tau,\bar\tau$.
 Let $\sigma=\tau$. Then, $\mathcal{R}\phi=\tau\phi$ implies 
 $c_\phi(\mathcal{R}^2\bm) = \tau\ c_\phi(\bm)$, by \eqref{c_phi1}. Also, by \eqref{c_phi2} we have
  $c_\phi(\mathcal{R}\bm) = \tau^2\ c_\phi(\bm) = \bar\tau\ c_\phi(\bm)$. This proves \eqref{c-eig-tau}. The 
  cases $\sigma=1,\bar\tau$ are similar. \end{proof}
  \medskip
  
 Proposition \ref{coeff-eigs} can now be used to find a representation of the eigenspaces of $\mathcal{R}$. We state the result for an arbitrary point, $\bK_\star$, of $\bK$ or $\bK'$ type.\bigskip\bigskip
 
 \begin{proposition}\label{Fourier-espaces}
 \begin{enumerate}
 \item $\phi\in L^2_{\bK_\star,\tau} \Leftrightarrow $  there exists $\{c(\bm)\}_{\bm\in\mathcal{S}}\in l^2(\mathcal{S})$ such that
 \begin{equation}
 \phi(\bx)\ =\  \sum_{\bm\in\mathcal{S}}\  c(\bm)\ \left(
 e^{i\bK_\star^\bm\cdot\bx}\ 
 +\ \bar\tau e^{iR\bK_\star^\bm\cdot\bx}  +\ \tau e^{i R^2\bK_\star^\bm\cdot\bx}\ \right).
 \label{tau-series}\end{equation}
 \item $\phi\in L^2_{\bK_\star,\bar\tau} \Leftrightarrow $  there exists $\{c(\bm)\}_{\bm\in\mathcal{S}}\in l^2(\mathcal{S})$ such that
 \begin{equation}
 \phi(\bx)\ =\  \sum_{\bm\in\mathcal{S}}\  c(\bm)\ \left(
 e^{i\bK_\star^\bm\cdot\bx}\ 
 +\ \tau e^{iR\bK_\star^\bm\cdot\bx}  +\ \bar\tau e^{i R^2\bK_\star^\bm\cdot\bx}\ \right).
\label{taubar-series} \end{equation}
\item If $\phi_1\in L^2_{\bK_\star,\tau}$ is given by
\begin{align}
 \phi_1(\bx,\bK_\star)\ &=\ \sum_{\bm\in{\cal S}}\ c(\bm)\ 
\left(\ e^{i\bK_\star^m\cdot\bx }\ +\ \bar{\tau}\ e^{iR\bK_\star^m\cdot\bx }\ +\  
\tau\ e^{iR^2\bK_\star^m\cdot\bx } \right), \
\label{Phi1-Fourier}\end{align} 
then  $ \phi_2(\bx,\bK_\star)\ \equiv \overline{\phi_1(-\bx,\bK_\star)}\in L^2_{\bK_\star,\overline{\tau}}$ and 
 \begin{align}
 \phi_2(\bx,\bK_\star)\ & =\  
 \sum_{\bm\in{\cal S}}\ \overline{c(\bm)}\ 
\left(\ e^{i\bK_\star^m\cdot\bx }\ +\ \tau\ e^{iR\bK_\star^m\cdot\bx }\ +\  
\bar{\tau}\ e^{i\cR^2\bK_\star^m\cdot\bx } \right)\ . \label{Phi2-Fourier}
\end{align}
 \item  $\phi\in L^2_{\bK_\star,1} \Leftrightarrow $  there exists $\{c(\bm)\}_{\bm\in\mathcal{S}}\in l^2(\mathcal{S})$ such that
 \begin{equation}
 \phi(\bx)\ =\  \sum_{\bm\in\mathcal{S}}\  c(\bm)\ \left(
 e^{i\bK_\star^\bm\cdot\bx}\ 
 +\  e^{iR\bK_\star^\bm\cdot\bx}  +\ e^{i R^2\bK_\star^\bm\cdot\bx}\ \right).
 \end{equation}
 \end{enumerate}
 \end{proposition}

  We summarize the preceding in a result which facilitates the study of $H_V$ on $L^2_{\bK_\star}$ in terms of the action of $\mathcal{R}$ on  invariant subspaces of $H_V$.\medskip

\begin{proposition}\label{L2decomp} Let $\bK_\star$ denote a point of $\bK$ or $\bK'$ type,
 $R$ denote the $2\pi/3$ clockwise rotation matrix (see \eqref{Rdef}) and $\mathcal{R}[f](\bx)=f(R^*\bx)$. Then 
 $\mathcal{R}$, acting on $L^2_{\bK_\star}$ has eigenvalues $1,\ \tau$ and $\bar{\tau}=\tau^2$ inducing  a corresponding orthogonal sum decomposition of $L^2_{\bK_\star}$ into  eigenspaces:
\begin{equation}
L^2_{\bK_\star}=L^2_{\bK_\star,1}\oplus L^2_{\bK_\star,\tau}\oplus L^2_{\bK_\star,\bar{\tau}}\ .
\label{L2-decomp}
\end{equation} 
The elements of each summand are represented as in Proposition \ref{Fourier-espaces}. 
\end{proposition}

\begin{remark}  Since, by Proposition \ref{SR-invariance}, $H_V$ and $\mathcal{R}$ commute on $L^2_{\bK_\star}$,  the spectral theory of 
  $H_V$ in $L^2_{\bK_\star}$ can  be reduced to its independent study in each of the 
eigenspaces in the orthogonal sum \eqref{L2-decomp}.  
\end{remark}

\bigskip\bigskip

\section{Spectral properties of $H^{(0)}$ in $L^2_{\bK_\star}$\ -\ Degeneracy at $\bK$ and $\bK'$ points }\label{Veq0}

Our starting point for the study of $H^{(\eps)}$ on $L^2_{\bK_\star}$ is 
the study of $H^{(0)}=-\Delta$. Consider the eigenvalue problem
\begin{equation}
H^{(0)} \phi^{(0)}\ =\ \mu^{(0)}(\bk)\phi^{(0)},\ \ \ \phi^{(0)}\in L^2_\bk .
\nn\end{equation}
Equivalently  $\phi^{(0)}(\bx;\bk)=e^{i\bk\cdot\bx} p^{(0)}(\bx;\bk)$, where $p^{(0)}(\cdot;\bk)\in L^2(\mathbb{R}^2/\Lambda_h)$:
 \begin{align}
 & H^{(0)}(\bk) p^{(0)}\ \equiv\ -(\nabla+i\bk)^2 p^{(0)}\ =\ \mu^{(0)}(\bk) p^{(0)} ,\label{honey0-eig}\\
 &p^{(0)}(\bx+\bv;\bk) \ = p^{(0)}(\bx;\bk),\ \ \bv\in\Lambda_h\ .
\label{p-Lam-periodic}\end{align}

The eigenvalue problem \eqref{honey0-eig}, \eqref{p-Lam-periodic} has solutions of the form:
\begin{equation}
p^{(0)}_{m_1,m_2}(\bx;\bk)\ =\ e^{i\left(m_1\bk_1+m_2\bk_2\right)\cdot\bx}
\nn\end{equation}
with associated eigenvalues
\begin{equation}
\mu^{(0)}_{m_1,m_2}(\bk)\ =\ \left| \bk + m_1\bk_1+m_2\bk_2 \right|^2,\ \ \bk\in\mathcal{B}\ .
\label{H0k-eigs}
\end{equation}
\medskip

\begin{proposition}\label{H0-spec}
Let $\bk=\bK_\star$ denote any vertex of the hexagon $\mathcal{B}_h$ (points of $\bK$ or $\bKp$ type); see \eqref{Bvertices}.
 Then, 
 \begin{enumerate}
 \item $\mu^{(0)}=|\bK_\star|^2$ is an eigenvalue of $H_0$ of multiplicity three with corresponding three-dimensional eigenspace
\begin{equation}
 {\rm span}\ \{\ e^{i\bK\star\cdot\bx}\ ,\  e^{iR\bK\star\cdot\x}\ ,\  e^{iR^2\bK\star\cdot\bx}\ \}\ .
 \label{3dL2per1}
 \end{equation}
 \item Restricted to each of the $\mathcal{R}-$ invariant subspaces of  
 \begin{equation}
L^2_{\bK_\star}\ \equiv\ L^2_{\bK_\star,1}\ \oplus\ L^2_{\bK_\star,\tau} \oplus\ \ L^2_{\bK_\star,\bar\tau}\ \ ,
\nn\end{equation}
$H^{(0)}$ has an eigenvalue $\mu^{(0)}=|\bK_\star|^2$ of multiplicity one  with eigenspaces:
\begin{align}
&{\rm span}\ \{ e^{i\bK_\star\cdot\bx}+e^{iR\bK_\star\cdot\x}+e^{iR^2\bK_\star\cdot\bx}\}\subset L^2_{\bK_\star,1}\nn\\
& {\rm span}\ \{ e^{i\bK_\star\cdot\x}+\bar\tau e^{iR\bK_\star\cdot\x} + 
\tau e^{iR^2\bK_\star\cdot\bx}\}\subset L^2_{\bK_\star,\tau}\ \  {\rm and}\nn\\
&  {\rm span}\ \{ e^{i\bK_\star\cdot\x}+ \tau e^{iR\bK_\star\cdot\x} + 
\bar\tau e^{iR^2\bK_\star\cdot\bx}\ \}\ \subset\ L^2_{\bK_\star,\bar\tau}.
\nn\end{align}
\item $\mu^{(0)}$ is the lowest eigenvalue of $H^{(0)}$ in $L^2_{\bK_\star}$.
\end{enumerate}
 \end{proposition}\medskip
 
\begin{proof} Without loss of generality, let $\bK_\star=\bK$. Since $R$ is orthogonal, $|\bK|=|R\bK|=|R^2\bK|$. Therefore, $-\Delta\Psi=|\bK|^2\Psi$ for $\Psi=e^{i\bK\cdot\bx},\  e^{iR\bK\cdot\bx}$ and $e^{iR^2\bK\cdot\bx}$. It follows that $\mu^{(0)}=|\bK|^2$ is an eigenvalue of multiplicity at least three. To show that the multiplicity is exactly three, we seek all $\bm=(m_1,m_2)$ for which $|\bK^\bm|^2=\ |\bK|^2$. Using $\bK^\bm=\bK + m_1\bk_1+m_2\bk_2$, we obtain
\begin{equation}
m_1^2 + m_2^2 + m_1-m_2 - m_1 m_2 = 0
\nn\end{equation}
which can be zero only if $\bm=(0,0), (0,1)$ or $(-1,0)$. 
In the first instance, $\bK^{(0,0)}=\bK$. If $\bm=(0,1)$, then $\bK^{(0,1)}=\bK + k_2 = R\bK$.
Finally, if $\bm=(-1,0)$ then 
 $\bK^{(-1,0)}=\bK-\bk_1= (\bK + \bk_2 ) - (\bk_1+\bk_2) = R\bK + R\bk_2 = R(\bK+\bk_2)=R^2\bK$. This proves conclusion 1. 
Proposition  \ref{Fourier-espaces} above, which characterizes the Fourier series of functions in $L^2_{\bK,\sigma},\ \sigma=1,\tau,\bar\tau$ implies conclusion 2.
 Conclusion 3. holds because $m_1^2+m_2^2-m_1m_2+m_1-m_2\ge1$ for $(m_1,m_2)\in\mathbb{Z}^2$ other than $(0,0), (0,1)$ and $(-1,0)$.
\end{proof}
\bigskip

Recall that for each $\bk\in\brill_h$, the  $L^2_\bk$ eigenvalues of $H^{(0)}$ are ordered \eqref{eig-ordering}:
\begin{equation}
\mu^{(0)}_1(\bk)\le \mu^{(0)}_2(\bk)\le \mu^{(0)}_3(\bk)\le\mu^{(0)}_4(\bk)\le\dots
\label{mu0_123}
\end{equation}

For $\bk=\bK$ we have
\begin{equation}
|\bK|^2\ =\ \mu^{(0)}_1(\bK)= \mu^{(0)}_2(\bK)= \mu^{(0)}_3(\bK)\  <\ \mu^{(0)}_4(\bK)\le\dots
\label{mu0_123A}
\end{equation}

We shall see in section \ref{sec:pfeps-small} that for small $\eps$, the spectrum $L^2_\bK$ perturbs to 
\begin{align}
&\textrm{either}\nn\\
&\ \ \ \  \ \mu^{(\eps)}_1(\bK)= \mu^{(\eps)}_2(\bK)< \mu^{(\eps)}_3(\bK)\ < \mu^{(\eps)}_4(\bK)\le \dots
\label{case12}\\
&\textrm{or}\nn\\
& \ \mu^{(\eps)}_1(\bK)< \mu^{(\eps)}_2(\bK)= \mu^{(\eps)}_3(\bK)\  <\ \mu^{(\eps)}_4(\bK)\le\dots
\label{case23}\end{align}
In either case, the multiplicity three eigenvalue splits into a multiplicity two eigenvalue and a simple eigenvalue. The connection between the double eigenvalue and conical singularities 
of the dispersion surface is explained in the next section; see Theorem  \ref{prop:2impliescone}.
We shall see from  Theorem \ref{main-thm}, or rather its proof (in section \ref{sec:pfeps-small}) that for all small $\eps$, 
  conical singularities occur at all vertices $\bK_\star$ of $\brill_h$,  and that these occur  at the intersection point of
the first and second band dispersion surfaces in the case of \eqref{case12}, and at the intersection of the second and third bands in the case of \eqref{case23}. As $\eps$ increases, we continue to have such conical intersections of dispersion surfaces, but we do not control {\it which} dispersion surfaces intersect.

 \section{Multiplicity two $L^2_\bK$ eigenvalues of $H^{(\eps)}$ and conical singularities}\label{sec:doubleimpliescone}
 
 Let $\bK_\star$ a point of $\bK$ or $\bK'$ type.  In this section we show that 
  if $H_V$ acting in $L^2_{\bK_\star}$ has a dimension two eigenspace 
  $\mathbb{E}_\tau\oplus\mathbb{E}_{\bar\tau}$, where $\mathbb{E}_\tau$ and $\mathbb{E}_{\bar\tau}$ are dimension one subspaces of  $L^2_{\bK_\star,\tau}$  and  $L^2_{\bK_\star,\bar\tau}$, respectively, then the dispersion surface is conical in a neighborhood of $\bK_\star$. A related analysis 
   is carried out in \cite{Grushin:09}, where a more general class of spectral problems is considered and weaker conclusions obtained, {\it e.g.} see the notion of {\it conical point} in \cite{Grushin:09}. \bigskip
  
\noindent Recall that we assume $V\in C^\infty(\mathbb{R}^2/\Lambda_h)$. Below we shall, for notational convenience, suppress the subscript $V$ and write simply $H$ for $H_V=-\Delta+V$.
 \bigskip
 
 \begin{theorem} \label{prop:2impliescone}
 Let $H=-\Delta +V$, where $V(\bx)$ is a honeycomb lattice potential in the sense of Definition \ref{honeyV}. Let $\bK_\star$ denote any vertex of the Brillouin zone, $\brill_h$.
  Assume further that
  \begin{itemize}
  \item[(h1.$\tau$)]\ $H$ has an $L^2_{\bK_\star,\tau}$\ - eigenvalue, $\mu_0=\mu(\bK_\star)$, of multiplicity one, with corresponding eigenvector $\Phi_1(\bx)=\Phi_1\left(\bx;\bK_\star\right)$, normalized to have $L^2(\Omega)$ norm equal to one.
  \item[(h1.$\bar\tau$)]\ $H$ has an $L^2_{\bK_\star,\bar\tau}$\ - eigenvalue, $\mu_0=\mu(\bK_\star)$, of multiplicity one, with corresponding eigenvector $\Phi_2(\bx)=\overline{\Phi_1(-\bx)}$, 
  \item[(h2)]\ $\mu_0=\mu(\bK_\star)$ is not an eigenvalue of $H$ on $L^2_{\bK_\star,1}$.
  \item[(h3)]\ the following nondegeneracy condition holds:
  \begin{equation}
\lambda_\sharp \equiv   3\times {\rm area}(\Omega)\times \sum_{\bm\in\mathcal{S}} 
c(\bm;\Phi_1)^2\ \left(\begin{array}{c}1\\ i\end{array}\right)\ \mathbf{\cdot}\ \bK_\star^\bm\ \ne\ 0,
\label{lambda-sharp1} 
\end{equation}
where $\{c(\bm;\Phi_1)\}_{\bm\in\mathcal{S}}$ are Fourier coefficients of $\Phi_1$, as defined in Proposition \ref{Fourier-espaces}.
  \end{itemize}
   Then $H$ acting on $L^2$ has a dispersion surface which, in a neighborhood of $\bk=\bK_\star$,  is conical. That is, for $\bk-\bK_\star$ near ${\bf 0}$,
    there are two distinct branches of eigenvalues of the Floquet-Bloch eigenvalue problem with quasi-momentum,
     $\bk$:
 \begin{align}
 \mu_+(\bk) - \mu(\bK_\star)\ &=\ +\left|\lambda_\sharp\right|\ \left|\bk-\bK_\star\right|\ \left(\ 1+E_+(\bk-\bK_\star)\ \right)\label{localcone+}\\
 \mu_-(\bk) - \mu(\bK_\star)\ &=\ -\left|\lambda_\sharp\right|\ |\bk-\bK_\star|\ \left(\ 1+E_-(\bk-\bK_\star)\ \right),
\label{localcone-}  \end{align}
where $E_\pm(\kappa)=\mathcal{O}(|\kappa|)$ as $|\kappa|\to0$ and $E_\pm$ are Lipschitz continuous functions in a neighborhood of $0$.
 \end{theorem}
 
 \begin{remark}\label{lam-sharp-rmk}
 \begin{enumerate}
 \item Elliptic regularity implies that the eigenfunctions $\Phi_j,\ j=1,2$ are in $H^2(\mathbb{R}^2)$. Therefore, $\sum_{\bm\in\mathcal{S}}(1+|\bm|^2)^2|c(\bm)|^2<\infty$. We conclude that the sum defining $\lambda_\sharp$ converges.
 \item In section \ref{sec:pfeps-small} we study the case of ``weak'' or small potentials, {\it i.e.} $V=\eps V_h$ with $\eps$ small. For all $\eps$ such that $0<|\eps|<\eps^0$,
   where  $\eps^0$ is a sufficiently small positive number, we will:\\
   (i)\ verify the double eigenvalue hypothesis (h1) of Theorem \ref{prop:2impliescone} by showing
    the persistence of a double eigenvalue due to intersection of the bands one and two in case \eqref{case12} or bands two and three in case \eqref{case23}, 
    (ii)\ verify hypothesis (h2) of Theorem \ref{prop:2impliescone} by showing, via explicit calculation, 
    that the $L^2_{\bK_\star,1}$ eigenvalue of $H$,  differs from the double eigenvalue, and 
    (iii)\ verify (h3) by showing:
   $ |\lambda_\sharp^\eps|^2\ =\ 16\ {\rm area}(\Omega)^2\ \pi^2/a^2\ +\ \mathcal{O}(\eps)$;
    see \eqref{lambda-sharp-eps-small2}. 
  Theorem \ref{prop:2impliescone} then implies the existence of a non-degenerate cone at  each vertex of $\brill_h$
   for \underline{all} sufficiently small non-zero $\eps$. 
   \item The condition: $\lambda_\sharp\ne0$ in \eqref{lambda-sharp1}   is independent of the normalization of the eigenfunction, $\Phi_1$. 
 \end{enumerate}
 \end{remark}
\bigskip

\nit {\it Proof of Theorem \ref{prop:2impliescone}:}\ 
 By Symmetry Remark \ref{symmetry-reduction}, we may without loss of generality consider the specific  $\brill_h$ vertex: $\bK_\star=\bK$. The local character of all others is identical.

 We consider a perturbation of $\bK$, $\bK+\bkappa$, with $|\bkappa|$ small. We express
$\Phi\in L^2_\bk$ as $\Phi(x;\bk) = e^{i\bk\cdot\bx} \psi(\bx;\bk)$, where $\psi(\bx;\bk)$ is $\Lambda$- periodic. The eigenvalue problem for $\bk=\bK+\bkappa$ takes the form:
\begin{align}
&\left(\ -\left(\nabla_\bx + i\left(\bK+\bkappa\right)\right)^2\ +\ V(\bx)\ \right) \psi(\bx;\bK+\bkappa)\ =\ \mu(\bK+\bkappa)\psi(\bx;\bK+\bkappa)\ ,\label{Hk-evp}\\
&\psi(\bx+\bv;\bK+\bkappa)=\psi(\bx;\bK+\bkappa),\ \ \textrm{for all}\ \bv\in \Lambda\ .
\label{psi-per}\end{align}

Let $\mu_0=\mu^{(0)}=\mu(\bK)$ be the double eigenvalue and let $\psi^{(0)}$ be in  the corresponding two-dimensional eigenspace. Express  $\mu(\bK+\bkappa)$ and $\psi(\bx;\bK+\bkappa)$ as:
\begin{align}
\mu(\bK+\bkappa) &= \mu^{(0)} +\ \mu^{(1)},\qquad  \psi(\bx;\bK+\bkappa) = \psi^{(0)} +\ \psi^{(1)}, \label{mu-psi-exp}
\end{align}
where  $\psi^{(1)}$  is to be chosen orthogonal to the nullspace of 
 $H(\bK)-\mu^{(0)}I$,  and 
$\mu^{(1)}$ are corrections to be determined. 
 Substituting \eqref{mu-psi-exp} into the eigenvalue problem \eqref{Hk-evp}-\eqref{psi-per} we obtain:
\begin{align}
& \left(\ H(\bK)\ - \mu_0 I \ \right)\psi^{(1)} \nn\\
&= \ \left(\ 2i\bkappa\cdot\left(\nabla+i\bK\right) - \kappa\cdot\kappa\ +\ \mu^{(1)}\ \right)\psi^{(1)}
\nn\\ &  +\ 
\left(2i\bkappa\cdot\left(\nabla+i\bK\right)-\kappa\cdot\kappa+\mu^{(1)}\right)\psi^{(0)}\nn\\
& \equiv\ F^{(1)},\qquad \psi^{(1)}\in L^2_{{\rm per},\Lambda}\ .
\label{order1}
\end{align}

Since $\psi^{(0)}$ is in  the $L^2_{{\rm per},\Lambda}$- nullspace of $H(\bK)-\mu_0 I$, we write it as  
\begin{align}
\psi^{(0)}(\bx)\ &=\ \alpha\phi_1(\bx)\ +\ \beta\phi_2(\bx), \ \ {\rm where}\label{psi0}\\
\phi_j(\bx)\ &=\ e^{-i\bK\cdot\bx}\ \Phi_j(\bx).\ \ j=1,2\label{phi-j}
\end{align}
Here $\phi_1$ and $\phi_2$ are normalized eigenstates with Fourier expansions as in part 3 of  Proposition
 \ref{Fourier-espaces} and  $\alpha, \beta$ are constants to be determined. \medskip

We now turn to the construction of $\psi^{(1)}$. Introduce the orthogonal projections: $Q_\parallel$,  onto  the two-dimensional kernel of $H(\bK)\ - \mu_0 I$,
and $Q_\perp=I-Q_\parallel$. Note that 
\begin{equation}
Q_\parallel \psi^{(1)}\ =\ Q_\perp \psi^{(0)}\ =\ 0,\ {\rm and}\ \  Q_\perp \psi^{(1)}= \psi^{(1)}\ . \label{Qpsi}
\end{equation}

We next seek a solution to \eqref{order1} by solving the following system for $\psi^{(1)}$ and $ \mu^{(1)}$:
\begin{align}
\left(\ H(\bK) - \mu_0 I \ \right)\psi^{(1)}\ 
  &=\ Q_\perp\ F^{(1)}(\alpha,\beta,\kappa,\mu^{(1)},\psi^{(1)})\label{LS1}\\
 0\ &=\ Q_\parallel F^{(1)}(\alpha,\beta,\kappa,\mu^{(1)},\psi^{(1)})
\label{LS2}\end{align}

Equation \eqref{LS2} is a system of two equations obtained by setting the projections of $F^{(1)}$
onto $\phi_1$ and $\phi_2$ equal to zero. Our strategy is to solve \eqref{LS1} for $\psi^{(1)}$ as a continuous functional of $\alpha,\beta,\kappa,\mu^{(1)}$ with appropriate estimates, then substitute the result into \eqref{LS2} to obtain a closed {\it bifurcation equation}. This is a linear homogeneous system of the form $\mathcal{M}(\mu^{(1)},\kappa)(\alpha,\beta)^t=0$. The function $\kappa\mapsto\mu^{(1)}(\kappa)$ is then determined by the condition
 that $\det\mathcal{M}(\mu^{(1)},\kappa)=0$.

 Written out in detail, the system \eqref{LS1}-\eqref{LS2} becomes:
 \begin{align}
 \left(\ H(\bK)\ - \mu_0 I \ \right)\psi^{(1)} 
&= \ Q_\perp\ \left(\ 2i\bkappa\cdot\left(\nabla+i\bK\right) \ -\kappa\cdot\kappa+\mu^{(1)}\right)\psi^{(1)}
\nn\\ &  +\ 
Q_\perp\ \left(2i\bkappa\cdot\left(\nabla+i\bK\right)\right)\psi^{(0)}
\label{LS1a}\\
&\nn\\
&Q_\parallel\ \left(2i\bkappa\cdot\left(\nabla+i\bK\right)-\kappa\cdot\kappa+\mu^{(1)}\right)\psi^{(0)}
\nn\\
&+ Q_\parallel\ \left(\ 2i\bkappa\cdot\left(\nabla+i\bK\right) \ \right)\psi^{(1)}=0
\label{LS2a}
\end{align}
Introduce the resolvent operator:
\begin{equation} R_\bK(\mu_0)\ =\ \left(\ H(\bK)\ - \mu_0\ I \right)^{-1}
\nn\end{equation}
defined as a bounded map from   $Q_\perp L^2$ to $Q_\perp H^2(\mathbb{R}^2/\Lambda_h)$. Equation \eqref{LS1a} for $\psi^{(1)}$ can be rewritten as:
\begin{align}
&\left( I\ +\ R_\bK(\mu_0)Q_\perp\ \left(\ -2i\bkappa\cdot\left(\nabla+i\bK\right) + \kappa\cdot\kappa\ -\ \mu^{(1)}\ \right)\ \right)\ \psi^{(1)}
\nn\\
&\ \ =\ R_\bK(\mu_0)\ Q_\perp\ \left(2i\bkappa\cdot\left(\nabla+i\bK\right) \right)\psi^{(0)}\label{psi1eqn}
\end{align}
In several equations above we have used \eqref{Qpsi}.

By elliptic regularity, the mapping \[f\mapsto Af\ \equiv\ R_\bK(\mu_0)Q_\perp\ \left(\ -2i\bkappa\cdot\left(\nabla+i\bK\right) + \kappa\cdot\kappa\ -\ \mu^{(1)}\ \right)f\] is a bounded operator on 
$H^s(\mathbb{R}^2/\Lambda_h)$, for any $s$. Furthermore, for $|\kappa|+|\mu^{(1)}|$ sufficiently small, the operator norm of $A$ is less than one, $(I+A)^{-1}$ exists, and hence \eqref{psi1eqn} is uniquely solvable in $Q_\perp H^2(\mathbb{R}^2/\Lambda_h)$:
\begin{align}
\psi^{(1)}\ 
&=\ \  \left( I\ +\ R_\bK(\mu_0)Q_\perp\ \left(\ -2i\bkappa\cdot\left(\nabla+i\bK\right) + \kappa\cdot\kappa\ -\ \mu^{(1)}\ \right)\ \right)^{-1}\nn\\
&\qquad\qquad\qquad  \ \circ\ \ \ \  R_\bK(\mu_0)\ Q_\perp\ \left(2i\bkappa\cdot\left(\nabla+i\bK\right)\ \right)\psi^{(0)}.
\nn\end{align}
Since $\psi^{(0)}$ is given by \eqref{psi0},  $\psi^{(1)} $ is clearly linear in $\alpha$ and $\beta$ and we write:
\begin{equation}
\psi^{(1)}\ =\ c^{(1)}[\kappa,\mu^{(1)}](\bx)\ \alpha\ +\ c^{(2)}[\kappa,\mu^{(1)}](\bx)\ \beta,
\label{psi1-solved}\end{equation}
where $(\kappa,\mu^{(1)})\mapsto c^{(j)}[\kappa,\mu^{(1)}]$ is a smooth mapping from 
a neighborhood of $(0,0)\in \mathbb{R}^2\times\mathbb{C}$ into $H^2(\mathbb{R}^2/\Lambda_h)$
satisfying the bound: 
\[\|c^{(j)}\|_{H^2}\le C(|\kappa|+|\mu^{(1)}|),\ \ j=1,2\ .\]
Note that $Q_\parallel c^{(j)}=0, j=1,2$.

We next substitute \eqref{psi1-solved} into \eqref{LS2a} to obtain a system of two homogeneous linear equations for $\alpha$ and $\beta$.  Using the relations:
\begin{align}
&\nabla_\bK \phi_j = e^{-i\bK\cdot\bx}\nabla e^{i\bK\cdot\bx}\phi_j= e^{-i\bK\cdot\bx}\nabla\Phi_j,\ \ 
\langle\phi_i,\phi_j\rangle  =\langle\Phi_i,\Phi_j\rangle=\delta_{ij},\ i,j=1,2\nn\\
&\nn\\
&C^{(j)}[\kappa,\mu^{(1)}](\bx) \equiv e^{i\bK\cdot\bx}c^{(j)}[\kappa,\mu^{(1)}](\bx),\ \ \left\langle\Phi_i,C^{(j)}\right\rangle\ =\ 0,\ \  i,j=1,2 \label{Psi1def}
\end{align}
 we have:
 \begin{equation}
 \mathcal{M}(\mu^{(1)},\kappa)\  \left(\begin{array}{c} \alpha \\ { }\\ \beta\end{array}\right)\ =\ 0\ ,
 \label{Mveca}
 \end{equation}
 where $\mathcal{M}(\mu^{(1)},\kappa)$ is the $2\times2$ matrix given by:
{\footnotesize{
\begin{align}
\mathcal{M}(\mu^{(1)},\kappa)&\equiv\  
\left(\begin{array}{cc} \mu^{(1)}-\kappa\cdot\kappa+\left\langle\Phi_1,2i\kappa\cdot\nabla\Phi_1\right\rangle & \langle\Phi_1,2i\kappa\cdot\nabla\Phi_2\rangle \\
&\nn\\
\left\langle\ \Phi_2,2i\kappa\cdot\nabla\Phi_1\ \right\rangle &
 \mu^{(1)}-\kappa\cdot\kappa+\langle\ \Phi_2,2i\kappa\cdot\nabla\Phi_2\ \rangle 
 \end{array}\right) \nn\\
 &\nn\\
 &+  \left(\begin{array}{cc} \left\langle\ \Phi_1,2i\kappa\cdot\nabla
 C^{(1)}(\kappa,\mu^{(1)})\ \right\rangle & \langle\ \Phi_1,2i\kappa\cdot\nabla C^{(2)}(\kappa,\mu^{(1)})\ \rangle\nn\\
 &\nn\\
\langle\ \Phi_2,2i\kappa\cdot\nabla C^{(1)}(\kappa,\mu^{(1)})\ \rangle &
  \langle\ \Phi_2,2i\kappa\cdot\nabla C^{(2)}(\kappa,\mu^{(1)})\ \rangle 
 \end{array}\right)\ . 
\end{align}
}}
Thus, $\mu(\bK+\kappa)=\mu^{(0)}+\mu^{(1)}$ is an eigenvalue for the spectral problem \eqref{Hk-evp}-\eqref{psi-per} if and only if $\mu^{(1)}$ solves:
\begin{equation}
\det\mathcal{M}(\mu^{(1)},\kappa)=0.\label{det0}
\end{equation}
Equation \eqref{det0} is an  equation for $\mu^{(1)}$, which characterizes the  splitting of the double eigenvalue at $\bkappa\ne0$. We now proceed to show that if the nondegeneracy condition \eqref{lambda-sharp1} holds, then the solution set of \eqref{det0}
is locally conic.\medskip

We anticipate that a  solution $\mu^{(1)}=\mathcal{O}(|\kappa|)$ and hence $C^{(j)}=\mathcal{O}(|\kappa|)$. This motivates expanding $\mathcal{M}$ as:
\begin{align}
\mathcal{M}(\mu^{(1)},\kappa) &= \mathcal{M}_0(\mu^{(1)},\kappa)+ \mathcal{M}_1(\mu^{(1)},\kappa),
\ \   {\rm where} \label{Mdecomp}\\
&\nn\\
\ \ \mathcal{M}_0(\mu^{(1)},\kappa) &= \left(\begin{array}{cc} \mu^{(1)} + 2i\left\langle \Phi_1,\bkappa\cdot\nabla \Phi_1\right\rangle & 2i\left\langle \Phi_1,\bkappa\cdot\nabla\Phi_2\right\rangle\\
2i\left\langle \Phi_2,\bkappa\cdot\nabla\Phi_1\right\rangle &
 \mu^{(1)} + 2i\left\langle \Phi_2,\bkappa\cdot\nabla\Phi_2\right\rangle
 \end{array}\right),\ \ \ \ {\rm and}\ \  \label{M0def}\\
 &\nn\\
  \mathcal{M}_{1,ij}(\mu^{(1)},\kappa) &=\mathcal{O}\left(|\kappa|\cdot|\mu^{(1)}|\ +\ |\kappa|^2\right)
 \label{M1est}
 \end{align}
 
 Note that 

 \begin{equation}
 \Phi_2(\bx) = \overline{\Phi_1(-\bx)}\ \implies\ 
 \left\langle \Phi_1,\bkappa\cdot\nabla \Phi_1\right\rangle\ =\ \left\langle \Phi_2,\bkappa\cdot\nabla \Phi_2\right\rangle
 \end{equation}
 Further, it is easily seen that this expression is purely imaginary:
 \begin{equation}
  \left\langle \Phi_1,\bkappa\cdot\nabla \Phi_1\right\rangle\ =\ i\ \Im\  \left\langle \Phi_1,\bkappa\cdot\nabla \Phi_1\right\rangle
  \label{Phi11-pi}
  \end{equation}

Now we claim that, in a neighborhood of $\kappa=0$,  the solutions of \eqref{det0} are well approximated by those of the truncated equation:
\begin{align}
\det\mathcal{M}_0(\nu,\kappa)\ &=\ 
\det\left(\begin{array}{cc} 
\nu - 2\Im \left\langle \Phi_1,\bkappa\cdot\nabla \Phi_1\right\rangle & 2i\left\langle \Phi_1,\bkappa\cdot\nabla\Phi_2\right\rangle\\
\overline{2i\left\langle \Phi_1,\bkappa\cdot\nabla\Phi_2\right\rangle} &
 \nu - 2\Im \left\langle \Phi_1,\bkappa\cdot\nabla\Phi_1\right\rangle
 \end{array}\right)\nn\\
 &=\  \left(\ \nu - 2\ \Im \left\langle \Phi_1,\bkappa\cdot\nabla \Phi_1\right\rangle\ \right)^2\ -\ 
  4\ \left|\ \left\langle \Phi_1,\bkappa\cdot\nabla\Phi_2\right\rangle\ \right|^2\ =\ 0\ .
\label{detM0eq0}
\end{align}
\medskip

\begin{remark}\label{vertexnotnecessary}
 We have not used $\mathcal{R}$- symmetry and special structure of the Fourier modes at vertices, $\bK_\star$, in obtaining \eqref{det0} and its approximation \eqref{detM0eq0}. We have only used that there is a two-dimensional eigenspace spanned by: $\Phi_1(\bx),\  \Phi_2(\bx) = \overline{\Phi_1(-\bx)}$.
\end{remark}
\medskip

We next use that $\bK$ is a vertex of $\brill_h$ to simplify and solve \eqref{detM0eq0} \ (Proposition \ref{prop:matrix-elements}) 
and then show that the solutions of \eqref{det0} are small corrections to these
 (Proposition \ref{M0pM1}).
\bigskip

 \begin{proposition}\label{prop:matrix-elements}
 \begin{align}
\left\langle \Phi_a,\bkappa\cdot \nabla\Phi_a\right\rangle &=\ 0,\ \ a=1,2\ .\label{diag0}\\
 2i\ \left\langle \Phi_1,\bkappa\cdot\nabla\Phi_2\right\rangle &=\  \overline{2i\ \left\langle \Phi_2,\bkappa\cdot\nabla\Phi_1\right\rangle}\nn\\
&=\  -3\ {\rm area}(\Omega)\ \left(\   \overline{\sum_{\bm\in\mathcal{S}}\ c(\bm)^2\ 
 \left(\begin{array}{c}1\\ i\end{array}\right)\cdot \bK^\bm\ }\ \right) \times\ (\kappa_1+ i\kappa_2)\nn\\
 &=\ -\overline{\lambda_\sharp}\ \times\ (\kappa_1+ i\kappa_2)\ ;\ \textrm{see \eqref{lambda-sharp1}}\ .
\label{offdiag}\end{align} 
 \end{proposition}
We prove Proposition \ref{prop:matrix-elements} just below. A consequence is that $\mathcal{M}_0$ simplifies to
\begin{equation}
\mathcal{M}_0(\nu;\kappa)\ =\ 
\left(\begin{array}{cc} \nu & -\overline{\lambda_\sharp}\ \times\ (\kappa_1+ i\kappa_2)\\
-\lambda_\sharp\ \times\ (\kappa_1- i\kappa_2) & \nu   \end{array}\right)
\label{M0nukappa}\end{equation} 
and therefore 
\begin{align}
\det\mathcal{M}_0(\nu;\kappa)\ &=\ \nu^2\ -\ |\lambda_\sharp|^2\ |\kappa|^2\nn\\
&\equiv\ \nu^2\ -\  
\left|3\times {\rm area}(\Omega)\times\sum_{ \bm\in\mathcal{S} } c(\bm)^2\  
 \left(\begin{array}{c}1\\ i\end{array}\right)\cdot \bK^\bm \right|^2\ \times\ |\bkappa|^2 .
 \label{mu1eqnA}
\end{align}
Therefore, the truncation of $\mathcal{M}$ to $\mathcal{M}_0$,  yields $\det\mathcal{M}_0(\nu,\kappa)\ =0$ 
 and a locally conical dispersion relation, provided $\lambda_\sharp\ne0$.

 \medskip

 \begin{proof}[Proof of Proposition \ref{prop:matrix-elements}]
  Recall that $\Phi_1\in L^2_{\bK,\tau}$ and $\Phi_2\in L^2_{\bK,\bar\tau}$ are given by \eqref{tau-series}
   and \eqref{taubar-series}, respectively. 
We first consider the diagonal elements and claim $ \left\langle \Phi_1,\bkappa\cdot\nabla\Phi_1\right\rangle=0$.  To see this, first apply $\bkappa\cdot\nabla$  to $\Phi_1$, given by \eqref{Phi1-Fourier}
  and obtain:
  \begin{align}
&  \bkappa\cdot\nabla\Phi_1 \nn\\
  &=\ i\ \bkappa\cdot\sum_{\bm\in{\cal S}}\ c(\bm)\ 
\left(\ \bK_\star^\bm e^{i\bK_\star^\bm\cdot\bx }\ +\ 
\bar{\tau}\  R\bK_\star^\bm e^{iR\bK_\star^\bm\cdot\bx }\ +\  
\tau\ R^2\bK_\star^\bm e^{iR^2\bK_\star^\bm\cdot\bx } \right) \ .
\nn\end{align}
Therefore, since $\int_{\cell} e^{-i\bK^\bn\cdot\bx}\ e^{i\bK^\bm\cdot\bx}d\bx=0$ if $\bm\ne\bn$ we have
\begin{align}
  \bkappa\cdot \left\langle \Phi_1,\nabla\Phi_1\right\rangle\ 
 &= {\rm area}(\Omega)\ \sum_{\bm\in\mathcal{S}} |c(\bm)|^2\ \bkappa\cdot \left( \bK_\star^\bm +
  |{\tau}|^2\  R\bK_\star^\bm  + |\tau|^2\ R^2\bK_\star^\bm\ \right)\ \ =\ 0.\nn\\
  &\label{a1b1}
 \end{align}
The latter equality holds since $|\tau|^2=1$, $I-R$ is invertible (${\rm spec}(R)=\{\tau,\bar\tau\}$) and $I+R+R^2=(I-R^3)(I-R)^{-1}= (I-I)(I-R)^{-1}=0$. Similarly,  $ \left\langle \Phi_2,\bkappa\cdot\nabla\Phi_2\right\rangle=0$. Thus, we have shown that the diagonal elements vanish,\ \eqref{diag0}. \medskip
 
 One can check directly that the  off-diagonal elements satisfy $\mathcal{M}_{0,12}=\overline{\mathcal{M}_{0,21}}$:
\begin{equation} 2i\left\langle \Phi_1,\bkappa\cdot\nabla\Phi_2\right\rangle
 =\ \overline{  2i\left\langle \Phi_2,\bkappa\cdot\nabla\Phi_1\right\rangle}\ .
\label{off-diag}\end{equation}
Furthermore, using \eqref{tau-series}, \eqref{taubar-series} we have  
\begin{align}
 & 2i\left\langle \Phi_1,\bkappa\cdot\nabla\Phi_2\right\rangle\nn\\
 &= -2\times {\rm area}(\Omega)\ \sum_{\bm\in\mathcal{S}} \overline{c(\bm)}^2\ \bkappa\cdot \left( \bK_\star^\bm +
  \tau^2\  R\bK_\star^\bm  + \bar{\tau}^2\ R^2\bK_\star^\bm\ \right),\nn\\
 &= -2\times {\rm area}(\Omega)\ \overline{\sum_{\bm\in\mathcal{S}}\ c(\bm)^2\ \bkappa\cdot 
  \left( I\ +
  \tau\  R  + \left(\tau\ R\right)^2\ \right)\ \bK_\star^\bm\ } .
  &\label{a1b2}
 \end{align}
 
 Note, by \eqref{Reigs} that $\tau R$ has  an eigenvalue $\bar\tau$ with corresponding eigenvector ${\bf \zeta}=2^{-1/2}(1, i)^t$ and an eigenvalue $1$ with corresponding eigenvector $\bar{\bf \zeta}=2^{-1/2}(1 , -i)^t$. We express $\bK^\bm$ as
 \begin{equation}
 \bK^\bm\ =\ \left\langle\zeta,\bK^\bm\right\rangle\zeta +  \left\langle\bar\zeta,\bK^\bm\right\rangle\bar\zeta,
 \nn\end{equation}
 where for ${\bf a}, {\bf b}\in\mathbb{C}^2$ we defined $\langle{\bf a}, {\bf b}\rangle = \overline{\bf a}\cdot {\bf b}$.
 Clearly, 
 $$\left(I\ +\ \tau\ R  + \left(\tau\ R\right)^2\right)\bar\zeta\ =\ 3\ \bar\zeta\ {\rm and}\ \left(I\ +
  \tau\  R  + \left(\tau\ R\right)^2\right)\zeta= \left(1+\bar\tau+(\bar\tau)^2\right)\ \zeta=0.$$
Therefore,  
  \begin{align}
&\bkappa\cdot \left( I\ +
  \tau\  R  + \left(\tau\ R\right)^2\right) \bK^\bm\ =\ 3\ \left\langle\bar\zeta,\bK^\bm\right\rangle\ 
  \bkappa\cdot\ \bar\zeta \ =\ 3\ \left(\zeta\cdot\bK^\bm\right)\ \times\ \left(\bkappa\cdot\ \bar\zeta\right)
  \label{sumRs} 
  \end{align}
  Substitution of \eqref{sumRs} into \eqref{a1b2} we obtain:
\begin{align}
 &2i\ \left\langle \Phi_1,\bkappa\cdot \nabla\Phi_2\right\rangle\nn\\
 &=\ -6\times {\rm area}(\Omega)\ \overline{\sum_{\bm\in\mathcal{S}}\ c(\bm)^2\ \left(\zeta\cdot\bK^\bm\right)\ \times\ \left(\bkappa\cdot\ \bar\zeta\right)\ } \nn\\
 &=\ -3\times {\rm area}(\Omega)\  \overline{\sum_{\bm\in\mathcal{S}}\ c(\bm)^2\
 \left[\ \left(\begin{array}{c} 1 \\ i\end{array}\right)\cdot \bK^\bm\ \right] \times 
 \left[\ \left(\begin{array}{c} \kappa_1 \\ \kappa_2\end{array}\right)\cdot \left(\begin{array}{c} 1 \\ -i\end{array}\right) \right]\ }\nn\\
 & =\ -3\times {\rm area}(\Omega)\   \overline{\sum_{\bm\in\mathcal{S}}\ c(\bm)^2\ 
 \left(\begin{array}{c}1\\ i\end{array}\right)\cdot \bK^\bm\ }\ \times\ (\kappa_1+i\kappa_2),\ 
 \label{Phi1gradPhi2} \end{align}
 which by the definition of $\lambda_\sharp$ in \eqref{lambda-sharp1} implies \eqref{offdiag}. Thus, $\det\mathcal{M}_0(\nu,\kappa)=\nu^2-|\lambda_\sharp|^2|\kappa|^2$ which proves \eqref{mu1eqnA}. Furthermore, the solutions of $\det\mathcal{M}_0(\nu,\kappa)=0$ define a non-trivial 
   conical surface provided:
   \begin{equation}
\lambda_\sharp \equiv   3\times{\rm area}(\Omega)\times\sum_{\bm\in\mathcal{S}} c(\bm)^2\ \left(\begin{array}{c}1\\ i\end{array}\right)\cdot \bK^\bm\ \ne\ 0.
\label{lambda-sharp}
\end{equation}
    The  proof of Proposition \ref{prop:matrix-elements} is complete.
  \end{proof}
  \bigskip
  
To complete the proof of Theorem \ref{prop:2impliescone} we next show that the local character of solutions to \eqref{det0} is,  for $|\kappa|$ small, essentially that derived   in  Proposition \ref{prop:matrix-elements} for the solutions of \eqref{detM0eq0}.
  
  \begin{proposition}\label{M0pM1}
 Suppose $\lambda_\sharp$, defined in \eqref{lambda-sharp}, is non-zero. Then, in a neighborhood of any $\bK$ or $\bK'$ point, the dispersion surface is conic. Specifically,
 the eigenvalue equation $\det\mathcal{M}(\mu^{(1)},\kappa)=0$ (see \eqref{det0}) defines, in a neighborhood $U\subset\mathbb{R}^2$ of $\kappa=0$, two functions:
 \begin{align}
 \mu_+^{(1)}(\kappa)\ =\ \left|\lambda_\sharp \right|\ |\kappa|\ \left(\ 1+ E_+(\kappa)\ \right),\ \ \ \mu_-^{(1)}(\kappa)\ =\ -\left|\lambda_\sharp \right|\ |\kappa|\ \left(\ 1+ E_-(\kappa)\ \right),
 \label{2branches}\end{align}
 where $E_\pm(\kappa)\to0$ as $\kappa\to0$ and $E_\pm(\bkappa)$ is Lipschitz continuous in $\bkappa$. 
  \end{proposition}
 \medskip
 
  \begin{proof} 
  By \eqref{det0}, and Proposition \ref{prop:matrix-elements}, $\mu^{(1)}$ satisfies:
  \begin{equation}
  \left(\mu^{(1)}\right)^2 \ =\ \left|\lambda_\sharp \right|^2\ |\kappa|^2\ +\ g_{21}\left(\mu^{(1)},\kappa\right)+
  g_{12}\left(\mu^{(1)},\kappa\right)+g_{03}(\kappa),
 \label{mu1eqn}
  \end{equation}
  where $g_{rs}$ are smooth functions satisfying the bounds:
  \begin{equation}
  \left|\ g_{rs}(\mu,\kappa)\ \right|\ \le\ C\ |\mu|^r\ |\kappa|^s
  \nn\end{equation}
  for $|\mu|\le1,\ |\kappa|\le1$. We now construct $\mu^{(1)}_+(\kappa)$. The construction of 
   $\mu^{(1)}_-(\kappa)$ is similar. Set $\mu_+^{(1)}=\left|\lambda_\sharp \right|\ |\bkappa|\ \left(1+\eta\right)$.
   Substitution into \eqref{mu1eqn} and using that $\lambda_\sharp\ne0$, we find that $\eta$ satisfies:
   \begin{align}
  G(\eta,\kappa)\ \equiv\ 2\eta+\eta^2 + J_1(\eta,\kappa) + J_2(\kappa)=0.
   \nn\end{align}
  Here $J_1$ and $J_2$ are smooth functions of $\eta$ and Lipschitz continuous functions of $\kappa$, such that: $J_1(\eta,\kappa)=\mathcal{O}(|\kappa|),\ \partial_\eta J_1(\eta,\kappa)=\mathcal{O}(|\kappa|),\ J_2(\kappa)=\mathcal{O}(|\kappa|)$ as $|\kappa|\to0$. Thus, $G(\eta,\kappa)$ and $\partial_\eta G(\eta,\kappa)$ are Lipschitz continuous in $(\eta,\kappa)$ with $G(0,0)=0$ and $\partial_\eta G(0,0)=2\ne0$. It follows easily that there exists $\eta=E(\kappa)$
defined and Lipschitz continuous in a neighborhood $U\subset\R^2$  of $\kappa=0$,  such that $E(0)=0$ and  $G(E(\kappa),\kappa)=0$ for all $\kappa\in U$.
  \end{proof}

\section{Main Theorem: Conical singularity in dispersion surfaces}\label{sec:MainTheorem}
\bigskip

 Assume that $V$ is a honeycomb lattice potential in the sense of Definition \ref{honeyV}. 
Since $V\in C^\infty(\mathbb{R}^2/\Lambda_h)$, its Fourier coefficients  satisfy
\begin{equation}
\hat{V}\in l^1(\Z^2),\ \ i.e.\ \ \|\ \hat{V}\ \|_{l^1(\Z^2)}=\sum_{\bm\in\Z^2} |V_\bm |<\infty\ .
\label{Vreg}\end{equation}

\begin{thm}\label{main-thm} {\bf Conical singularities and the dispersion surfaces of $H^{(\eps)}$}

Let $V(\bx)$ honeycomb lattice potential. Assume further that the Fourier coefficient of $V$, $V_{1,1}$, is non-vanishing, {\it i.e.}
\begin{equation}
V_{1,1}\ =\ \int_\Omega e^{-i(k_1+k_2)\cdot\by}\ V(\by)\ d\by\ \ne0\ .
\label{V11eq0}
\end{equation}
There exists a countable and 
closed set $\tilde{\mathcal{C}}\subset\mathbb{R}$ such that for any vertex $\bK_\star$  of $\brill_h$ and   all $\eps\notin\tilde{\mathcal{C}}$ the following holds:
\begin{enumerate}
\item   There exists a Floquet-Bloch eigenpair $\Phi^\eps(\bx;\bK_\star), \mu^\eps(\bK_\star)$ such that
\subitem $\mu^\eps(\bK_\star)$ is an $L^2_{\bK,\tau}$ - eigenvalue of $H^{(\eps)}$ of multiplicity one, with corresponding eigenfunction, $\Phi^\eps(\bx;\bK_\star)$.
\subitem $\mu^\eps(\bK_\star)$ is an $L^2_{\bK,\bar\tau}$ - eigenvalue of $H^{(\eps)}$ of multiplicity one, with corresponding eigenfunction, $\overline{\Phi^\eps(-\bx;\bK_\star)}$.
\subitem $\mu^\eps(\bK_\star)$ is \underline{not} an $L^2_{\bK,1}$- eigenvalue of $H^{(\eps)}$.
\item There exist $\delta_\eps>0,\ C_\eps>0$
  and Floquet-Bloch eigenpairs: $(\Phi_+^\eps(\bx;\bk), \mu_+^\eps(\bk))$  
 and $(\Phi_-^\eps(\bx;\bk), \mu_-^\eps(\bk))$, and  Lipschitz continuous functions, $E_\pm(\bk)$,   defined for  $|\bk-\bK_\star|<\delta_\eps$,  such that 
\begin{align}
\mu^\eps_+(\bk)-\mu^\eps(\bK_\star)\ &=\ +\ |\lambda^\eps_\sharp|\ 
\left| \bk-\bK_\star \right|\ 
\left(\ 1\ +\ E^\eps_+(\bk)\ \right)\ \ {\rm and}\nn\\
\mu^\eps_+(\bk)-\mu^\eps(\bK_\star)\ &=\ -\ |\lambda^\eps_\sharp|\ 
\left| \bk-\bK_\star \right|\ 
\left(\ 1\ +\ E^\eps_-(\bk)\ \right),\nn
\end{align}
where $\lambda_\sharp^\eps\ne0$ is given in terms of $\Phi^\eps(\bx;\bK_\star)$ by the expression in \eqref{lambda-sharp1} and 
$|E_\pm^\eps(\bk)| \le C_\eps |\bk-\bK_\star|$. 
Thus, in a neighborhood of the point $(\bk,\mu)=(\bK_\star,\mu_\star^\eps)\in \R^3 $, the dispersion surface is {\it conic}.
\item 
There exists $\eps^0>0$, such that for all $\eps\in(-\eps^0,\eps^0)\setminus\{0\}$\\
  (i)\  $\eps V_{1,1}>0\ \implies$ conical intersection of $1^{st}$ and $2^{nd}$ dispersion surfaces \\
 (ii)\ $\eps V_{1,1}<0\ \implies$ conical intersection of $2^{nd}$ and $3^{rd}$ dispersion surfaces\ .
 \end{enumerate}
\end{thm}

\medskip

\begin{rem}
Part 3 of Theorem \ref{main-thm} gives conditions for intersections of the first and second band dispersion surfaces or interesections of the second and third. 
 As the magnitude of $\eps$ is increased it is possible that there are crossings among the $L^2_{\bK,\sigma}$ - eigenvalues of $H^{(\eps)}$, so in general the theorem does not specify which band dispersion surfaces intersect.
\end{rem}\medskip

\subsection{Outline of the proof of Theorem \ref{main-thm}}\label{sec:proof-outline}

By Symmetry Remark \ref{symmetry-reduction}, it suffices to prove Theorem \ref{main-thm}
 for $\bK_\star=\bK$. 
We have seen that the central point is to verify for all $\eps$, except possibly those in a closed countable exceptional set, that  hypotheses (h1-h3) of 
 Theorem \ref{prop:2impliescone} hold. These hypotheses state that 
 $H^{(\eps)}$ has simple $L^2_{\bK_\star,\tau}$ and $L^2_{\bK_\star,\bar\tau}$ eigenvalues which are related by symmetry, which are not $L^2_{\bK_\star, 1}$- eigenvalues, and moreover that $\lambda_\sharp^\eps\ne0$.  We proceed as follows.

In section \ref{sec:pfeps-small} we show that there is a positive number, $\eps^0$, such that for all $\eps\in(-\eps^0,\eps^0)\setminus\{0\}$ (h1-h3) of Theorem \ref{prop:2impliescone}  hold. That is, the conclusions of Theorem \ref{main-thm} hold for all sufficiently small, non-zero $\eps$. In section \ref{det2} we introduce the key tool, a renormalized determinant, to detect and track the $L^2_{\bK,\sigma}$ eigenvalues of $H^{(\eps)}$ for $\sigma=1,\tau,\bar\tau$. A continuation argument is then implemented using tools from complex function theory in section \ref{sec:continuation}, to pass to large $\eps$. 
  We now embark on the detailed proofs.

\section{Proof of Main Theorem \ref{main-thm} for small $\eps$}\label{sec:pfeps-small}

We begin the proof of  Theorem \ref{main-thm} by first establishing
it for some interval $-\eps^0<\eps<\eps^0$, where $\eps^0$ is positive but possibly small.
We shall consider the eigenvalue problem for $H^{(\eps)}$ on the three  eigen-spaces of $\mathcal{R}$: $L^2_{\bK_\star,\tau}, L^2_{\bK_\star,\bar\tau}$ and $L^2_{\bK_\star,1}$:
\begin{align}
&H^{(\eps)}\ \Phi(\bx;\bK_\star)\ \equiv\ \left[\ -\Delta +\varepsilon V(\bx)\ \right]\ \Phi(\bx;\bK_\star) =\ \mu(\bK_\star)\ \Phi(\bx;\bK_\star)\label{Phi-evp}\\
&\Phi(\bx+\bv,\bK_\star)\ =\ e^{i\bK_\star\cdot\bv}\Phi(\bx,\bK_\star),\ \ \bx\in\R^2\nn\\
&\mathcal{R}\left[\Phi(\cdot;\bK_\star)\right](\bx)\ =\ \sigma\ \Phi(\bx;\bK_\star),\ \ \textrm{where} \ \ \sigma\in\{1,\tau,\bar\tau\}.\nn
\end{align}

 An eigenstate $\Phi(\bx;\bK_\star)$ in $ L^2_{\bK_\star,\sigma}$ is, by Proposition  \ref{Fourier-espaces}, of the form: 
 \begin{align}
 \Phi(\bx;\bK_\star)\ &=\ \sum_{\bm\in{\cal S}}\ c(\bm;\Phi)\ 
\left(\ e^{i\bK_\star^m\cdot\bx }\ +\ \bar{\sigma}\ e^{iR\bK_\star^m\cdot\bx }\ +\  
\sigma\ e^{iR^2\bK_\star^m\cdot\bx } \right)\ .
\label{Phi-expand}\end{align} 
The summation is over the set, $\mathcal{S}$,  introduced in Definition \ref{Sdef}. 
 Note that by Proposition \ref{Fourier-espaces} and Remark \ref{Vreal-even}, solutions to the eigenvalue problem on $ L^2_{\bK_\star,\bar\tau}$ can be obtained
  from those in $ L^2_{\bK_\star,\tau}$ via the symmetry:  $\Phi(\bx)\mapsto\overline{\Phi(-\bx)}$.

Recall that  $c_\Phi(\bm)$ or $c(\bm;\Phi)$ denote the $L^2_{\bK_\star,\sigma}$- Fourier coefficients of $\Phi $.
  Our next task is to reformulate the eigenvalue problem \eqref{Phi-evp} as an equivalent algebraic problem for the Fourier coefficients $\{c(\bm;\Phi(\cdot;\bK_\star))\}_{\bm\in\mathcal{S}}$.
 First, applying $-\Delta-\mu$ to $\Phi$, given by \eqref{Phi-expand}, and using that $R$ is orthogonal, we have that
\begin{align}
&\left(-\Delta-\mu\right)\Phi(\bx;\bK_\star)\nn\\
&\qquad\qquad =\sum_{\bm\in{\cal S}}\ \left(\ \left|\bK_\star^m\right|^2-\mu\right)\ c(\bm,\Phi)\ 
\left(\ e^{i\bK_\star^m\cdot\bx }\ +\ \bar{\sigma}\ e^{iR\bK_\star^m\cdot\bx }\ +\  
\sigma\ e^{iR^2\bK_\star^m\cdot\bx } \right)\ .
\label{DmuPhi-expand}\end{align} 
 Next, we claim that  $V(\bx)\Phi(\bx;\bK_\star)\in L^2_{\bK,\sigma}$. Indeed, since $V$ is $\mathcal{R}$-  invariant,  $\mathcal{R}[V](\bx)=V(R^*\bx)=V(\bx)$. Moreover, since $\Phi(\cdot;\bK_\star)\in L^2_{\bK_\star,\sigma}$, we have $\mathcal{R}[\Phi]=\sigma\Phi$. Therefore 
 \begin{equation*}
\mathcal{R}[V\Phi] = V(R^*\bx)\ \Phi(R^*\bx;\bK_\star) = V(\bx)\ \sigma\ \Phi(\bx;\bK_\star) =
\sigma\ V\Phi,\end{equation*}
Therefore, by Proposition \ref{Fourier-espaces},  $V\Phi(\cdot;\bK_\star)$ has the expansion
\begin{align}
 V(\bx)\Phi(\bx;\bK_\star)\ &=\ \sum_{\bm\in{\cal S}}\ c(\bm;V\Phi)\ 
\left(\ e^{i\bK_\star^m\cdot\bx }\ +\ \bar{\sigma}\ e^{iR\bK_\star^m\cdot\bx }\ +\  
\sigma\ e^{iR^2\bK_\star^m\cdot\bx } \right), \label{VPhi-expand}\\
c(\bm;V\Phi)\ &=\ \frac{1}{|\Omega|}\ \int_\Omega\ e^{-i\bK^\bm_\star\cdot\by}\ V(\by)\Phi(\by;\bK_\star)\ d\by\ .
\label{VPhi-Fourier-co}
\end{align} 

Furthermore, with the notation $\bq\bk\cdot\bx = (q_1\bk_1+q_2\bk_2)\cdot\bx$,

\begin{align}
&c(\bm;V\Phi) \  =\ \frac{1}{|\Omega|}\int_\Omega e^{-i\bK_\star^\bm\cdot\by} (V\Phi)(\by)\ d\by\nn\\
\ &=\ \frac{1}{|\Omega|}\int e^{-i\bK_\star^\bm\cdot\by}
\left(\sum_{\bq\in\Z^2} V_\bq\ e^{i\bk\bq\cdot\by}\right)\nn\\ 
&\ \ \ \ \times\ \left(\ \sum_{\br\in\mathcal{S}}\ c(\br;\Phi)
\left[e^{i\bK_\star^\br\cdot\by}+\bar\sigma\ e^{iR\bK_\star^\br\cdot\by} 
+\sigma\ e^{iR^2\bK_\star^\br\cdot\by}\right]\ \right)\nn\\ 
&=\   \frac{1}{|\Omega|}\sum_{\bq\in\Z^2,\br\in\mathcal{S}}\ V_\bq\ c(\br;\Phi)\nn\\
&\qquad\times\int_\Omega d\by
\left[e^{i(\bK_\star^\br-\bK_\star^\bm+\bq\bk)\cdot\by}+
\bar\sigma\ e^{i(R\bK_\star^\br-\bK_\star^\bm+\bq\bk)\cdot\by} 
+\sigma\ e^{i(R^2\bK_\star^\br-\bK_\star^\bm+\bq\bk)\cdot\by}\right]\nn\\
&=\ \frac{1}{|\Omega|} \sum_{\bq\in\mathbb{Z}^2,\br\in\mathcal{S}} V_\bq\ c(\br;\Phi) \int d\by\nn\\
&\qquad\times \left[e^{i\left(\bq-(\bm-\br)\right)\bk\cdot\by}+
\bar\sigma\ e^{i\left(\bq-(\bm-\mathcal{R}\br)\right)\bk\cdot\by}
+\sigma\ e^{i\left(\bq-(\bm-\mathcal{R}^2\br)\right)\bk\cdot\by}\right]\nn\\
&=\ \sum_{\bq\in\mathbb{Z}^2,\br\in\mathcal{S}} V_\bq\ c(\br;\Phi)\nn\\
&\qquad \times\ \left[\ \delta\left(\bq-(\bm-\br)\right)\ +\ \bar\sigma\ \delta\left(\bq-(\bm-\mathcal{R}\br)\right)\ +\ \sigma\ \delta\left(\bq-(\bm-\mathcal{R}^2\br)\right))\ \right]
\nn\end{align}

Thus,
\begin{align}
&c(\bm;V\Phi)\ =\ \sum_{ \br\in\mathcal{S} }\  \mathcal{K}_\sigma(\bm,\br)\ c(\br;\Phi),
\nn
\end{align}
 where (recall \eqref{Rbm})
\begin{align}
 \mathcal{K}_\sigma(\bm,\br)\ &\equiv V_{\bm-\br}+\bar\sigma\ V_{\bm-\cR\br}
+\ \sigma\ V_{\bm-\cR^2\br}\label{Kdef}\\
&=\ V_{m_1-r_1,m_2-r_2}\ +\ \bar\sigma\ V_{m_1+r_2,m_2+r_2-r_1-1}\ 
+\ \sigma\ V_{m_1+r_1-r_2+1,m_2+r_1}\ . \nn
\end{align}

Summarizing, we have
\begin{proposition}\label{equiv-alg} 
Let $\sigma\in\{1,\tau,\bar\tau\}$. Then, 
the spectral problem \eqref{Phi-evp} on $L^2_{\bK_\star,\sigma}$ is equivalent to  algebraic eigenvalue problem for $c(\bm)=c(\bm;\Phi)$ and $\mu$:
\begin{align}
&\left(\ \left|\bK_\star^m\right|^2-\mu\right)\ c(\bm)\ +\ \eps\ \sum_{ \br\in\mathcal{S} }\  
\mathcal{K}_\sigma(\bm,\br)\ c(\br)\ =\ 0,\ \ \ \bm\in\mathcal{S},
\label{a-spectral}\end{align}
where $\{c(\bm)\}_{\bm\in \mathcal{S}}\in l^2(\mathcal{S})$. 
\end{proposition}
\medskip

 To fix ideas, let $\bK_\star=\bK$; starting with $\bK'$, we would proceed similarly.
  For $\eps=0$, we have the algebraic eigenvalue problem:
\begin{align}
&\left(\ \left|\bK+m_1\bk_1+m_2\bk_2\right|^2-\mu\right)\ c(\bm)\ =\ 0,\ \ \ \bm\in\mathcal{S}.
 \label{alg-sys-eps0}
\end{align}
Equation \eqref{alg-sys-eps0}, viewed as an eigenvalue problem for  $\left(\{c(\bm)\}_{\bm\in\mathbb{Z}^2},\ \mu\right)$ is equivalent to the eigenvalue problem for $-\Delta$ on $L^2_\bK$ treated in 
Proposition \ref{H0-spec}. Restated in terms of Fourier coefficients, Proposition  \ref{H0-spec} states
that $\mu^{(0)}=|\bK|^2$ is an eigenvalue of multiplicity three with corresponding eigenvectors:
\begin{align*}
&c^{(0)}_1(m_1,m_2)=\delta_{m_1,m_2}\ \ \ &\leftrightarrow&\ \ \mu^{(0)}=|\bK|^2\\
&c^{(0)}_2(m_1,m_2)=\delta_{m_1,m_2-1}\ \ \ &\leftrightarrow&\ \ \mu^{(0)}=|\bK|^2=|R\bK|^2=|\bK+\bk_2|^2\\
&c^{(0)}_2(m_1,m_2)=\delta_{m_1+1,m_2}\ \ \ &\leftrightarrow&\ \ \mu^{(0)}=|\bK|^2=|R^2\bK|^2=|\bK-\bk_1|^2
\end{align*}
{\it Recall from Definition \ref{Sdef}  that the equivalence class of indices $\{(0,0),(0,1),(-1,0)\}$ has as its representative in $\mathcal{S}$ the point $(0,1)$. } \medskip

The  eigenvalue problem \eqref{alg-sys-eps0}
 has a one dimensional $L^2_{\bK,\sigma}$- eigenspace with eigenpair:
 \begin{equation}
 \mu^{(0)}=|\bK+\bk_2|^2=|\bK|^2,\ \ c(m_1,m_2)=\delta_{m_1,m_2-1},\ \ (m_1,m_2)\in\mathcal{S}
 \nn\end{equation}
corresponding to the $L^2_{\bK,\sigma}$ eigenstate of $H_0$:
 \begin{align*}
 \Phi^{\eps=0}(\bx;\bK)\ &=\  e^{i\bK^{0,1}\cdot\bx} + \bar\sigma e^{iR\bK^{0,1}\cdot\bx} + 
 \sigma e^{iR^2\bK^{0,1}\cdot\bx}\\
 &=\ e^{i(\bK+\bk_2)\cdot\bx} + \bar\sigma e^{i(\bK-\bk_1)\cdot\bx} + 
 \sigma e^{i\bK\cdot\bx}\nn\\
 &=\  \sigma\ e^{i\bK\cdot\bx}
 \left(\ 1\ +\ \bar\sigma\ e^{i\bk_2\cdot\bx}\ +\ \sigma e^{-i\bk_1\cdot\bx} \right)
\end{align*}
 
 We seek a solution of \eqref{a-spectral} for $\eps$ varying in a  small open interval about $\eps=0$. We proceed via a Lyapunov-Schmidt reduction argument. First, decompose the system \eqref{a-spectral} into coupled equations for:
 \begin{equation}
 c_\parallel \equiv c(0,1)\ \in\ \C,\ \ \textrm{and}\ \ \{c_\perp(\bm)\}_{\bm\in\mathcal{S}^\perp}\ \in\ l^2(\mathcal{S}^\perp)\ ,
 \label{c-par-perp}
 \end{equation}
 where
 \begin{equation}
 \mathcal{S}^\perp\equiv\mathcal{S}\setminus\{(0,1)\}\ .
 \label{Sperp}
 \end{equation}
 and rewrite \eqref{a-spectral} as a coupled system for $c_\parallel$ and $c_\perp$:
 \begin{align}
&\left[ \left|\bK^{0,1}\right|^2-\mu\ +\ \eps\ \mathcal{K}_\sigma(0,1,0,1)\right]\ c_\parallel\ +\ \eps\ \sum_{ \br\in\mathcal{S}^\perp }\  \mathcal{K}_\sigma(0,1,\br)\ c_\perp(\br)\ =\ 0,
\label{c-parallel-eqn}\\
&\eps\ \mathcal{K}_\sigma(\bm,0,1)\ c_\parallel\ \ +\ \left(\ \left|\bK^m\right|^2-\mu\right)\ c_\perp(\bm)\nn\\
&\nn\\
 & \ \  \ \ \ \  \ \  \ \ \ \  \ \  \ \ \ \ +\ \eps\ \sum_{ \br\in\mathcal{S}^\perp}\  \mathcal{K}_\sigma(\bm,\br)\ c_\perp(\br)\ =\ 0,\ \ \  \ \ \ \ \  \bm\in\mathcal{S}^\perp.
\label{a-spectral-split}\end{align}
We next seek a solution of \eqref{c-parallel-eqn}-\eqref{a-spectral-split}, for $\eps$ small, in a neighborhood of the solution to the $\eps=0$ problem: $c^0_\parallel=1, \mu^{(0)}=|\bK|^2, c_\perp(\br)=0, \br\in\mathcal{S}^\perp$. \medskip

We begin by solving the second equation in \eqref{a-spectral-split} for $c_\perp$ as a function of the scalar parameter $c_\parallel$. For $\eps$ small, the operator to be inverted is diagonally dominant 
 with diagonal elements: $\left|\bK^\bm\right|^2-\mu$, which we bound from below for $\bm\in\mathcal{S}^\perp$. By
  the relations \eqref{bk12} we have 
 \begin{align}
\left|\bK^m\right|^2-\mu\ &\equiv\ \ \left|\bK+m_1\bk_1+m_2\bk_2\right|^2\ -\ \mu\nn\\
  &= \left|\bK\right|^2-\mu\ +\ q^2\left(m_1^2+m_2^2-m_1m_2+m_1-m_2\right),\ \ 	
  q=\frac{4\pi}{ a\sqrt{3} }.
\nn \end{align}
If $\mu$ varies near $\mu^{(0)}=\left|\bK\right|^2$, then 
\begin{equation}
\left|\ \left|\bK^\bm\right|^2-\mu\ \right|^2\ \ge\ \frac{c_1}{a^2},\ \ \bm\in\mathcal{S}^\perp
\nn\end{equation}
for some $c_1>0$.
We now rewrite the equation for $c_\perp$ as:
\begin{align}
\left[\ \delta_{\bm,\br}\ +\ \frac{\eps}{|\bK^\bm|^2-\mu}\ \sum_{\br\in\mathcal{S}^\perp} \mathcal{K}_\sigma(\bm,\br)\ \right]\ c_\perp(\br)\ &=\ -\ \eps\ c_\parallel\ \frac{ \mathcal{K}_\sigma(\bm,0,1)}{|\bK^\bm|^2-\mu}\nn\\
& \equiv\ \eps\ c_\parallel\ F^\sigma_\bm(\mu),\ \bm\in\mathcal{S}^\perp
 \label{cperp-eqn}\end{align}
 or, more compactly,
 \begin{equation}
 \left(\ I\ +\ \eps\mathcal{T}_{\mathcal{K}_\sigma}(\mu)\ \right)c_\perp\ =\ \eps\ c_\parallel\ F^\sigma(\mu)
\label{compact-cperp} \end{equation}
 \nit Recall Young's inequality, which states that the operator  defined by
 $$T_Lf(\bm)\ =\ \sum_\br L(\bm,\br)f(\br)$$
 satisfies the bound
 \begin{align}
& \|T_Lf\|_{l^2(\mathcal{S}^\perp)}\ \le C_L\ \|f\|_{l^2(\mathcal{S}^\perp)},\ \ {\rm where}\label{youngs}\\
 & C_L\ =\ \sup_{\br}\sum_{\bm}\ |L(\bm,\br)|\ +\  \sup_{\bm}\sum_{\br} |L(\bm,\br)|\ \ .
\nn \end{align}
 We apply \eqref{youngs} with $L(\bm,\br)=\mathcal{K}_\sigma(\bm,\br)$, defined by \eqref{Kdef}, and conclude using the bound 
 $$\frac{1}{\left| |\bK^\bm|^2-\mu \right|}\ \le\ C\ \frac{1}{1+|\bm|^2},\ \ \bm\in\mathcal{S}^\perp,$$
and that  $\hat{V}=\{V_\bm\}_{\bm\in\mathcal{S}^\perp}\in l^1(\mathcal{S}^\perp)$ (recall \eqref{Vreg}),
 that 
 the operator 
 $$
 \mathcal{T}_{\mathcal{K}_\sigma}(\mu) f(\bm)\ =\ \frac{1}{|\bK^\bm|^2-\mu}\ \sum_{\br\in\mathcal{S}^\perp} \mathcal{K}_\sigma(\bm,\br) \ f(\br)$$
 maps $l^2(\mathcal{S}^\perp)\ \to\ l^2_2(\mathcal{S}^\perp)$ with the bound 
 \begin{equation}
 \|\ \mathcal{T}_{\mathcal{K}_\sigma}(\mu)f\ \|_{l^2_2(\mathcal{S}^\perp)}\ \le\ C\ \| f\|_{l^2(\Z^2)}.
 \label{calTbound}
 \end{equation}
 Here, $\| f\|_{l^2_2(\mathcal{S}^\perp)}^2\equiv\sum_{\bm\in\mathcal{S}^\perp}(1+|\bm|^2)^2|f(\bm)|^2$.

 \begin{proposition}\label{Tinvert1}
 There exists $\eps^0 >0$ such that for all $ |\varepsilon| < \varepsilon^0$ 
  and  any ${\bf f}\in l^2(\mathcal{S}^\perp)$ 
\begin{eqnarray}
\label{exis_sol1}
\left(I+\varepsilon \mathcal{T}_{\mathcal{K}_\sigma}(\mu)\right)c_{\perp}={\bf f}\;,
\end{eqnarray}
has a unique solution $c_\perp=c^\eps_\perp\in l^2_2(\mathcal{S}^\perp)$, analytic in $\eps$, satisfying 
$\left|\left|c^\eps_{\perp}\right|\right|_{l^2_2(\mathcal{S}^\perp)}
\le 2\left|\left| {\bf f}\right|\right|_{l^2_2(\mathcal{S}^\perp)}.$\\
\end{proposition}
We now apply Proposition \ref{Tinvert1} to solve \eqref{compact-cperp} to obtain
\begin{equation}
c_\perp(\br)=\eps\ c_\parallel\ \left[ \left(I+\eps\mathcal{T}_{\mathcal{K}_\sigma}(\mu)\right)^{-1} F^\sigma(\mu)\ \right](\br).
\label{cperp}
\end{equation}
Substitution into \eqref{c-parallel-eqn} yields a closed scalar equation for $c_\parallel$ of the form
 $\mathcal{M}_\sigma(\mu,\eps)c_\parallel=0$, which has a non-trivial solution if and only if:
{\small{
\begin{align}
&\mathcal{M}_\sigma(\mu,\eps)\nn\\
&\equiv\ \left|\bK\right|^2-\mu\ +\ \eps\ \mathcal{K}_\sigma(0,1,0,1)\ +\ \eps^2\ \sum_{ \br\in\mathcal{S}^\perp }\  \mathcal{K}_\sigma(0,1,\br)\  \left[\left(I+\eps\mathcal{T}_{\mathcal{K}_\sigma}(\mu)\right)^{-1} F^\sigma(\mu)\right](\br)\ =\ 0\nn\\
&\label{cparallecl}
\end{align}
}}
$M_\sigma(\mu,\eps)$ is analytic in a neighborhood of $(\mu,\eps)=(\mu^{(0)},0)=(|\bK|^2,0)$. 
Clearly, $M_\sigma(\mu^{(0)},0)=0$ and $\partial_\mu M_\sigma(\mu^0,0)=-1$. By the implicit function theorem, there exists 
$\eps^0>0$ such that defined in a complex neighborhood of the interval 
$|\eps|<\eps^0$, there is an analytic function $\eps\mapsto \mu^\eps$, such that 
\begin{equation}
\mathcal{M}_\sigma(\mu^\eps,\eps)\ =\ 0,\ \textrm{for}\  -\eps^0<\eps<\eps^0
\nn\end{equation}

Thus, we take $c_\parallel=1$ and via \eqref{cperp}-\eqref{cparallecl} our solution for $|\eps|<\eps^0$  is 
\begin{align}
\mu &=\ \mu^\eps\ =\ |\bK|^2 + \eps\mathcal{K}_\sigma(0,1,0,1) + \mathcal{O}(\eps^2)\nn\\
c^\eps_\parallel &=\ c(0,1)\ \equiv\ 1\nn\\
c^\eps_\perp &= \{c^\eps(\bm)\}_{\bm\in\mathcal{S}^\perp}\ =\ \eps\  \left(I+\eps\mathcal{T}_{\mathcal{K}_\sigma}(\mu^\eps)\right)^{-1} F_\sigma(\mu^\eps),\nn\\
& \ \ {\rm where}\ \ 
F_{\sigma,\bm}(\mu)\ =\ -\frac{\mathcal{K}_\sigma(\bm,0,1)}{|\bK^\bm|^2-\mu},\ \ \bm\in\mathcal{S}^\perp 
\ .\label{c-mu-expand}\end{align}

From the definition of $\mathcal{K}_\sigma(\bm,\br)$, displayed in \eqref{Kdef}, we find:
\begin{equation}
\mathcal{K}_\sigma(0,1,0,1) = 
V_{0,0}\ +\ \bar\sigma\ V_{1,1}\ +\ \sigma\ V_{0,1} = V_{0,0}+V_{1,1}\left(\sigma+\bar\sigma\right),
 \ \ \sigma = 1,\tau,\bar\tau.
\end{equation}
The latter equality uses:\\
(a) constraints on $V_{m_1,m_2}$ by  $\mathcal{R}- $ symmetry of $V$
 ($V_{0,-1}=V_{1,1}$) and  that\\
 (b) $V$ is even ($V_{0,-1}=V_{0,1}$). \\
 Furthermore, $V_{1,1}$ is real, since $V(\bx)$ is even and real ($V_{1,1}=V_{-1,-1}=\overline{V_{1,1}}$).  Therefore, $\mathcal{K}_\sigma(0,1,0,1)$ is real, as expected. 
 \medskip
 
The small $\eps$ perturbation theory of the three-dimensional eigenspace is now summarized:
\begin{proposition}\label{perturbed-espace}
 Assume $V_{1,1}\ne0$. Then, there exists $\eps^0>0$ such that
for $0<|\eps|<\eps^0$, the multiplicity three eigenvalue $\mu=|\bK|^2$ perturbs to $2$-dimensional and $1$-dimensional eigenspaces with corresponding eigenvalues $\mu^\eps(\bK)$ and $\tilde\mu^\eps(\bK)$ as follows:
\begin{enumerate}
\item $\mu^\eps(\bK)$ is of geometric multiplicity $2$ with  a $2$-dimensional eigenspace $\mathbb{X}_\tau\oplus \mathbb{X}_{\bar\tau}\subset L^2_{\bK,\tau}\oplus L^2_{\bK,\bar\tau}$ given by:
\begin{align}
\mu^\eps(\bK) &=\ \ |\bK|^2 + \eps\left( V_{0,0} + 2 V_{1,1}\ \cos(2\pi /3)\right) + \mathcal{O}(\eps^2)\nn\\
& =\ 
  |\bK|^2 + \eps\left(V_{0,0}- V_{1,1}\right)+ \mathcal{O}(\eps^2)
\label{mu-eps-small}
\end{align}
with eigenstates  $\Phi_1^\eps\in L^2_{\bK,\tau}$ and $ \Phi_2^\eps\in L^2_{\bK,\bar\tau}$, obtained by the symmetry (see Remark \ref{Vreal-even}):  \[\Phi_2(\bx;\bK)=\overline{\Phi_1(-\bx;\bK)},\] with Fourier expansions:
\begin{align}
 \Phi_1^\eps(\bx,\bK)\ &=\ \sum_{\bm\in{\cal S}}\ c^\eps(\bm)\ 
\left(\ e^{i\bK^m\cdot\bx }\ +\ \bar{\tau}\ e^{iR\bK^m\cdot\bx }\ +\  
\tau\ e^{iR^2\bK^m\cdot\bx } \right)\ . 
\label{eig-Phi1-Fourier}\\
 \Phi_2^\eps(\bx,\bK)\ &=\ \sum_{\bm\in{\cal S}}\ \overline{c^\eps(\bm)}\ 
\left(\ e^{i\bK^m\cdot\bx }\ +\ {\tau}\ e^{iR\bK^m\cdot\bx }\ +\  
\bar\tau\ e^{iR^2\bK^m\cdot\bx } \right),\ \ \textrm{and}
\label{eig-Phi12-Fourier}\end{align} 

\item $\tilde{\mu}^\eps(\bK)$ is a simple eigenvalue with eigenspace $\mathbb{X}_1\subset L^2_{\bK,1}$:
 \begin{equation}
\tilde\mu^\eps(\bK) =\ \ |\bK|^2 + \eps\left(V_{0,0}+ 2V_{1,1}\right)\  + \mathcal{O}(\eps^2)
\label{tilde-mu-eps-small}
\end{equation}
with eigenstate $\tilde\Phi$:
\begin{equation}
 \tilde\Phi^\eps(\bx,\bK)\ =\ \sum_{\bm\in{\cal S}}\ \tilde{c}^\eps(\bm)\ 
\left(\ e^{i\bK^m\cdot\bx }\ +\  e^{iR\bK^m\cdot\bx }\ +\  e^{iR^2\bK^m\cdot\bx } \right)\ . 
\label{eig-tildePhi-Fourier}
\end{equation}
\end{enumerate}
\end{proposition}
\medskip

Proposition \ref{perturbed-espace} implies that the double-eigenvalue hypotheses of Theorem \ref{prop:2impliescone} holds for $\eps$ positive and small. In particular, by \eqref{mu-eps-small} and \eqref{tilde-mu-eps-small} 
\begin{align}
&\textrm{If}\  \eps V_{1,1}>0, \textrm{then}\nn\\
&\ \ \ \  \ \mu^{(\eps)}_1(\bK)= \mu^{(\eps)}_2(\bK)< \mu^{(\eps)}_3(\bK)\ <\ \mu^{(\eps)}_4(\bK)\le \dots
\label{case12a}\\
&\nn\\
&\textrm{and if}\ \eps V_{1,1}<0, \ \textrm{then }\nn\\
& \ \mu^{(\eps)}_1(\bK)< \mu^{(\eps)}_2(\bK)= \mu^{(\eps)}_3(\bK)\  <\ \mu^{(\eps)}_4(\bK)\le\dots
\ \ \ \ \ \ \ .\label{case23a}\end{align}
By Theorem \ref{prop:2impliescone}, assuming $\lambda_\sharp^\eps\ne0$: 
\begin{itemize}
\item[(i)]\ if $\eps V_{1,1}<0$, the dispersion surfaces $\bk\mapsto \mu_1(\bk)$ and and $\bk\mapsto\mu_2(\bk)$
 intersect conically at the vertices of $\brill_h$
 \item[(ii)] if $\eps V_{1,1}>0$, the dispersion surfaces $\bk\mapsto \mu_2(\bk)$ and and $\bk\mapsto\mu_3(\bk)$  intersect conically at the vertices of $\brill_h$.
 \end{itemize}
 \medskip
 
\noindent So, in  order to apply Theorem \ref{prop:2impliescone}  it remains  to  check that $\lambda_\sharp^\eps\ne0$ . Here,
 $\lambda_\sharp^\eps$ is the expression given in \eqref{lambda-sharp1} . For $\eps$ small we have
\begin{equation}
\lambda_\sharp^\eps = 
3\ {\rm area}(\Omega)\ \left[\ \left(\begin{array}{c} 1 \\ i\end{array}\right)\ \cdot\ \bK^{(0,1)}\ \right]\ +\ \mathcal{O}(\eps),\label{lambda-sharp-eps-small1}
\end{equation}
where we have used that $c^\eps(0,1)=1$, \eqref{c-mu-expand} and that $\|c^\perp\|_{l^2_2(\mathcal{S}^\perp)}=\mathcal{O}(\eps)$. Note also that
\begin{equation}
\bK^{(0,1)} = \bK + \bk_2 = \frac{1}{3}\bk_1+\frac{2}{3}\bk_2 = \frac{q}{3}\left(\begin{array}{c} 3/2 \\ -\sqrt{3}/2\end{array}\right),\ \ q=4\pi/a\sqrt{3}.
\nn\end{equation}
Therefore, for $|\eps|<\eps^0$, with $\eps^0$ chosen sufficiently small,
\begin{equation}
\left| \lambda_\sharp^\eps \right|^2\ = 16\ {\rm area}(\Omega)^2\ \frac{\pi^2}{a^2}\ +\ \mathcal{O}(\eps).
\label{lambda-sharp-eps-small2}
\end{equation}
This completes the proof of our main theorem, Theorem \ref{main-thm},  for the case where $\eps$ is taken to be sufficiently small. We now turn to extending Theorem \ref{main-thm} to large $\eps$.

\section{Characterization of eigenvalues of $H^{(\eps)}$ for large $\eps$}\label{sec:det2}

To extend the assertions of Theorem \ref{main-thm} to large values of $\eps$, we introduce a  characterization of the $L^2_{\bK,\sigma}$- eigenvalues of the eigenvalue problem \eqref{Phi-evp} 
as zeros of an analytic function of $\eps$. \medskip

Since we can add an arbitrary constant to the potential, by redefinition of the eigenvalue parameter, $\mu$, 
we may assume without loss of generality  that 
 \[0\le V(\bx)\le V_{\rm max}.\]
Assume first that $\eps\in\mathbb{C}$ and $\Re\eps>0$. Then, 
$H^{(\eps)}-\mu I=-\Delta+\eps V -\mu I = \left(-\Delta+\eps V + I\right) - (\mu+1)I. $ The eigenvalue problem 
\eqref{Phi-evp} may be rewritten as
\begin{equation}
\left(-\Delta+\eps V + I\right)\Phi - (\mu+1)\Phi=0,\ \ u\in L^2_{\bK,\sigma}\ .
\label{evp-K1a}
\end{equation}
Now  for any real $\eps>0$ we have  $-\Delta+\eps V + I\ge I$. Hence we introduce
\footnote[3]{ For $\Re\eps>0$ and $f$ smooth, we have $\Re\left\langle\left(-\Delta+\eps V+I\right)f,f\right\rangle\ge \|f\|^2$. Hence the nullspace of $-\Delta+\eps V+I$ and its adjoint are $\{0\}$. By elliptic regularity theory $-\Delta+\eps V+I$ is invertible on $L^2_{\bK,\sigma}$.
}
\begin{equation}
 T(\eps)\ \equiv\ (I-\Delta+\eps V)^{-1},
 \label{Tdef}
 \end{equation}
 which exists as a bounded operator from $L^2_{\bK,\sigma}$ to $H^2_{\bK,\sigma}$  
 and obtain the following Lippmann\ -\ Schwinger equation, equivalent to the eigenvalue problem \eqref{Phi-evp}:
 \begin{align}
 &\left[\ I\ -\ \left(\mu+1\right)\ T(\eps)\ 
 \ \right]\ \Phi=0,\ \  \Phi\in L^2_{\bK,\sigma}\ .\label{LSeqn}
 \end{align}

We now show that if $\Re\eps<0$, we also obtain an equation of the same type as in \eqref{LSeqn}. In this case,
we observe that  $\eps \left(V-V_{\rm max}\right)\ge0$. Therefore, $-\Delta+\eps \left(V-V_{\rm max}\right)+I\ge I$  and we rewrite \eqref{Phi-evp} as
\begin{equation}
\left(-\Delta+\eps \left(V-V_{\rm max}\right) + I \right)\Phi -\left(\mu+1-\eps V_{\rm max}\right)\Phi=0,\ \ 
\Phi\in L^2_{\bK,\sigma}\ .
\label{evp-K1b}
\end{equation}
If for $\eps<0$ we define $\tilde{T}(\eps)=I-\Delta+\eps \left(V-V_{\rm max}\right)$, then \eqref{Phi-evp}
 is equivalent to 
  \begin{align}
 &\left[\ I\ -\ \left(\mu+1-\eps V_{\rm max}\right)\ \tilde{T}(\eps)\ 
 \ \right]\ \Phi=0,\ \  \Phi\in L^2_{\bK,\sigma}\ .\label{LSeqn-b}
 \end{align}
\smallskip
 
\noindent{\it For the remainder of this section we shall assume $\Re\eps>0$ and work with the form of the eigenvalue problem given  in \eqref{LSeqn}. The analysis below applies with only trivial  modifications to the case $\eps<0$ and the 
 form of the eigenvalue problem given in \eqref{LSeqn-b}.}
\smallskip

For each $\eps>0$, we would like to characterize $L^2_{\bK,\sigma}$- eigenvalues, $\mu(\eps)$, as points where the determinant
of the operator $I - \left(\mu+1\right)T(\eps)$ vanishes. To define the determinant of $I-zT$, one requires that $T$ be trace class.  Although 
$T(\eps)$ is compact on $L^2_{\bK,\sigma}$,  it  is \underline{not} trace class. Indeed, in spatial dimension two,  $\lambda_j$, the $j^{th}$ eigenvalue of  $-\Delta_K+W$ acting in $L^2_{per,\Lambda}$ satisfies the asymptotics $\lambda_j\sim |j|$ (Weyl). Therefore 
\begin{equation}
{\rm trace}(T(\eps))\ =\ \sum_{j}|\lambda_j|^{-1}\ \sim\ \sum_{j}|j|^{-1}\ =\ \infty\ .
\nn\end{equation}

The divergence of the determinant can be removed if we work with the regularized or renormalized determinant; see \cite{Jost-Pais:51,Newton:72}. Note that  $T(\eps)$ is Hilbert Schmidt, {\it i.e.} 
\begin{equation}
\|T\|_{H.S.}^2 = \sum_{j}|\lambda_j|^{-2}\sim
\sum_{j}|j|^{-2}<\infty.\nn
\end{equation}

 For a Hilbert-Schmidt operator, $A$, {\it i.e.} ${\rm tr}(A^2)<\infty$,  define
\begin{equation}
R_2(A)\ \equiv\ \left[I\ +\ A\ \right]e^{-A}-I .\label{R2def}
\end{equation}
Note that $I+A$ is singular if and only if $I+R_2(A)=(I+A)e^{-A}$ is singular. \bigskip
 
Since $e^{-z}=1-z-z^2\int_0^1(s-1)e^{-sz}ds$, we have\\ 
$R_2(z)=(1+z)e^{-z}-1=-z^2\left(1+(1+z)\int_0^1(s-1)e^{-sz}ds\right)$. Therefore 
\begin{equation}
R_2(A) = -A^2\left(I+(I+A)\int_0^1(s-1)e^{-sA}ds\right).
\label{R2fac}\end{equation}
Since $A^2$,  is trace class and the second factor is bounded, $R_2(A)$ is trace class. Therefore the regularized determinant of $I+A$:
\begin{equation}
{\rm det}_{2}(I+A)\ \equiv\ {\rm det}\left( \ I+R_2(A)\ \ \right),
\label{det2}
\end{equation}
is well-defined. With $A=-(\mu+1)T(\eps)$, we have the following 
\cite{Jost-Pais:51,Newton:72,Simon}:

\begin{theorem}\label{eig-det0} 
Let $\sigma$ take on the values $1, \tau$ or $\bar\tau$. 
\begin{enumerate}
\item $\eps\mapsto T(\eps)$ is an analytic mapping from $\{\eps\in\mathbb{C}^1:\Re{\eps}>0\}$ to the space of Hilbert-Schmidt operators on $L^2_{\bK,\sigma}$.
\item For $T(\eps)$, considered as a mapping on $L^2_{\bK,\sigma}$,\  define:
\begin{equation}
\mathcal{E}_\sigma(\mu,\eps)\equiv {\rm det}_{2}\left(I-(\mu+1)T(\eps)\right).
\label{E-sigma-def}
\end{equation}
The mapping $(\mu,\eps)\mapsto \mathcal{E}_\sigma(\mu,\eps)$,  which takes  $(\mu,\eps)\in\mathbb{C}^2$ ($\Re\eps>0$) to $\mathbb{C}$ is  analytic.
\item For $\eps$ real, $\mu$ is an $L^2_{\bK,\sigma}$- eigenvalue of the eigenvalue problem \eqref{Phi-evp}
if and only if 
\begin{equation}
\mathcal{E}_\sigma(\mu,\eps)=0\ .\label{Esigma0}
\end{equation}
\item For $\eps$ real, $\mu$ is an $L^2_{\bK,\sigma}$ eigenvalue of \eqref{Phi-evp} of geometric multiplicity $m$ if and only if $\mu$ is a root of \eqref{Esigma0} of multiplicity $m$.
\end{enumerate}
\end{theorem}

\section{Continuation past a critical $\eps$}\label{sec:continuation}

In section \ref{sec:pfeps-small} we proved Theorem \ref{main-thm} for all $\eps\in(-\eps^0,\eps^0)\setminus\{0\}$, with $\eps^0>0$ sufficiently small, by establishing the following properties: \bigskip

\begin{itemize}
 \item[I.]\ $\mu^\eps(\bK_\star)$ is a simple $L^2_{\bK_\star,\tau}$ eigenvalue of $H^{(\eps)}$ with corresponding 1 - dimensional eigenspace $\mathbb{X}_\tau={\rm span}\{\ \Phi_1^{\mu^\eps}(\bx;\bK_\star)\ \}\subset L^2_{\bK_\star,\tau}$.
 \item[II.]\ $\mu^\eps(\bK_\star)$ is a simple $L^2_{\bK_\star,\bar\tau}$ eigenvalue of $H^{(\eps)}$ with corresponding 1 - dimensional eigenspace $\mathbb{X}_{\bar\tau}=
 {\rm span}\{\ \overline{\Phi_1^{\mu^\eps}(-\bx;\bK_\star)}\ \} \subset L^2_{\bK_\star,\bar\tau}$.
 \item[III.]\ $\mu^\eps(\bK_\star)$ is \underline{not} a $L^2_{\bK_\star,1}$ eigenvalue of $H^{(\eps)}$.\smallskip
 
  \item[IV.]\  We have
   \begin{equation}
\lambda_\sharp^\eps\ \equiv\   \sum_{\bm\in\mathcal{S}} c(\bm,\mu^\eps,\eps)^2\ \left(\begin{array}{c}1\\ i\end{array}\right)\cdot \bK_\star^\bm\ \ne\ 0,
\label{lambda-sharp2}
\end{equation}
where  $c(\bm,\mu^\eps,\eps)$ are Fourier coefficients of $\Phi_1\left[\mu^\eps(\bK_\star),\eps\right](\bx)$, an 
$L^2_{\bK_\star,\tau}$ eigenfunction of $H^{(\eps)}$ with eigenvalue $\mu^\eps=\mu^\eps(\bK_\star)$ ; see Proposition \ref{Fourier-espaces}.
  \end{itemize}
\medskip

We next study the persistence of properties I.-IV. for $\eps$ of arbitrary size. 
\subsection{Continuation strategy}\label{sec:strategy}
Denote by $\mathcal{A}$, the set of all $\eps>0$ for which at least one of the properties I.-IV. fail. With  $\eps^0$ given as above, we clearly have $ \mathcal{A}\subset[\eps^0,\infty)$. The main result of this section is that  
\begin{equation}
\mathcal{A}\ \textrm{  is contained in a countable closed set.}
\label{ACC}
\end{equation}
 Once \eqref{ACC} is shown, we'll have completed the proof of Theorem \ref{main-thm}, our main result. 
\medskip

Our continuation strategy is based on the following general
\begin{lemma}\label{either-or} 
{}

\noindent Let $A\subset(\eps^0,\infty)$ with $\eps^0>0$. Then one of the following assertions holds:\\
(1) $A$ is contained in a  closed countable  set.\\
(2) There exists $\eps_c\in(0,\infty)$ for which the set $A\cap[0,\eps_c)$ is contained in a closed countable set, but 
for any $\eps'>\eps_c$, the set $A\cap[0,\eps')$ is \underline{not} contained in a closed countable set.
\end{lemma}
\medskip

The main work of this section is to prove \eqref{ACC} for $A=\mathcal{A}$ by precluding option (2) of Lemma \ref{either-or}. This suggests introducing the notion of a critical value of $\eps$:
\begin{definition}[Critical $\eps_c$]\label{critical-eps} 
Call a real and positive number $\eps_c$ {\it critical} if there is an increasing sequence $\{\eps_\nu\}$  tending to $\eps_c$ and a corresponding  sequence of geometric multiplicity-two $L^2_{\bK}$- eigenvalues, $\{\mu_\nu\}$, such that 
\begin{itemize} 
\item[(a)]\ properties I.-IV. above, with $\eps$ replaced by $\eps_\nu$ and $\mu^\eps$ 
replaced by $\mu_\nu$, hold for all $\nu=1,2,\dots$, and 
\item[(b)]\ for $\eps=\eps_c$ and $\mu_c=\mu^{\eps_c}\equiv\lim_{\nu\to\infty}\mu_\nu<\infty$ at least one of the properties I.-IV. does not hold.
\end{itemize}
\end{definition}
To prove Lemma \ref{either-or} we use the following:\medskip

\begin{lemma}\label{imp-either-or}
 Let $0=\eps_1<\eps_2<\dots$, and let $\eps_\infty=\lim_{\nu\to\infty}\eps_\nu$. (Perhaps $\eps_\infty=\infty$.) Suppose $A\cap[0,\eps_\nu)$ is contained in a closed countable set $\mathcal{C}_\nu$ for each $\nu\ge1$.
 Then, $A\cap[0,\eps_\infty)$ is contained in a closed countable set $\tilde{\mathcal{C}}$. 
\end{lemma}
\medskip

\noindent First let's use Lemma \ref{imp-either-or} to prove Lemma \ref{either-or}. We then give the proof of  Lemma \ref{imp-either-or}. \medskip

\noindent{\it Proof of Lemma \ref{either-or}:}\\ Let $\eps_c=\sup\{\eps\in(0,\infty): A\cap[0,\eps)\ 
\textrm{ is contained in a  closed countable set.}\}$. Clearly $0<\eps^0\le\eps_c\le\infty$. If $\eps_c=\infty$, then option 
 (1) holds, thanks to Lemma
 \ref{imp-either-or}. And if $\eps_c<\infty$, then by definition, $A\cap[0,\eps')$ is not contained in a closed countable set for any $\eps'>\eps_c$. Again applying Lemma \ref{imp-either-or} shows that $A\cap[0,\eps_c)$ is contained in a countable closed set. In this case, (2) holds and the proof of Lemma \ref{either-or} is complete.\bigskip

\noindent{\it Proof of Lemma \ref{imp-either-or}}\label{lemma-imp-either-or-proof:}
 \noindent  Define
\[ \mathcal{C}=
  \bigcup_{\nu\ge1}\left(\ \mathcal{C}_\nu\cap[\eps_{\nu-1},\eps_\nu]\ \right)\cup\{\eps_\nu:\nu\ge0\}\ 
  \] 
  and set $\tilde{\mathcal{C}}=\mathcal{C}$ if $\eps_\infty=\infty$ and $\tilde{\mathcal{C}}=\mathcal{C}\cup\{\eps_\infty\}$ if $\eps_\infty<\infty$. One checks easily that $A\cap[0,\eps_\infty)\subset\tilde{\mathcal{C}}$, $\tilde{\mathcal{C}}$ is countable, and $\tilde{\mathcal{C}}$ is closed. This completes the proof of Lemma \ref{imp-either-or}.
  \medskip 
  
 We now outline our implementation of the continuation argument. By the discussion of section \ref{sec:det2} and Proposition \ref{prop:2impliescone},  $\eps$ is in $\R\setminus{\cal A}$ provided:\\
 (i) $\mathcal{E}_\tau(\mu^\eps,\eps)=0,\ \partial_\mu\mathcal{E}_\tau(\mu^\eps,\eps)\ne0$, 
  (ii) $\mathcal{E}_1(\mu^\eps,\eps)\ne0$ and  (iii) $\lambda_\sharp^\eps\ne0$.  To continue property (i) past a finite critical value, $\eps_c$, one must show the persistence of a simple zero of  $\mathcal{E}_\tau(\mu,\eps)$ for $\eps>\eps_c$. To continue (ii) and (iii) beyond $\eps_c$ it seems at first natural to introduce the function 
  $\mathcal{E}_1(\mu,\eps)\times \lambda_\sharp\left({\bf c}[\mu,\eps]\right)$, where ${\bf c}[\mu,\eps]$ is the collection of Fourier coefficients of
   the $L^2_{\bK,\tau}$ eigenvector for the eigenvalue $\mu$, and  
  $\lambda_\sharp$ is the expression in \eqref{lambda-sharp2}. Unfortunately the above function is not necessarily analytic; in a neighborhood of $\eps_c$, $\eps\mapsto{\bf c}[\mu^\eps,\eps]$ and therefore
    $\eps\mapsto\lambda_\sharp\left({\bf c}[\mu^\eps,\eps]\right)$ may not vary analytically; see Appendix \ref{sec:topology}. Indeed there is a topological obstruction related to the following observation: along a path of matrices in the space of complex $N\times N$ matrices of rank $N-1$, each matrix has a non-vanishing sub-determinant of dimension $N-1$, although the particular sub-determinant which is non-vanishing changes along the path. The heart of the matter and its remedy are clarified in linear algebra Lemma \ref{lemma2}. That Lemma is applied in section \ref{ham-param} to construct a 
    vector-valued analytic function $F(\mu,\eps)$, whose non-vanishing ensures that $\mathcal{E}_1(\mu,\eps)\ne0$
     as well as the non-degeneracy condition, $\lambda_\sharp^\eps\ne0$. A continuation lemma, Lemma \ref{lemma1}, of section \ref{sec:pickbranch}, is then applied to the pair of analytic functions: $P(\mu,\eps)=\mathcal{E}_\tau(\mu^\eps,\eps),\ 
      F(\mu,\eps)$ to establish the continuation beyond any finite $\eps_c$.
   
%
%
%
%
%
 
\subsection{Picking a branch}\label{sec:pickbranch}
Let 
\begin{equation}
U=\{(\lambda,z)\in\mathbb{C}^2 :\ |\lambda|<\eps_1,\ \ |z|<\eps_2\}
\end{equation}
where $\eps_1$ and $\eps_2$ are given positive numbers. Suppose we are given an analytic function $P:U\to \mathbb{C}$ and an analytic mapping $F:U\to \mathbb{C}^m$. We make the following
\medskip
{\bf Assumptions:}
\begin{itemize}
\item[(A1)]\ If $(\lambda,z)\in U,\ P(\lambda,z)=0$ and $z\in\mathbb{R}$, then $\lambda\in\mathbb{R}$.
\item[(A2)]\ There exists $\{(\lambda_\nu,z_\nu)\}\subset U,\ \ \nu\ge1$ tending to $(0,0)$ as $\nu\to\infty$, such that for each $\nu\ge1$, $z_\nu\in\mathbb{R}\setminus\{0\}$, $P(\lambda_\nu,z_\nu)=0,\ \partial_\lambda P(\lambda_\nu,z_\nu)\ne0,\ F(\lambda_\nu,z_\nu)\ne0$.
\end{itemize}
\medskip

 \begin{remark}\label{rmk-re-centering}
With the above setup, we have centered the analysis about $(z,\lambda)=(0,0)$. We shall apply the results of this
 section to an appropriate analytic function of $(\mu,\eps)$ centered about $(\mu_c,\eps_c)$.
 \end{remark}
\medskip

Under assumptions (A1) and (A2) we will prove the following\medskip

\begin{lemma}\label{lemma1} There exist $\delta>0$ and a real-analytic function $\beta(z)$, defined for $z\in(0,\delta)$, such that for all but at most countably many $z\in (0,\delta)$ we have: 
\begin{equation}
P(\beta(z),z)=0,\ \partial_\lambda P(\beta(z),z)\ne0,\ F(\beta(z),z)\ne0\ .
\label{lemma1-prop}
\end{equation}
Moreover, $\lim_{z\to0^+} \beta(z)=0$.
\end{lemma}
\medskip

\noindent{\it Proof of Lemma \ref{lemma1}:}\ Assumption (A1) implies that $\lambda\mapsto P(\lambda,0)$ is not identically zero. By the Weierstrass Preparation Theorem \cite{Krantz:92}, we may write $P(\lambda,z)=H(\lambda,z)\cdot \tilde{P}(\lambda,z)$ for $(\lambda,z)$ in some polydisc $\tilde{U}\equiv\{|\lambda|<\eps_3,\ |z|<\eps_4\}$, where $\tilde{P}$ is a Weierstrass polynomial (see \eqref{Prep} below) and $H$ is a non-vanishing analytic function.  Assumptions (A1), (A2) hold also for $\tilde{P}, F, \tilde{U}$. Moreover, the conclusion of Lemma \ref{lemma1} for $\tilde{P}, F, \tilde{U}$ implies the conclusion for $P, F,U$. Therefore, it is enough to prove Lemma \ref{lemma1} under the additional assumption that $\tilde{P}$ is a Weierstrass polynomial.
Henceforth we make this assumption. Thus, we have for some $D\ge1$:
\begin{equation}
P(\lambda,z) = \lambda^D+\sum_{j=0}^{D-1}g_j(z)\lambda^j\ =\ 
\prod_{\nu=1}^D\left(\lambda-\alpha_\nu(z)\right),
\label{Prep}
\end{equation}
where $\alpha_1(z),\dots,\alpha_D(z)$ denote the roots of $\lambda\mapsto P(\lambda,z)$ (multiplicity counted), where 
\begin{equation}
\alpha_j(0)=\lim_{z\to0}\alpha_j(z)=0,\ \   j=1,\dots,D\ .\label{alpha_j}
\end{equation}
Moreover, $g_j(z)$ are analytic in $|z|<\eps_4$. 
Note that $D\ge1$, since Assumption (A2) implies $P(0,0)=0$. For $k\ge1$, define
\begin{align}
Q_k(z)\ &=\ \left\{
 \begin{array}{cc}  &D,\ \ {\rm for}\ \ k=1\\
     &\sum_{\nu_1,\dots,\nu_k=1}^D\ \prod_{i,j=1,\ i\ne j}^k \left(\alpha_{\nu_i}(z)-\alpha_{\nu_j}(z)\right)^2,\ \ k\ge2\ .
    \end{array} \right.
 \label{Qkdef}   \end{align}
The right hand side of \eqref{Qkdef} is a symmetric polynomial in $\alpha_1(z),\dots, \alpha_D(z)$ and is therefore a polynomial in the coefficients  $g_j(z)$ of $P(\lambda,z)$ \cite{Herstein:64}, which are analytic in $z$. Consequently, each $Q_k(z)$ is an analytic function of $z$. Moreover, when $z$ is real, the 
$\alpha_\nu(z)$ are also real, and therefore, for $z\in\mathbb{R}$, $Q_k(z)\ne0$ if and only if 
$\lambda\mapsto P(\lambda,z)$ has at least $k$ distinct zeros.  In particular,  for $k\ge D+1$,   $Q_k(z)=0$ for all real $z$, since $\lambda\mapsto P(\lambda,z)$ has only $D$ zeros; see \eqref{Prep}. 
 Hence, there exists $\bar{k}$ with $1\le \bar k\le D$ such that
 \begin{equation}
 Q_{\bar k}(z)\ \ \textrm{is not identically zero, but}\ Q_k(z)\equiv0\ \textrm{for all}\ k>\bar k.
 \nn\end{equation}
 Since $Q_{\bar k}(z)$ is analytic on a neighborhood of $0$ and not identically zero, there is an open interval $(0,\delta_1)$ such that $Q_{\bar k}(z)\ne0$ for all $z\in(0,\delta_1)$. Thus, $P(z,\lambda)$ has at least $\bar k$ distinct zeros, for each $z\in(0,\delta_1)$. On the other hand, $Q_{\bar k+1}(z)\equiv0$ and hence $\lambda\mapsto P(\lambda,z)$, for real $z$,  never has at least $\bar k+1$ distinct zeros. So, $\lambda\mapsto P(\lambda,z)$ has exactly $\bar k$ distinct zeros for $z\in(0,\delta_1)$. We denote these zeros by
 \begin{equation}
 \beta_1(z) <\beta_2(z)<\dots<\beta_{\bar k}(z);\nn
 \end{equation}
 they are real by Assumption (A1). Note that each $\beta_k(z)$ is among the $\alpha_j(z)\ (j=1,\dots,D)$.  Hence,  by \eqref{alpha_j} $\lim_{z\to0^+}\beta_k(z)=0$ for each $k$.
 
 Fix $x\in(0,\delta_1)$, and let $m_1,\dots,m_{\bar k}$ (respectively) be the multiplicities of the zeros 
  $\beta_1(x),\beta_2(x),\dots,\beta_{\bar k}(x)$ of $\lambda\mapsto P(\lambda,z)$; $m_1+\dots+m_{\bar k}=D$.  For $z\in(0,\delta_1)$ close enough to $x$ and for each $j$, there exist $m_j$ zeros of $\lambda\mapsto P(\lambda,z)$ (multiplicities counted) that lie close to $\beta_j(x)$. Unless these  $m_j$ zeros of 
 $\lambda\mapsto P(\lambda,z)$ are all equal, the function $\lambda\mapsto P(\lambda,z)$ would have more than $\bar k$ distinct zeros,  which is impossible. Therefore, for $z\in(0,\delta_1)$ close to $x$, and for each $j$, the polynomial $\lambda\mapsto P(\lambda,z)$ has a single zero $\beta_j(z)$ of multiplicity $m_j$, close to $\beta_j(x)$. In particular, the multiplicities of the zeros $\beta_1(z),\beta_2(z),\dots,\beta_{\bar k}(z)$ are constant as $z\in(0,\delta_1)$ varies over a small enough neighborhood of $x$. Since $x\in(0,\delta_1)$ is arbitrary and since $(0,\delta_1)$ is connected it follows that the multiplicities $m_1,\dots,m_{\bar k}$ (respectively) of the zeros $\beta_1(z) <\beta_2(z)<\dots<\beta_{\bar k}(z)$
  of $\lambda\mapsto P(\lambda,z)$ are constant as $z$ varies over the entire interval $(0,\delta_1)$. Therefore, we have
  \begin{equation}
  P(z,\lambda)\ =\ \prod_{j=1}^{\bar k} (\lambda-\beta_j(z))^{m_j},\ \ \ z\in(0,\delta_1),\ \ |\lambda|<\eps_1
  \label{Prep1}
  \end{equation}
 Here, $\beta_1(z) <\beta_2(z)<\dots<\beta_{\bar k}(z)$ for each $z\in(0,\delta_1)$ and each $m_j$ 
 is a positive integer. 
 Now note that each $\beta_j(z)$ is a real-analytic function on $(0,\delta_1)$, since $\beta_j(z)$ is a simple zero of $\lambda\mapsto\partial_\lambda^{m_j-1}P(\lambda,z)$. 
 \medskip
 
 We now turn to $F:U\to\mathbb{C}^m$. Let us write $F(\lambda,z)=(F_1(\lambda,z),\dots,F_m(\lambda,z))$. For a small positive number $\rho$, to be chosen just below, we define:
 \begin{equation}
 G(z)\ =\ \frac{1}{2\pi i}\ \oint_{|\lambda|=\rho}\ \sum_{l=1}^m\ F_l(\lambda,z)\cdot \overline{F_l(\bar\lambda,\bar z)}\cdot \frac{(\partial_\lambda P(\lambda,z))^2}{P(\lambda,z)}\ \overline{\partial_\lambda P(\bar\lambda,\bar z)}\ d\lambda
 \end{equation}
 We can pick $\rho$ so that $P(\lambda, 0)\ne0$ for $|\lambda|=\rho$. Therefore, for small enough $\eta$, if $|z|<\eta$, we still have $P(\lambda, z)\ne0$ for $|\lambda|=\rho$. Fix such $\rho$ and $\eta$. 
 Then, $G(z)$ is an analytic function of $z$ in the disc $\{|z|<\eta\}$. Moreover a residue calculation shows that 
 \begin{equation}
 G(z)\ =\ \sum_\lambda\ \sum_{l=1}^m\  
 F_l(\lambda,z)\cdot \overline{F_l(\bar\lambda,\bar z)}\cdot \partial_\lambda P(\lambda,z)\cdot
  \overline{\partial_\lambda P(\bar \lambda,\bar z)},
 \end{equation}
 where the sum is over all $\lambda$ in the set:
 \[ \{\ \lambda\ :\ |\lambda|<\rho,\ P(\lambda,z)=0\ \}\ ,\]
with multiplicities included in the sum. In particular, if $z$ is real, then the relevant $\lambda$'s are also real (see (A1)), and therefore 
\begin{equation}
G(z)\ =\ \sum_\lambda\ \sum_{l=1}^m\ |F_l(\lambda,z)|^2\cdot |\partial_\lambda P(\lambda,z)|^2\,
\label{Gsum}\end{equation}
where the sum is over real $\lambda\in(-\rho,\rho)$ such that $P(\lambda,z)=0$. Note that all non-zero contributions to the sum \eqref{Gsum} come from $\lambda$'s that are zeros of $P$ with multiplicity one. Consequently, for real $z$, we have $G(z)\ne0$ if and only if there exists $\lambda\in(-\rho,\rho)$ such that 
 $P(\lambda,z)=0,\ \partial_\lambda P(\lambda,z)\ne0$ and $F(\lambda,z)\ne0$.
 
 Therefore, assumption (A2) tells us that the analytic function $G(z)$ doesn't vanish identically in $\{|z|<\eta\}$. It follows that we can pick a positive $\delta$, less than ${\rm min}(\eta,\delta_1)$, such that 
 \[ G(z)\ne0\ \ {\rm for}\ \ 0<|z|<\delta\ .\] 
 
 Now suppose $z\in(0,\delta)$. Then, there exists $\lambda\in(-\rho,\rho)$ such that $P(\lambda,z)=0,\  \partial_\lambda P(\lambda,z)\ne0, F(\lambda,z)\ne0$. This $\lambda$ must be equal to one of the $\beta_j(z),\ \ j=1,\dots,\bar k$,  for which $m_j=1$. So, for each $z\in(0,\delta)$ there exists $j$ such that $m_j=1$ and $F(\beta_j(z),z)\ne0$.   
 
 Unfortunately, the above $j$ may depend on $z$. However, we may simply fix some $x_0\in(0,\delta)$, and pick $j_0$ such that $m_{j_0}=1$ and $F(\beta_{j_0}(x_0),x_0)\ne0$. The function $z\mapsto\beta_{j_0}(z)$ is a real-analytic function of $z\in(0,\delta)$. Moreover, we know that the real analytic function $z\mapsto F(\beta_{j_0}(z),z)$ is not identically zero on $(0,\delta)$, since it is nonzero for $z=x_0$. So, it can vanish only on a set of discrete points which accumulates at $0$ or at $\delta$. The proof of Lemma  \ref{lemma1} is complete.
 
 \begin{remark} We have proven more than asserted in Lemma \ref{lemma1}. In fact, $P(\beta(z),z)=0$
  and $\partial_\lambda P(\beta(z),z)\ne0$ for {\it all} $z\in(0,\delta)$; and $F(\beta(z),z)\ne0$ for all $z\in(0,\delta)$ except perhaps for countably many $z$ tending to $0$. Note that we can arrange for this countable sequence not to accumulate at $\delta$ by simply taking $\delta$ to be slightly smaller.
 \end{remark}
  
 \subsection{Linear Algebra}\label{sec:linearalgebra}
 
 Given an $N\times N$ (complex) matrix $A$ of rank $N-1$, we would like to produce a nonzero vector in the nullspace of $A$, depending analytically on the entries of $A$. In general there is a topological obstruction to this; 
 see  Appendix \ref{sec:topology}.  However, the following  result will be enough for our purposes.

 Fix $N\ge1$. Let ${\rm Mat}(N)$ be the space of all complex $N\times N$ matrices. We denote an $N\times N$ matrix by $A\in {\rm Mat}(N)$. We say that a map $\Gamma:{\rm Mat}(N)\to\mathbb{C}^N$ is a {\it polynomial map} if the components of $\Gamma(A)$ are polynomials in the entries of $A$. 
  Polynomial maps are therefore analytic in the entries of $A$.
 
 In this section we prove 
 \begin{lemma}\label{lemma2}
 There exist polynomial maps $\Gamma_{jk}:{\rm Mat}(N)\to\mathbb{C}^N$, where $j,k=1,\dots, N$, with the following property:\\
 Let $A\in {\rm Mat}(N)$ have rank $N-1$. Then all the vectors $\Gamma_{jk}(A)$ belong to the nullspace of $A$, and at least one of these vectors is non-zero.
 \end{lemma}\medskip
 
 \noindent{\it Proof of Lemma \ref{lemma2}:}\ We begin by setting up some notation. Given $A\in {\rm Mat}(N)$, we write $A^{(j,k)}$ to denote the matrix obtained from $A$ by deleting row $j$ and column $k$. We write ${\rm col}(A,k)$ to denote the $k^{th}$ column of $A$. If $v=\left(v_1,\dots,v_N\right)^t\in\mathbb{C}^N$ is a column vector, then we write $v_j$ to denote the $j^{th}$ coordinate of $v$, and write $\hat{v}^{(k)}$ to denote the column vector obtained from $v$ by deleting the $k^{th}$ coordinate.
 Thus, $\hat{v}^{(k)}\in\mathbb{C}^{N-1}$.
 
 From linear algebra, we recall that 
 \begin{align}
& \textrm{For any}\ A\in{\rm Mat}(N)\ \textrm{of rank}\ N-1,\ \textrm{we have}\ \det\left(A^{(j,k)}\right)\ne0\
  \textrm{for some}\ (j,k)\ . \label{linalg1}\\
&  \textrm{For any}\ A\in{\rm Mat}(N),\ \textrm{and any}\ v\in{\rm Nullspace}(A),\ \textrm{we have}\nn\\
 &\qquad\qquad\  A^{(j,k)}\hat{v}^{(k)}\ =\ -\left[{\rm col}(A,k)\right]\hat{\ }^{\  (j)}\ v_k.
  \label{kerA-Ajk}
  \end{align}
(Equation \eqref{kerA-Ajk} expresses the fact that $(Av)_i=0$, for all $i\ne j$.)

If ${\rm rank}(A)=N-1$ and $\det A^{(j,k)}\ne0$, then the space of solutions of \eqref{kerA-Ajk}, and the nullspace of $A$ are one-dimensional; hence, in this case the nullspace of $A$ consists precisely of the solutions of \eqref{kerA-Ajk}. \medskip

We now define $\Gamma_{jk}(A)\in\mathbb{C}^N$ to be the element, $v$,  in the nullspace of $A$, whose $N$ components are constructed as follows:
\begin{align}
&\textrm{Set}\ v_k=\left[\det A^{(j,k)}\right]^2\label{vk-def}\\
& \textrm{and obtain the other}\  N-1\ \textrm{entries comprising the vector}\ \hat{v}^{(k)}\ \textrm{by solving}\nn\\
& A^{(j,k)}\ \hat{v}^{(k)}\ =\ -\left[{\rm col}(A,k)\right]\hat{\ }^{\  (j)}\ \left[\det A^{(j,k)}\right]^2\ .
 \label{Gamma-jk-def}
 \end{align}
If $\det A^{(j,k)}\ne0$, we can solve  \eqref{Gamma-jk-def}  uniquely for $\hat{v}^{(k)}$ by Cramer's rule and together with \eqref{vk-def} construct $v$. 

Note that each component of the vector $\Gamma_{jk}(A)$ has the form 
\[ \det A^{(j,k)}\ \times\ \textrm{Polynomial in the entries of}\ A\ .\]
In particular, $A\mapsto\Gamma_{jk}(A)$ is a polynomial map. Furthermore, $\Gamma_{jk}(A)=0$ if
 $\det A^{(j,k)}=0$. Also, $\Gamma_{jk}(A)\ne 0$ if $\det A^{(j,k)}\ne0$, since the $k^{th}$ coordinate of 
 $ \Gamma_{jk}(A)$ is $\left[\det A^{(j,k)}\right]^2$.\medskip
 
If ${\rm rank}(A)=N-1$, then some $\Gamma_{jk}(A)$ is non-zero, thanks to \eqref{linalg1}. Moreover, each $\Gamma_{jk}(A)$ always belongs to the nullspace of $A$. Indeed, fix $j,k$; if 
 $\det A^{(j,k)}=0$ then $\Gamma_{jk}(A)=0\in {\rm Nullspace(A)}$. If instead $\det A^{(j,k)}\ne0$,
  then ${\rm Nullspace(A)}$ consists of the solutions of \eqref{kerA-Ajk}; and we defined 
  $\Gamma_{jk}$ to solve \eqref{kerA-Ajk}. Thus, in all cases we have $\Gamma_{jk}(A)\in {\rm Nullspace(A)}$ if ${\rm rank}(A)=N-1$. The proof of  Lemma \ref{lemma2} is complete.\bigskip

  \subsection{Hamiltonians Depending on Parameters}\label{ham-param}
  
  Recall the operator $H^{(\eps)}=-\Delta+\eps V(\bx)$.
In this section we complete the continuation argument (and the proof of Theorem \ref{main-thm}) by showing how to continue Properties I.\ -\ IV. , listed at the beginning of section \ref{sec:strategy}, 
 beyond any critical value, $\eps_c$, where one of these properties may fail.
This argument is based on appropriate application of Lemma \ref{lemma1} and Lemma \ref{lemma2}.\bigskip

Let $\eps_c$ and $\mu_c$ be as in Definition \ref{critical-eps}.
\medskip

 Without loss of generality we can assume $\bK_\star=\bK$.    
 We work in the Hilbert spaces
   \begin{align}
   L^2_{\bK,\tau} &=\ \left\{ \sum_{\bm\in\mathcal{S}}  c(\bm) \left[e^{i\bK^\bm\cdot\bx}+\bar\tau e^{iR\bK^\bm\cdot\bx}+\tau e^{i R^2\bK^\bm\cdot\bx} \right]:\ \sum_{\bm\in\mathcal{S}}|c(\bm)|^2<\infty\right\}
\label{L2Ktau}   \\
   H^s_{\bK,\tau} &=\ \left\{\sum_{\bm\in\mathcal{S}} c(\bm)\left[e^{i\bK^\bm\cdot\bx}+ \bar\tau e^{iR\bK^\bm\cdot\bx}+\tau e^{i R^2\bK^\bm\cdot\bx} \right]\right.:\nn\\
   &\qquad\qquad\qquad\qquad \qquad\qquad\qquad\qquad
     \left.\sum_{\bm\in\mathcal{S}}(1+|\bm|^2)^s|c(\bm)|^2<\infty\right\}
  \label{HsKtau} \end{align}

{\it We will apply the results of section \ref{sec:pickbranch} and section \ref{sec:linearalgebra},
 with the analysis centered at $(\mu_c,\eps_c)$ rather than at $(0,0)$; see Remark \ref{rmk-re-centering}. }
 We shall use that  
 \begin{equation} 
 H^{(\eps)}:H^2_{\bK,\sigma}\to L^2_{\bK,\sigma}\ \ \textrm{ is self-adjoint for $\eps$ real\ .}\ 
 \nn\end{equation} 
 Let $M$ be a positive integer, chosen below to be sufficiently large.  We regard $L^2_{\bK,\tau}$ as the direct sum $L^2_{\rm lo}\oplus L^2_{\rm hi}$, where $L^2_{\rm lo}$ consists of Fourier series as in \eqref{L2Ktau} such that the $c(\bm)=0$ whenever $|\bm|>M$, and $L^2_{\rm hi}$ consists of Fourier series as in \eqref{L2Ktau} such that the $c(\bm)=0$ whenever $|\bm|\le M$. Similarly, we regard $H^s_{\bK,\tau}$
    as the direct sum $H^s_{\rm lo}\oplus H^s_{\rm hi}$ using \eqref{HsKtau}. We set 
   $N= {\rm dim}( L^2_{\rm lo})$.\medskip
   
   Let $\Pi_{\rm lo}$ and $\Pi_{\rm hi}$ be the projections that map a Fourier series as in \eqref{L2Ktau}
    or \eqref{HsKtau} to the truncated Fourier series obtained by setting all the $c(\bm)$ with $|\bm|>M$, or with $|\bm|\le M$, respectively, equal to zero. We may view $H^{(\eps)}$ as the mapping
    \begin{align*}
&   \left(\begin{array}{c} \psi_{\rm hi} \\ \psi_{\rm lo}\end{array}\right)\ \mapsto\ \left(\begin{array}{cc} A^{(\eps)} & B^{(\eps)}\\ C^{(\eps)} & D^{(\eps)}\end{array}\right)\ \left(\begin{array}{c} \psi_{\rm hi} \\ \psi_{\rm lo}\end{array}\right)
 \\ & {\rm with}\ A^{(\eps)}\ =\ \Pi_{\rm hi}\ H^{(\eps)}\ \Pi_{\rm hi},\ \ \ 
B^{(\eps)}=\Pi_{\rm hi}\ H^{(\eps)}\ \Pi_{\rm lo},\\
&  \ \ C^{(\eps)}=\Pi_{\rm lo}\ H^{(\eps)}\ \Pi_{\rm hi},\ \ \ {\rm and}\ \ \ 
 D^{(\eps)}=\Pi_{\rm lo}\ H^{(\eps)}\ \Pi_{\rm lo}\ \ \ .
 \end{align*}
 
By choosing the frequency cutoff, $M$, to be sufficiently large, we have 
 \begin{align}
&  A^{(\eps_c)}-\mu_c I:H^2_{\rm hi}\to L^2_{\rm hi}\ \textrm{has a bounded inverse; say}\nn\\
 &  \left\| \left(\ A^{(\eps_c)}\ -\ \mu_c\ I\right)^{-1}\ \right\|_{L^2_{\rm hi}\to H^2_{\rm hi}}\ \le\ C\ .
 \label{Ahighbound}
 \end{align}
 Therefore, for all $(\mu,\eps)$ in some fixed small neighborhood of $(\mu_c,\eps_c)$ we have
 \begin{equation}
 \left\| \left(\ A^{(\eps)}\ -\ \mu\ I\ \right)^{-1}\ \right\|_{L^2_{\rm hi}\to H^2_{\rm hi}}\ \le\ C'\ .
 \label{Aminusmu-highbound}\end{equation}
 
 The  eigenvalue problem 
 \begin{equation}
 H^{(\eps)}\psi= \mu\ \psi\  {\rm for}\ \psi=\left(\begin{array}{c}\psi_{\rm hi} \\ \psi_{\rm lo}\end{array}\right)\in H^2_{\rm lo}\oplus H^2_{\rm hi}
 \label{evp-vec}
 \end{equation}
 is equivalent to  the system
\[   A^{(\eps)}\psi_{\rm hi}\ +\ B^{(\eps)}\psi_{\rm lo}\ \ =\ \mu\ \psi_{\rm hi},\ \ 
 C^{(\eps)}\psi_{\rm hi}\ +\ D^{(\eps)}\psi_{\rm lo}\ \ =\ \mu\ \psi_{\rm lo}\ .\]
 That is,
 \begin{align}
& \psi_{\rm hi}\ = \ -\left( A^{(\eps)}-\mu I\right)^{-1}\ B^{(\eps)}\ \psi_{\rm lo}, \ {\rm and}\label{psi-hi}\\
&\left[ -C^{(\eps)}\left(A^{(\eps)}-\mu I\right)^{-1}B^{(\eps)}\ +\ \left(D^{(\eps)} -\ \mu I\right)\ \right]\psi_{\rm lo}\ =\ 0,\label{psi-low-eq}
\end{align}
where we regard $A^{(\eps)}-\mu I$ as an operator from $H^2_{\rm hi}$ to $L^2_{\rm hi}$. Note also that $B^{(\eps)}\psi_{\rm lo}\in L^2$ since $\psi_{\rm lo}\in H^2_{\bK,\tau}$; hence 
$\left( A^{(\eps)}-\mu I\right)^{-1}\ B^{(\eps)}\ \psi_{\rm lo}\ \in H^2_{\rm hi}$, thanks to
 \eqref{Aminusmu-highbound}, which holds under our  assumption that $(\mu,\eps)$ is near $(\mu_c,\eps_c)$. It follows that 
  \[C^{(\eps)}\left( A^{(\eps)}-\mu I\right)^{-1}\ B^{(\eps)}\ \psi_{\rm lo}\in L^2_{\rm lo}\]
   by the definition of $C^{(\eps)}$.
  
\begin{equation}
\textrm{  The operator in square brackets in \eqref{psi-low-eq} will be denoted as}\ \mathcal{D}(\mu,\eps)\ .
\label{calDdef}\end{equation}
Thus, $\mathcal{D}(\mu,\eps)$ is analytic in $(\mu,\eps)$, where $(\mu,\eps)$ varies over a small disc about 
 $(\mu_c,\eps_c)$ in $\mathbb{C}^2$. We may regard $\mathcal{D}(\mu,\eps)$ as an $N\times N$ matrix. Thus, 
 \begin{align}
& \psi=\left(\begin{array}{c}\psi_{\rm lo}\\ \psi_{\rm hi}\end{array}\right)\ \textrm{is an eigenfunction of}\ H^{(\eps)}\ \textrm{with eigenvalue}\ \mu \nn\\
&\textrm{ if and only if}\ \textrm{\eqref{psi-hi} holds and}\ \psi_{\rm lo}\ \textrm{ is a non-trivial solution of}\ \  \mathcal{D}(\mu,\eps)\psi_{\rm lo}=0 .\ \ \label{Dpsilo0}
 \end{align}
 It follows that $\mu$ is simple, {\it i.e.} a  multiplicity one eigenvalue of $H^{(\eps)}$ if and only if the $N\times N$ matrix $ \mathcal{D}(\mu,\eps)$ has rank $N-1$.\bigskip
 
 We shall now apply Lemma \ref{lemma2} to $ \mathcal{D}(\mu,\eps)\in{\rm Mat}(N)$, the space of $N\times N$ complex matrices. Let $\Gamma_{jk}$ denote the polynomial map given by Lemma \ref{lemma2}. For $j,k=1,\dots, N$ we define
 \[ \psi_{\rm lo}^{jk}(\mu,\eps)\ =\ \Gamma_{jk}\left(\mathcal{D}(\mu,\eps)\right)\]
 and set
  \[ \psi_{\rm hi}^{jk}(\mu,\eps)\ =\ -\left(A^{(\eps)}-\mu I\right)^{-1}B^{(\eps)}\psi_{\rm lo}^{jk}\]
  as in \eqref{psi-hi}.
 
 By  Lemma \ref{lemma2}, \eqref{psi-hi} and \eqref{psi-low-eq}
 we now know the following for 
 \begin{equation*}
  \psi^{jk}(\mu,\eps)\ =\ \left(\begin{array}{c} \psi_{\rm lo}^{jk}(\mu,\eps) \\{ } \\  \psi_{\rm hi}^{jk}(\mu,\eps)\end{array}\right)\ : \end{equation*}
 
 \begin{align}
(A)\ \ \ & \psi^{jk}(\mu,\eps)\in H^2_{\bK,\tau}\ \textrm{depends analytically on}\ (\mu,\eps),\nn\\
& \textrm{for}\ 
 (\mu,\eps)\ \textrm{in a small neighborhood of}\ (\mu_c,\eps_c) \label{prprty1}\\
 &\nn\\
(B)\ \ \ & \textrm{If}\ \mu\ \textrm{is a simple eigenvalue of}\ H^{(\eps)},\ (\mu-\mu_c,\eps-\eps_c\ \textrm{small}),\nn\\
 &\textrm{then all}\ \psi^{jk}(\mu,\eps)\ \textrm{are in the nullspace of}\ H^{(\eps)}-\mu I\ .\nn\\
 &\textrm{Furthermore, at least one of the}\  \psi^{jk}(\mu,\eps)\ \textrm{is non-zero }\nn\\
&\textrm{and is therefore an eigenfunction of}\ H^{(\eps)}.
  \label{prprty2}
 \end{align}
 
 Let us write out the Fourier expansions of the $\psi^{jk}(\mu,\eps)$. We have
 \begin{equation}
 \left[\psi^{jk}(\mu,\eps)\right](\bx)\ =\
  \sum_{\bm\in\mathcal{S}}c^{jk}(\bm,\mu,\eps)\left[\ e^{i\bK^\bm\cdot\bx}+\bar\tau e^{iR\bK^\bm\cdot\bx}
   +\tau e^{iR^2\bK^\bm\cdot\bx}\ \right]
   \label{FSpsijk}
   \end{equation}
   The coefficients $c^{jk}(\bm,\mu,\eps)$ depend analytically on $(\mu,\eps)\in U$, where $U$ is a small neighborhood of $(\mu_c,\eps_c)$, which is independent of $\bm$. Moreover, since $\psi^{jk}(\mu,\eps)$ is an analytic $H^2_{\bK,\tau}$- valued function, it follows that
  \begin{equation}
  \sum_{\bm\in\mathcal{S}} \left(1+|\bm|^2\right)^2 \left|c^{jk}(\bm,\mu,\eps)\right|^2\ \ \textrm{is bounded as}\ \   (\mu,\eps)\ \textrm{varies over}\ U\ .
  \label{H2var}\end{equation}
  (Perhaps we must shrink $U$ to achieve \eqref{H2var}.)\bigskip
  
  With a view toward continuation of the Properties I.-IV.  (enumerated at the start of section \ref{sec:continuation}) as 
   $\eps$ traverses any critical value (Definition \ref{critical-eps}), $\eps_c$ 
    we state the following \medskip
  
\begin{lemma}\label{lemma3}  
 Suppose  there exists a sequence of eigenvalues $(\mu_\nu,\eps_\nu)\to(\mu_c,\eps_c)$  with  $0<\eps_\nu<\eps_c$, such  that  for each $\nu$ the following properties ($\mathcal{A}1$)-($\mathcal{A}4$) hold:
 \begin{itemize}
 \item[($\mathcal{A}1$)]\ $\mu_\nu$ is a simple eigenvalue of $H^{(\eps_\nu)}$ on $L^2_{\bK,\tau}$, with eigenfunction 
\[\Psi_\nu(\bx)\ =\ \sum_{\bm\in\mathcal{S}} c_\nu(\bm)\ 
 \left[\ e^{i\bK^\bm\cdot\bx}\ +\ \bar\tau e^{iR\bK^\bm\cdot\bx}\ +\ \tau e^{iR^2\bK^\bm\cdot\bx}\ \right]\in H^2_{\bK,\tau}\ . \]
  \item[($\mathcal{A}2$)]\ $\mu_\nu$ is a simple eigenvalue of $H^{(\eps_\nu)}$ on $L^2_{\bK,\bar\tau}$, with eigenfunction 
\[\overline{\Psi_\nu(-\bx)}\ =\ \sum_{\bm\in\mathcal{S}} \overline{c_\nu(\bm)}\ 
 \left[\ e^{i\bK^\bm\cdot\bx}\ +\ \tau e^{iR\bK^\bm\cdot\bx}\ +\ \bar\tau e^{iR^2\bK^\bm\cdot\bx}\ \right]\in H^2_{\bK,\tau}\ .\]
\item[($\mathcal{A}3$)]\  $\mathcal{E}_1(\mu_\nu,\eps_\nu)\ne0$, {\it i.e.} $\mu_\nu$ is \underline{not} a $L^2_{\bK,1}$ eigenvalue of $H^{(\eps_\nu)}$. 
\item[($\mathcal{A}4$)]\  The following non-degeneracy condition ($\lambda_\sharp^\eps\ne0$) holds:
 \begin{equation} \sum_{\bm\in\mathcal{S}} w(\bm)\ \left[c_\nu(\bm)\right]^2\ \ne\ 0,
\label{lambda-eps-ne0}
 \end{equation}
 where $\{w(\bm)\}_{\bm\in\mathcal{S}}$ are fixed weights, such that
  \begin{equation}
  \left| w(\bm) \right|\ \le\ C\left(1+|\bm|\right),\ \ \bm\in\mathcal{S}\ .
  \label{wtbnd}
  \end{equation}
  Our choice of weights (see  \eqref{lambda-sharp1}) is:
  \begin{equation}
  w(\bm) = \left(\begin{array}{c}1\\ i\end{array}\right)\cdot \bK^\bm\ . \nn
  \end{equation}
 \end{itemize}
 \medskip
 
 Then, there exist a (non-empty) open interval $\mathcal{I}=(\eps_c,\eps_c+\delta)$, a real-valued real-analytic function $\beta(\eps)$ defined on $I$, a function $\varphi^\eps\in L^2_{\bK,\tau}$ depending on the parameter $\eps\in\mathcal{I}$, and a countable subset $\mathcal{C}\subset\mathcal{I}$, such that the following hold:
 \begin{itemize}
 \item[(i)]\ $\left(-\Delta+\eps V_h\right)\varphi^{(\eps)}=\beta(\eps)\varphi^{(\eps)}$ for each $\eps\in\mathcal{I}$.
 \item[(ii)]\ $\lim_{\eps\to\eps_c^+} \beta(\eps) = \mu_c$.
 \item[(iii)]\ $\mathcal{C}$ has no accumulation points in $\mathcal{I}$, although $\eps_c$ may be an accumulation point of $\mathcal{C}$.
 \item[(iv)]\ For each $\eps$ in $\mathcal{I}\setminus\mathcal{C}$
 \subitem\ (a)\ $\beta(\eps)$ is a \underline{simple} eigenvalue of $-\Delta+\eps V_h$ on $L^2_{\bK,\tau}$,
 \subitem\ (b)\ $\beta(\eps)$ is not an  eigenvalue of $-\Delta+\eps V_h$ on $L^2_{\bK,1}$, and 
 \subitem\ (c)\ the quantity $\lambda_\sharp^\eps$, arising from the eigenfunction $\varphi^{(\eps)}$ via formula \eqref{lambda-sharp1} (with $\Phi_1$ replaced by $\varphi^{(\eps)})$ is non-zero.
 \end{itemize}
\end{lemma}
\medskip

\noindent{\it Proof of Lemma \ref{lemma3}:}\   Recall that the zeros, $\mu$, of the renormalized determinant, $\mathcal{E}_1(\mu,\eps)$, defined in  section \ref{det2}, are precisely the set of $L^2_{\bK,1}$ eigenvalues of $H^{(\eps)}$.  
Thus, tracking the set of $(\mu,\eps)$ such that Assumptions ($\mathcal{A}1$)-($\mathcal{A}4$) and in particular ($\mathcal{A}3$) suggests that we introduce, 
for  $(\mu,\eps)\in U$, the matrix-valued  function:
  \begin{equation}
  F^{jk}(\mu,\eps)\ \equiv\ 
  \left(\ \sum_{\bm\in\mathcal{S}}\ w(\bm)\ \left[c^{jk}(\bm,\mu,\eps)\right]^2\ \right)\ \times\  
  \mathcal{E}_1(\mu,\eps)  .
   \label{Fdef}
   \end{equation}
   $F^{jk}(\mu,\eps)$ is an analytic function on $U$.  We define
   \begin{equation}
   F(\mu,\eps)\ \equiv\ \left(\  F^{jk}(\mu,\eps)\ \right)_{j,k=1,\dots,N}
   \label{F-def}
   \end{equation}
   Thus, $F:U\to\mathbb{C}^{N^2}$ is an analytic map.  
   
 Now for each $\nu$, \eqref{prprty2} applies to $(\mu_\nu,\eps_\nu)$ , since $\mu_\nu$ is a simple eigenvalue. Thus, for some $jk$, the function $\psi^{jk}(\mu_\nu,\eps_\nu)$ is a non-zero null-vector
  of $H^{(\eps_\nu)}-\mu_\nu I$, {\it i.e.} an eigenfunction of $H^{(\eps_\nu)}$. Since by hypothesis  $\psi_\nu$  is an eigenfunction of $H^{(\eps_\nu)}$ satisfying \eqref{lambda-eps-ne0} with eigenvalue $\mu_\nu$ and since $\mu_\nu$ is a simple eigenvalue of $H^{(\eps_\nu)}$, the corresponding eigenfunction $\Psi_\nu$ satisfies:
  \begin{equation*}
  \Psi_\nu\ =\ \gamma_\nu\ \psi^{jk}(\mu_\nu,\eps_\nu)
  \ \ \textrm{ for a  complex constant}\ \gamma_\nu\ne0 \ .
  \end{equation*}
  
  Therefore, 
  \begin{align*}
  0\ne\sum_{\bm\in\mathcal{S}}w(\bm) \left[c_\nu(\bm)\right]^2 &=
  \gamma_\nu^2\sum_{\bm\in\mathcal{S}}w(\bm)\left[c^{jk}(\bm;\mu_\nu,\eps_\nu)\right]^2\nn\\
  &= \gamma_\nu^2\ \frac{F^{jk}(\mu_\nu,\eps_\nu)}{\mathcal{E}_1(\mu_\nu,\eps_\nu)};\ \ \textrm{see \eqref{Fdef} and \eqref{lambda-eps-ne0}}.
  \end{align*}
  The second equality holds since $\mathcal{E}_1(\mu_\nu,\eps_\nu)\ne0$; see hypothesis ($\mathcal{A}3$).
   It follows that $F^{jk}(\mu_\nu,\eps_\nu)\ne0$ for some $jk$, {\it i.e.}
  \begin{equation}
  F(\mu_\nu,\eps_\nu)\ne0,\ \ \ \textrm{for each}\ \ \nu\ ,
  \label{Fne0}
  \end{equation}
  \bigskip
  where $\{\eps_\nu\}$ is a sequence tending to $\eps_c$ from below.
  
 We complete the proof of Lemma \ref{lemma3}  by application of Lemmata \ref{lemma1} and \ref{lemma2} for appropriate choices 
 of  $P(\mu,\eps)$ and $F(\mu,\eps)$. 
  Let $\mathcal{E}_\tau(\mu,\eps)$, denote the renormalized determinant \eqref{E-sigma-def}.
Let  $P(\mu,\eps)=\mathcal{E}_\tau(\mu,\eps)$ and $F(\mu,\eps)$ be given by \eqref{Fdef}, \eqref{F-def}.  We now check the hypotheses of Lemma \ref{lemma1}.  First note
  \begin{align}
 & P(\mu,\eps)=0\ \ \textrm{ if and only if}\ \mu\ \textrm{is an $L^2_{\bK,\tau}$ eigenvalue of}\ H^{(\eps)},\ \textrm{and the multiplicity of}\nn\\
 & \textrm{$\mu$ as a zero of $P(\mu,\eps)$ is equal to its multiplicity as an eigenvalue of $H^{(\eps)}$.}\label{zerocount}
  \end{align}

Because $H^{(\eps)}$ is self-adjoint for real $\eps$, we see from \eqref{zerocount} that 

\begin{align}
\textrm{if}\ (\mu,\eps)\in U,\ \eps\ \ \textrm{is real},\ \textrm{and} \ P(\mu,\eps)=0,\  
\textrm{then}\  \mu\in\mathbb{R}.
\label{lamreal}\end{align}

Moreover, from \eqref{lambda-eps-ne0}, \eqref{Fne0} and \eqref{zerocount} we see that
\begin{align}
&(\mu_\nu,\eps_\nu)\in U\ \textrm{for each}\ \nu\ge1,\ \ (\mu_\nu,\eps_\nu)\to(\mu_c,\eps_c),\ {\rm as}\
 \nu\to\infty;\\
&\qquad\qquad\qquad \textrm{and for each}\ \nu,\ \textrm{we have}\nn\\
&P(\mu_\nu,\eps_\nu)=0,\ \partial_\mu P(\mu_\nu,\eps_\nu)\ne0,\ F(\mu_\nu,\eps_\nu)\ne0,\ \ \eps_\nu\in\mathbb{R},\ \ 0<\eps_0\le\eps_\nu\ . \label{allisnowset}
\end{align}

Recall that $F:U\to\mathbb{C}^{N^2}$ is an analytic mapping and  $P(\mu,\eps):U\to\mathbb{C}$ is analytic. Results \eqref{lamreal}-\eqref{allisnowset} tell us that conditions (A1) and (A2) of 
section \ref{sec:pickbranch} hold for our present choice of $F$ and $P$. Therefore, Lemma \ref{lemma1} applies; see also the remark immediately after its proof. Thus, we obtain a positive number $\delta$ and a real-analytic function $\beta(\eps)$ such that the following  holds:
\begin{align}
&\textrm{For each}\ \eps\in (\eps_c,\eps_c+\delta), \textrm{we have}\ P(\beta(\eps),\eps)=0\ \ \textrm{and}\ \ \partial_\lambda P(\beta(\eps),\eps)\ne0. \label{simplezerobeyond} \\
&\textrm{Moreover, for all but countably many}\ \eps\in(\eps_c,\eps_c+\delta),\nn\\ 
&\textrm{with their only possible accumulation point at $\eps_c$, we have}\ F(\beta(\eps),\eps)\ne0.\label{Fne0-beyond-0}
\end{align}

By \eqref{simplezerobeyond} and \eqref{zerocount}, we have 
\begin{equation}
\textrm{For each $\eps$ in $(\eps_c,\eps_c+\delta)$, the number $\beta(\eps)$ is a simple $L^2_{\bK,\tau}$ eigenvalue of}\ H^{(\eps)}\ .
\label{beta-simpleeig}
\end{equation}

Therefore, from \eqref{prprty2} we have that
\begin{equation}
\textrm{For each $\eps$ in $(\eps_c,\eps_c+\delta)$, all the $\psi^{jk}(\mu,\eps)$ are in the nullspace of 
 $H^{(\eps)}-\beta(\eps) I$.}\label{simple-eig}
 \end{equation} 
 
 Recalling \eqref{FSpsijk} and \eqref{Fdef},  \eqref{F-def}, \eqref{Fne0-beyond-0}, we now see that
 \begin{align}
& \textrm{For each $\eps$ in $(\eps_c,\eps_c+\delta)$ outside a countable set}\\
&\textrm{ with its only possible accumulation point at $\eps_c$, $\mathcal{E}_1(\beta(\eps),\eps)\ne0$ } \nn\\
 &\textrm{and there exists $jk$ such that}\qquad  \sum_{\bm\in\mathcal{S}} w(\bm)\ \left[c^{jk}(\bm,\beta(\eps),\eps)\right]^2\ \ne\ 0\ .
 \label{lambdane0-jk}
 \end{align}
 
 Now unfortunately the pair $(j,k)$ in \eqref{lambdane0-jk} may depend on $\eps$. However, 
 \eqref{lambdane0-jk} implies that for some fixed $(j,k)=(j_0,k_0)$,  the function
 \begin{equation}
 \eps\mapsto \sum_{\bm\in\mathcal{S}} w(\bm)\ \left[c^{j_0k_0}(\bm,\beta(\eps),\eps)\right]^2
 \label{lambdaeps}
 \end{equation}
 defined for $\eps\in(\eps_c,\eps_c+\delta)$ is not identically zero. Since this function is analytic in $\eps$, it is equal to zero at most at countably many $\eps$. Moreover, the zeros of the function \eqref{lambdaeps} in $(\eps_c,\eps_c+\delta)$ can accumulate only at $\eps_c$ and at $\eps_c+\delta$. By taking $\delta$ smaller, we may assume that the zeros of the function \eqref{lambdaeps} can only accumulate at $\eps_c$. 
 
 We now set, for $ \eps\in(\eps_c,\eps_c+\delta)$:
 \begin{align}
 \Psi^{(\eps)}\ &\equiv\ \Psi^{j_0,k_0}(\beta(\eps),\eps)\in H^2_{\bK,\tau}\ \ {\rm and}\nn\\
 \Psi^{(\eps)}(\bx)\ &=\ 
 \sum_{\bm\in\mathcal{S}} c^{j_0,k_0}(\bm;\beta(\eps),\eps)\ \left[e^{i\bK^\bm\cdot\bx} + \bar\tau e^{iR\bK^m\cdot\bx}+\tau e^{iR^2\bK^\bm\cdot\bx}\right].
 \nn\end{align}
 
We now have that the function $\mu=\beta(\eps)$ satisfies Properties I.-IV.  
 for all $\eps\in(\eps_c,\eps_c+\delta)$ except possibly along a sequence of ``bad'' 
 $\eps$'s which tends to $\eps_c$.  Properties I.-III.,  that $\beta(\eps)$ is a eigenvalue in each of the subspaces $L^2_{\bK,\tau}$ and $L^2_{\bK,\tau}$, and not an  $L^2_{\bK,1}-$ eigenvalue, hold for all $\eps\in(\eps_c,\eps_c+\delta)$, except possibly along the above sequence of bad $\eps$'s. This completes the proof of Lemma \ref{lemma3}. 
\medskip

To complete the proof of Theorem \ref{main-thm} for $\eps$ of arbitrary size we require\bigskip

\begin{lemma}\label{lemma4}
For all  $\eps\in(0,\infty)$ outside a countable closed set there exists a Floquet-Bloch eigenpair $\mu\in\mathbb{R},\ \varphi\in L^2_{\bK,\tau}$ for $-\Delta+\eps V_h$, with the following properties:
\begin{itemize}
\item[a.]\ $|\mu| \le \overline{C}_0\eps+\overline{C}_1$, where $\overline{C}_0$ and $\overline{C}_1$ depend only on $V_h$.
\item[b.]\ $\mu$ is a multiplicity one eigenvalue of $-\Delta+\eps V_h$ on $L^2_{\bK,\tau}$.
\item[c.]\ $\mu$ is not an eigenvalue of $-\Delta+\eps V_h$ on $L^2_{\bK,1}$.
\item[d.]\ The quantity $\lambda_\sharp^\eps$, arising from $\varphi$ by formula \eqref{lambda-sharp1} is non-zero.
\end{itemize}
\end{lemma}\medskip

\noindent Theorem \ref{main-thm} is an immediate consequence of Lemma \ref{lemma4}, which precludes option (2) of Lemma \ref{either-or}, and Proposition 
\ref{prop:2impliescone}.\medskip

\noindent{\it Proof of Lemma \ref{lemma4}:}\ \ Set $\overline{C}_0=\max\ |V_h|$. By our the analysis of section \ref{sec:pfeps-small}, there exists $\eps^0>0$, a sufficiently large constant $\overline{C}_1$,  such that for $\eps\in(0,\eps^0)$ there exist  $\mu,\ \varphi$ satisfying  (a.)-(d.) . \medskip

Now suppose that Lemma \ref{lemma4} fails. Then, by Lemma \ref{either-or}
 there exists $\eps_c\in(0,\infty)$ such that for all  $\eps\in(0,\eps_c)$ outside a  countable closed set, 
 there exist $\mu, \varphi$ satisfying (a.)-(d.) but 
\begin{align}
&\textrm{for \underline{all} $\eps_c^1>\eps_c$, assertions (a.-d.) fail on a subset of $(0,\eps_c^1)$}
\nn\\
& \textrm{that is not contained in any countable closed set. }
\label{defining-prop}
\end{align}
We will deduce a contradiction, from which we conclude Lemma \ref{lemma4}.\medskip

By assumption, we can find a sequence $\eps_1<\eps_2<\dots<\eps_\nu<\dots$ converging to $\eps_c$, such that each $\eps_\nu$ gives rise to a Floquet-Bloch eigenpair $\mu_\nu\in\mathbb{R},\ \varphi_\nu\in L^2_{\bK,\tau}$ satisfying properties (a.-d.)\ . Thanks to (a), we may pass to a subsequence, and assume that $\mu_\nu\to\mu_c$ as $\nu\to\infty$,
 for some real number $\mu_c$, with 
 \begin{equation}
 \left| \mu_c \right|\ \le\ \overline{C}_0\ \eps_c+\overline{C}_1
 \label{lip-bound-c}
 \end{equation}
 
 Since the Floquet-Bloch pairs $\mu_\nu, \varphi_\nu$ satisfy (a.-d.), and since $\eps_\nu\uparrow\eps_c$ and $\mu\to\mu_c$ as $\nu\to\infty$, Lemma  \ref{lemma3} applies. Thus we obtain a non-empty open interval $\mathcal{I}=(\eps_c,\eps_c+\delta)$, a real-valued real-analytic function $\beta(\eps)$ defined on $\mathcal{I}$, a function $\varphi^{(\eps)}\in L^2_{\bK,\tau}$ parametrized by $\eps\in\mathcal{I}$, and a countable closed subset $\mathcal{C}\subset\mathcal{I}$ satisfying properties (i)-(iv) of Lemma \ref{lemma3}. We will prove that 
 \begin{equation}
 |\beta(\eps)|\ \le\ \overline{C}_0\ \eps\ +\ \overline{C}_1,\ \ \textrm{for all}\ \ \eps\in\mathcal{I}=(\eps_c,\eps_c+\delta).
 \label{lip-bound-eps}
 \end{equation}
 Once \eqref{lip-bound-eps} is established, we conclude that we can satisfy all assertions (a.-d.) of Lemma \ref{lemma4} for all $\eps\in\mathcal{I}\setminus\mathcal{C}$, by taking $\mu=\beta(\eps)$ and $\varphi=\varphi^{(\eps)}$.
 However this contradicts  property \eqref{defining-prop} of $\eps_c$. Thus, it suffices to prove the bound \eqref{lip-bound-eps}.
 \medskip
 
 To establish \eqref{lip-bound-eps}, we fix $\bar\eps\in\mathcal{I}\setminus\mathcal{C}=(\eps_c,\eps_c+\delta)\setminus\mathcal{C}$. For any $\eps>0$, let $\lambda_1(\eps)\le\lambda_2(\eps)\le\dots$ denote the eigenvalues (multiplicity counted) of $-\Delta+\eps V_h$ on $L^2_{\bK,\tau}$. Then, since $\beta(\bar\eps)$ is a simple eigenvalue 
  ( assertion (iv) of Lemma \ref{lemma3} ), there exists $\bar k$ such that 
  \[ \beta(\bar\eps)\ =\ \lambda_{\bar k}(\bar\eps) < \lambda_{\bar k+1}(\bar\eps),\ \ 
   {\rm and}\ \ \lambda_{\bar k-1}(\bar\eps) < \lambda_{\bar k}(\bar\eps)\ {\rm unless}\ \bar k=1 .\]
   Fix $\bar\eta>0$ such that 
   \[ \lambda_{\bar k}(\bar\eps) < \lambda_{\bar k+1}(\bar\eps) - \bar\eta,\ \ 
   {\rm and}\ \ \lambda_{\bar k-1}(\bar\eps)  < \lambda_{\bar k}(\bar\eps) - \bar\eta\ {\rm unless}\ \bar k=1 .\]\medskip
   
   From the min-max characterization of eigenvalues, we have the Lipschitz bound
   \begin{equation}
   \left| \lambda_k(\eps) - \lambda_k(\eps') \right|\ \le\ |\eps-\eps'|\cdot \max\left| V_h \right|\ =\ \overline{C}_0\cdot |\eps - \eps'|,
   \label{lipbound}
   \end{equation}
   for any $\eps,\eps'>0$ and any $k\ge1$. Hence, as $\eps$ varies in a small neighborhood of $\bar\eps$, we have 
\begin{equation}
   \left| \lambda_{\bar k}(\eps) - \lambda_{\bar k}(\bar\eps) \right|\ \le\ 
   \overline{C}_0\cdot |\eps - \bar\eps|,
   \label{lipbound-epsbar}
   \end{equation}
   and also
    \begin{equation}
    \lambda_{\bar k+1}(\eps)  >  \lambda_{\bar k}(\bar\eps) + \frac{1}{2}\bar\eta,\ \ 
   {\rm and}\ \ \lambda_{\bar k-1}(\eps)  < \lambda_{\bar k}(\bar\eps) - \frac{1}{2}\bar\eta\ {\rm unless}\ \bar k=1 .
   \label{epsbareta}
   \end{equation}
   We have taken $\bar\eps\in\mathcal{I}\setminus\mathcal{C}$. As $\eps$ varies in a small neighborhood of $\bar\eps$, we have $\eps\in\mathcal{I}\setminus{C}$, thanks to property (iii) of Lemma \ref{lemma3}. Therefore, $\beta(\eps)$ is an eigenvalue of $-\Delta+\eps V_h$ on $L^2_{\bK,\tau}$, {\it i.e.} $\beta(\eps)=\lambda_k(\eps)$ for \underline{some} $k$. We now show that this value must be $\bar k$. 
   
   Since $\beta(\eps)$ is a real-analytic function of $\eps$, and since $\beta(\bar\eps)=\lambda_{\bar k}(\bar\eps)$, we have
   \begin{equation}
   \lambda_{\bar k}(\bar\eps)-\frac{1}{2}\bar\eta\ <\ \beta(\eps)\ <\ \lambda_{\bar k}(\bar\eps)+
   \frac{1}{2}\bar\eta
\label{lam-bar-upperlower}
  \end{equation}
for all $\eps$ close enough to $\bar\eps$.  From \eqref{epsbareta} and \eqref{lam-bar-upperlower}, we have
\footnote{ For $\bar k=1$, we have $\beta(\eps) < \lambda_{2}(\eps)$.}
\begin{align}
\lambda_{\bar k-1}(\eps) < \beta(\eps) < \lambda_{\bar k+1}(\eps)
\nn\end{align}
and therefore, $\beta(\eps)=\lambda_{\bar k}(\eps)$ for all $\eps$ close enough to $\bar\eps$. Estimate \eqref{lipbound-epsbar} now shows that the real-analytic function $\beta(\eps)$ satisfies
\begin{equation}
\left|\ \frac{d\beta(\eps)}{d\eps}\ \right|\ \le\ \overline{C}_0,\ \ \textrm{for}\ \eps=\bar\eps.
\nn\end{equation}
Since $\bar\eps$ was taken to be an arbitrary point of $\mathcal{I}\setminus{C}$, and since 
$\mathcal{I}\setminus{C}$ is dense in $\mathcal{I}$ by (iii) of Lemma \ref{lemma3}, we have
\begin{equation}
\left|\ \frac{d\beta(\eps)}{d\eps}\ \right|\ \le\ \overline{C}_0,\ \ \textrm{for all}\ \eps\in\mathcal{I}.
\label{lip-I}\end{equation}
Recall that $\mathcal{I}=(\eps_c,\eps_c+\delta)$. Our desired estimate \eqref{lip-bound-eps} now follows at once from \eqref{lip-bound-c}, \eqref{lip-I} and (ii) of Lemma \ref{lemma3}. The proof of Lemma \ref{lemma4} and therefore of Theorem \ref{main-thm} is now complete.
\bigskip

 \section{Persistence of conical (Dirac) points under perturbation}\label{sec:deformed}
 { }\medskip

In the previous sections we established the existence of conical singularities, Dirac points, in the dispersion surface 
for honeycomb lattice potentials. These Dirac points are at the vertices of the Brillouin zone, $\brill_h$.
In this section we explore the structural stability question of whether such Dirac points persist under small, even and $\Lambda_h$- periodic 
perturbations of a base honeycomb lattice potential.
We prove the following \medskip
 \begin{theorem}\label{deformed}
 Let $V(\bx)$ denote a honeycomb lattice potential in the sense of Definition \ref{honeyV}. Let $W(\bx)$
 denote a real-valued, smooth, even and $\Lambda_h$- periodic function, which does not necessarily have honeycomb structure symmetry, {\it i.e.} $W(\bx)$ is not necessarily $\mathcal{R}$- invariant. 
 Consider the operator
\begin{equation}
H(\eta)=-\Delta+V(\bx)+\eta W(\bx)\ ,
\label{Heta}
\end{equation}
where $\eta$ is a real parameter.
Let $\bk=\bK_\star$ be a vertex of $\brill_h$. Assume that for $\eta=0$, the operator $H(0)$ has an  $L^2_{\bK_\star}$- eigenvalue, $\mu(\bK_\star)$, of multiplicity two, with corresponding orthonormal basis $\{\Phi_1,\Phi_2\}$ with $\Phi_1\in L^2_{\bK_\star,\tau}$ and $\Phi_2(\bx)=\overline{\Phi_1(-\bx)}$. Assume  $\lambda_\sharp$,  given in \eqref{lambda-sharp1}, is non-zero.  Then, the following hold:\medskip

 \begin{enumerate}
 \item There exist a positive number $\eta_1$ and a smooth mappings 
 \begin{align}
 &\eta\mapsto \mu^{(\eta)}=\mu(\bK_\star)+\mathcal{O}(\eta)\in\mathbb{R}\ \textrm{and}\ 
 \eta\mapsto\bK^{(\eta)}=\bK_\star+\mathcal{O}(\eta)\in\brill_h,\label{muKeta}\\
& \eta\mapsto \phi_j^{(\eta)}(\bx;\bK^{(\eta)})\ =\ \phi_j(\bx)+\mathcal{O}(\eta)\in L^2(\mathbb{R}/\Lambda_h)\label{phijeta}
 \end{align}
  defined for $|\eta|<\eta_1$, such that   $H(\eta)$ has an $L^2_{\bK^{(\eta)}}$ 
 eigenvalue, $\mu^{(\eta)}$, of geometric multiplicity two, with corresponding eigenspace spanned by
\[ \left\{\Phi_1^{(\eta)}, \Phi_2^{(\eta)}\right\}\ =\ \left\{e^{i\bK^{(\eta)}\cdot\bx}\phi_1^{(\eta)}(\bx;\bK^{(\eta)}) , e^{i\bK^{(\eta)}\cdot\bx}\phi_2^{(\eta)}(\bx;\bK^{(\eta)} )\right\}.\]
\item  The operator $H(\eta)$ has conical-type dispersion surfaces
 in a neighborhood of points $\bK_\star^{(\eta)}=\bK_\star+\mathcal{O}(\eta)$ with associated band dispersion functions,  $\mu^{(\eta)}_\pm(\bk)$, defined for $\bk$ near $\bK^{(\eta)}_\star$:
{\small{
 \begin{align}
 \mu^{(\eta)}_+(\bk) - \mu(\bK^{(\eta)}_\star) &= \eta{\bf b}^{(\eta)}\cdot(\bk-\bK_\star^{(\eta)})+\left(\mathcal{Q}^{(\eta)}(\bk-\bK_\star^{(\eta)}) \right)^{1\over2}\ \left(\ 1+E^{(\eta)}_+(\bk-\bK^{(\eta)}_\star)\ \right)\label{localcone+eta}\\
 \mu^{(\eta)}_-(\bk) - \mu(\bK^{(\eta)}_\star)\ &=\ \eta{\bf b}^{(\eta)}\cdot(\bk-\bK_\star^{(\eta)})-\left(\mathcal{Q}^{(\eta)}(\bk-\bK_\star^{(\eta)}) \right)^{1\over2}\ \left(\ 1+E^{(\eta)}_-(\bk-\bK^{(\eta)}_\star)\ \right)\ ,
\label{localcone-eta}  \end{align}
}}
where
\begin{itemize}
\item ${\bf b}^{(\eta)}\in\mathbb{R}^2$ depends smoothly on $\eta$.
\item $\mathcal{Q}^{(\eta)}(\cdot)$ is a quadratic form in $\bkappa=(\kappa_1,\kappa_2)\in\mathbb{R}^2$, depending smoothly on $\eta$ and such that 
\begin{equation}
\left( |\lambda_\sharp|^2-C|\eta| \right)(\kappa_1^2+\kappa_2^2)\
 \le\ \mathcal{Q}(\bkappa;\eta)\ \le\ \left( |\lambda_\sharp|^2+C|\eta| \right) (\kappa_1^2+\kappa_2^2)
 \label{clf11a}
 \end{equation}
 for $|\eta|\le\eta_1$ and $\bkappa=(\kappa_1,\kappa_2)\in \mathbb{R}^2$, 
  with $\eta_1$ small, and 
\item $\left| E^{(\eta)}_+(\kappa) \right|,\ \left| E^{(\eta)}_-(\kappa) \right|\ \le\ C\ \left|\bkappa\right|$ for $|\eta|\le\eta_1$ and $|\kappa|\le\tilde{\kappa}$, where $\tilde{\kappa}$ (small) and $C<\infty$ are constants.
\end{itemize}
 \end{enumerate}
 \end{theorem}
 \medskip
 
\noindent Remark \ref{instability} below shows that Dirac points are unstable to typical perturbations,
$W\in C^\infty(\mathbb{R}^2/\Lambda_h)$, which are not even.\medskip

 \begin{remark}
 For $\eta=0$, Theorem \ref{deformed} reduces to 
  Theorem \ref{main-thm}, which covers the case of the undeformed honeycomb lattice potential.
In particular, if the perturbation $W$ is itself a honeycomb lattice potential,  then  
$\bK_\star^{(\eta)}\equiv \bK_\star$.
 \end{remark}

 \bigskip
 
 \noindent We now prove Theorem \ref{deformed}. As earlier, without loss of generality, we assume $\bK_\star=\bK$.   The family of Floquet-Bloch eigenvalue problems, parametrized by $\bk\in\brill_h$, is given by:
\begin{align}
H(\bk;\eta)\ \phi^{(\eta)}(\bx;\bk)\ =\ \mu^{(\eta)}(\bk)\ \phi^{(\eta)}(\bx;\bk),\ \textrm{where}\nn\\
H(\bk;\eta)\ =\ -\left(\nabla+i\bk\right)^2\ +\ V(\bx)\ +\ \eta\ W(\bx)\ . \label{Heta-evp}
\end{align}
By hypothesis, we have:
\begin{align}
&\textrm{For $\eta=0$,  $H(\bK;0)=-\left(\nabla+i\bK\right)^2+V(\bx)$ has a degenerate eigenvalue $\mu(\bK)$,}\nn\\ 
&\textrm{ of multiplicity two with  $2$-dimensional $L^2_{\bK}$- eigenspace: 
  ${\rm span}\{\phi_1(\bx) , \phi_2(\bx)\}$.}
  \end{align}
  
  Introduce the projection operators
   \begin{equation}
   Q_\parallel f\ =\ \sum_{j=1}^2\langle \phi_j,f\rangle\ \phi_j(\bx)
   \ \textrm{and}\ Q_\perp\ =\ I-Q_\parallel\ .
   \label{projections}
   \end{equation}
  
  We seek solutions of the Floquet-Bloch eigenvalue problem in the form:
  \begin{align}
  & \phi^{(\eta)}(\bx;\bK^{\eta})\ =\ \sum_{j=1}^2\alpha_j\phi_j(\bx)\ +\ \eta\ \phi^{(1,\eta)}(\bx),\ \ Q_\parallel  \phi^{(1,\eta)}=0\ ,\label{peta-expand}\\
 & \bK^{(\eta)}\ =\ \bK\ +\ \eta\bK^{1,\eta}\ , \label{K-expand}\\
 & \mu^{(\eta)}\ =\ \mu^{(0)}\ +\ \eta \mu^{(1,\eta)}\ , \label{mueta-expand}\\
&  {\rm dim}\ \textrm {null space}\left(\ H(\eta) - \mu^{(\eta)} I\ \right)\ = 2\ .\label{double-eig}
  \end{align}
  
  Substituting these expansions into \eqref{Heta-evp} yields:
  \begin{align}
&  \left(\ H(\bK)-\mu^{(0)} I \ \right)\ \phi^{(1,\eta)}\ +\ 
  \eta\ \left(\ -2i\bK^{1,\eta}\cdot\nabla_\bK+W-\mu^{(1,\eta)}-\eta |\bK^{1,\eta}|^2 \right) \phi^{(1,\eta)}\
  \nn\\
  &\ \ \ \ \ =\ -\left(\ -2i\bK^{1,\eta}\cdot\nabla_\bK+W-\mu^{(1,\eta)} - \eta |\bK^{1,\eta}|^2 \right)
   \left(\ \sum_{j=1}^2\alpha_j\phi_j\ \right)\ .\label{p1eta-eqn}
   \end{align}

We take the inner product of \eqref{p1eta-eqn} with $\phi_1$ and $\phi_2$.  This yields the system:

{\small{   
   \begin{align}
&  \left[-\left(\mu^{(1,\eta)}+\eta|\bK^{1,\eta}|^2\right) +\langle \phi_1,W \phi_1\rangle\right]\alpha_1
 \ +\ \left[\ \overline{ \lambda_\sharp } \left(K_1^{1,\eta}+i K_2^{1,\eta}\right)
 +\langle \phi_1,W \phi_2 \rangle \right] \alpha_2
  \nn\\
 &+\ \ \ \ \ \ \ \ \ \ \eta\ \left\langle \phi_1, \left(-2i\bK^{(1,\eta)}\cdot\nabla_\bK+W\right) Q_\perp \phi^{(1,\eta)} \right\rangle\ =\ 0\ ,\label{ortho-p1}\\
 &\nn\\
 &\overline{\left[\ \overline{ \lambda_\sharp } \left( K_1^{1,\eta}+i K_2^{1,\eta}\right)
 +\langle \phi_1,W \phi_2\rangle \right]}\alpha_1 + \left[-\left(\mu^{(1,\eta)}+\eta|\bK^{1,\eta}|^2\right) +\langle \phi_2,W\phi_2\rangle\right]\alpha_2\nn\\
  &+\ \ \ \ \ \ \ \ \ \ \eta\ \left\langle \phi_2, \left(-2i\bK^{(1,\eta)}\cdot\nabla_\bK+W\right) Q_\perp \phi^{(1,\eta)} \right\rangle\ =\ 0\ .\label{ortho-p2} 
   \end{align}
   }}
   In obtaining the system \eqref{ortho-p1}-\eqref{ortho-p2} we have used
   \begin{enumerate}
   \item 
   \begin{equation}
   \nabla_\bK\ \equiv\ \left(\nabla_\bx+i\bK\right)\ =\ e^{-i\bK\cdot\bx}\nabla_\bx e^{i\bK\cdot\bx},\label{nablaK}\\
   \end{equation} 
 \item $W$ is real-valued and thus $-2i\bK^{1,\eta}\cdot\nabla+W$ is self-adjoint,\smallskip
 
%
\item   
{\small{
$-2i\left\langle \phi_1,\ \bK^{1,\eta}\cdot\nabla_\bK\ \phi_2\right\rangle\ =\ 
-2i\left\langle \Phi_1 ,\ \bK^{1,\eta}\cdot\nabla_\bx\ \Phi_2\right\rangle\ =\
 \overline{\lambda_\sharp}\left(\bK_1^{(\eta,1)}+i\bK_1^{(\eta,1)}\right),$ 
 }}  where $\Phi_j=e^{i\bK\cdot\bx}\phi_j$, by  Proposition \ref{prop:matrix-elements},\smallskip

  \item $Q_\parallel \phi^{(1,\eta)}=0$ and $Q_\perp \phi^{(1,\eta)}=\phi^{(1,\eta)}$\ .
 \end{enumerate}
  \bigskip
  
\noindent    An equation 
 for $\phi^{(1,\eta)}$ is derived by applying $Q_\perp$ to \eqref{p1eta-eqn} and  using
 \[Q_\perp \phi_j=0,\ j=1,2\ \textrm{and}\ Q_\perp \phi^{(1,\eta)} = \phi^{(1,\eta)}.\]
   We obtain
  {\small{
   \begin{align}
&  \left(\ H(\bK)-\mu^{(0)} I \ \right)\ \phi^{(1,\eta)}\ +\ 
  \eta\ Q_\perp \left(\ -2i\bK^{1,\eta}\cdot\nabla_\bK+W-\mu^{(1,\eta)}-\eta |\bK^{1,\eta}|^2 \right) Q_\perp\ \phi^{(1,\eta)}\
  \nn\\
  &\ \ \ \ \ =\ -Q_\perp\ \left(\ -2i\bK^{1,\eta}\cdot\nabla_\bK+W \right)\ 
   \left(\ \sum_{j=1}^2\alpha_j\phi_j\ \right) \ .\label{Qp1eta-eqn}
   \end{align}
   }}
   Introduce the projections
   \begin{align}
  & \tilde{Q}_\parallel F = \sum_{j=1}^2 \langle\Phi_j,F\rangle\ \Phi_j,\ \ \tilde{Q}_\perp = I - \tilde{Q}_\parallel
  \label{tprojections}\\
 & \textrm{and note that these projections satisfy the commutation relation}\nn\\
 & e^{i\bK\cdot\bx}\ Q\ =\ \tilde{Q}\ e^{i\bK\cdot\bx}\ .
  \label{tQ-commute}
  \end{align}
   Using \eqref{nablaK} we rewrite \eqref{Qp1eta-eqn} as an equation for
   \begin{equation}
   \Phi^{(1,\eta)}\ =\ e^{i\bK\cdot\bx} \phi^{(1,\eta)}\ :
   \label{P1def}
   \end{equation}
   \begin{align}
   \mathcal{L}\left(\mu,\bK^{1,\eta},\eta\right) \Phi^{(1,\eta)}\ =\ 
   -\sum_{j=1}^2\  \alpha_j\ \tilde{Q}_\perp\ \left(\ -2i\bK^{1,\eta}\cdot\nabla_\bx+W \right)
  \Phi_j,\ {\rm where}\label{LPeqn}\\
   \mathcal{L}\left(\mu,\bK^{1,\eta},\eta\right)\ \equiv\ -\Delta+V-\mu^{(0)}+\eta\ \tilde{Q}_\perp\ \left(-2i\bK^{1,\eta}\cdot\nabla_\bx + \eta^2|\bK^{1,\eta}|^2\right)\ \tilde{Q}_\perp
\label{Ldef}   \end{align}
  Note that $ \mathcal{L}$ is self-adjoint on $\tilde{Q}_\perp L^2_\bK$. For $\eta$ sufficiently small, we have that $\mathcal{L}:\tilde{Q}_\perp H^2_\bK\to \tilde{Q}_\perp L^2_\bK$ is invertible. Solving \eqref{LPeqn} yields
  \begin{equation}
  \Phi^{(1,\eta)}\ =\ -\sum_{j=1}^2\  \alpha_j\ \mathcal{L}\left(\mu,\bK^{1,\eta},\eta\right)^{-1}
   \tilde{Q}_\perp\ \left(\ -2i\bK^{1,\eta}\cdot\nabla_\bx+W \right)\  \Phi_j\ .\label{P1eta-solve}
  \end{equation}
  Using \eqref{P1def} to express $\phi^{(1,\eta)}$ in terms of $\Phi^{(1,\eta)}$ and substituting  
  \eqref{P1eta-solve} into \eqref{ortho-p1}-\eqref{ortho-p2}, we obtain the following linear homogeneous system for $\alpha_j,\ j=1,2$:
  \begin{equation}
  \mathcal{M}\left(\mu^{(1,\eta)},\bK^{1,\eta},\eta\right) \left(\begin{array}{c} \alpha_1 \\ \alpha_2\end{array}\right)\ =\ \left(\begin{array}{c} 0 \\ 0\end{array}\right)\  ,
  \end{equation}
  where
  {\footnotesize{
 \begin{align} &\mathcal{M}\left(\mu^{(1,\eta)},\bK^{1,\eta},\eta\right)\ \equiv\nn\\
  &\ \left(\begin{array}{cc}
 -\left(\mu^{(1,\eta)}+\eta|\bK^{1,\eta}|^2\right) +\langle\Phi_1,W \Phi_1\rangle + \eta\ a_{11} 
 & \quad \overline{ \lambda_\sharp } \left( K_1^{1,\eta}+i K_2^{1,\eta}\right)
 +\langle \Phi_1,W \Phi_2 \rangle + \eta\ b\\
 {\quad } & {\quad }\\
 \overline{ \overline{ \lambda_\sharp } \left( K_1^{1,\eta}+i K_2^{1,\eta}\right)
 +\langle \Phi_1,W \Phi_2 \rangle + \eta\ b} & \quad -\left(\mu^{(1,\eta)}+\eta|\bK^{1,\eta}|^2\right) +\langle \Phi_1,W \Phi_1\rangle + \eta\ a_{22}
  \end{array} \right).\nn\\
  &\ \label{Mdef1}
 \end{align}
 }}
$a_{11}, a_{22}$ and $b$ are  functions of $\mu^{(1,\eta)}, \bK^{1,\eta}$ and $\eta$ and are given by the expressions:
 \begin{align}
 a_{ll}&=\ \left\langle \tilde{Q}_\perp\left(\ -2i\bK^{1,\eta}\cdot\nabla_\bx+W\ \right)\Phi_l\ ,\ \mathcal{L}^{-1}\ 
 \tilde{Q}_\perp \left(\ -2i\bK^{1,\eta}\cdot\nabla_\bx+W\ \right)\Phi_l\right\rangle\ ,\ l=1,2
  \label{adef}\\
 b&=\ \left\langle \tilde{Q}_\perp\left(\ -2i\bK^{1,\eta}\cdot\nabla_\bx+W\ \right)\Phi_1\ ,\ \mathcal{L}^{-1}\ 
 \tilde{Q}_\perp \left(\ -2i\bK^{1,\eta}\cdot\nabla_\bx+W\ \right)\Phi_2\right\rangle
 \label{bdef}
 \end{align}
 where $\mathcal{L}= \mathcal{L}\left(\mu,\bK^{1},\eta\right) $ is defined in \eqref{Ldef}. Note that $a_{11}$ and $a_{22}$ are real. The matrix $\mathcal{M}\left(\mu^{(1)},\bK^{1},\eta\right)$ has the structure:
 
 \begin{align}
&\mathcal{M}\left(\mu^{(1)},\bK^{1},\eta\right)\ =\nn\\
&  \left(\begin{array}{cc} 
 -\mu^{(1)}+A_{11}\left(\mu^{(1)},\bK^{1},\eta\right) & B\left(\mu^{(1)},\bK^{1},\eta\right)\\ 
 {\quad} & {\quad}\\
 \overline{B\left(\mu^{(1)},\bK^{1},\eta\right)} & -\mu^{(1)}+A_{22}\left(\mu^{(1)},\bK^{1},\eta\right)
\end{array} \right)\label{Mstructure}
 \end{align}
 where $A_{ll}$ and $B$ are smooth functions of $\left(\mu^{(1)},\bK^{1,\eta},\eta\right)$
, which can be read off from  \eqref{Mdef1} and \eqref{adef}-\eqref{bdef}:

\begin{align}
A_{ll}\ &=\ \langle\Phi_l,W\Phi_l\rangle-\eta|\bK^{1}|^2+\eta\ a_{ll},\ \ l=1,2\label{Adef}\\
B\ &=\ \overline{\lambda_\sharp}\ (K_1^{1}+i K_2^{1})\ +\ 
\langle\Phi_1,W\Phi_2\rangle + \eta\ b\ .\label{Bdef}
\end{align}

A consequence of the above discussion is
\begin{proposition}\label{mu1eta-peig}
\begin{enumerate}
\item The pair $(\mu^{(\eta)},\phi^{(\eta)})$ is an $L^2_{\bK^{(\eta)}}$- eigenpair of 
$H(\eta)=-\Delta+V+\eta W$, where $\eta$ is real and $\bK^{(\eta)}\in\mathbb{R}^2$, $\mu^{(\eta)}$, $\phi^{(\eta)}\in L^2_{\bK^{(\eta)}}$ are defined in
 \eqref{peta-expand}-\eqref{mueta-expand}, if and only if
 \begin{align}
 &\det\mathcal{M}\left(\mu^{(1)},\bK^{1},\eta\right)=0\ .
 \label{detmu1}\end{align}
 \item By self-adjointness, for $\eta\in\mathbb{R}$ and $\bK^{1}\in\mathbb{R}^2$, if $\mu^{(1)}$ is a solution of \eqref{detmu1} then $\mu^{(1)}$ is real.
 \item $\mu^{(\eta)}$ is a geometric multiplicity two $L^2_{\bK^{(\eta)}}$- eigenvalue of $H(\eta)$ if and only if the triple
 $\left(\mu^{(1)}, \bK^{1}, \eta\right)$ is such that the $2\times2$ Hermitian matrix, $\mathcal{M}\left(\mu^{(1)},\bK^{1},\eta\right)$ has zero as a double eigenvalue, {\it i.e.} $\mathcal{M}\left(\mu^{(1)},\bK^{1},\eta\right)$ is the zero matrix.
 \end{enumerate}
\end{proposition}
\medskip

Now, up to this point we have not used the hypothesis that $W(\bx)$ is an even function (inversion symmetry). We now impose this condition on $W$. For the case where $W$ is not even, see Remark \ref{instability} at this end of this section. \medskip

\noindent{\bf Claim 1:}\ 
$ W(\bx)=W(-\bx)\ \implies\ \langle \Phi_1,W\Phi_{1}\rangle\ =\ \langle \Phi_2,W\Phi_{2}\rangle\ \ {\rm and}\ \ a_{11}\ =\ a_{22}$. By \eqref{Adef}, it follows that $A_{11} = A_{22}$.\\
\smallskip

\noindent{\it Proof of Claim 1:}  Recall, as in Theorem \ref{main-thm}, that $\Phi_2(\bx;\bK)=\overline{\Phi_1(-\bx;\bK)}$ and thus
\[\langle \Phi_2,W\Phi_{2}\rangle=\int \Phi_1(-\bx) W(\bx) \overline{\Phi_1(-\bx)} d\bx =\int |\Phi_1(\bx)|^2 W(-\bx) d\bx  =
\langle \Phi_1,W\Phi_{1}\rangle.\]
Furthermore, one checks easily that $a_{11}=a_{22}$.
 
\medskip

\noindent In this case, we set 
\begin{equation}
a_{11}=a_{22}\equiv a= a\left(\mu^{(1)},\bK^1,\eta\right)\ \ {\rm and}\ \ 
A_{11}=A_{22}\equiv A= A\left(\mu^{(1)},\bK^1,\eta\right). \label{aA}
\end{equation}
Here $a=a_{11}=a_{22}$ and $b$ are functions of $\mu^{(1)}, \bK^1, \eta$, displayed in \eqref{adef} and \eqref{bdef}.\medskip
\bigskip

 By Proposition \ref{mu1eta-peig} and the above Claim 1, if 
 $W(\bx)=W(-\bx)$,  then we obtain a double eigenvalue if and only if
  \begin{equation}
\mu^{(1)}-A\left(\mu^{(1)},\bK^1,\eta\right)=0\ \ {\rm and}\  \ B\left(\mu^{(1)},\bK^1,\eta\right)=0.
\label{decoupledAB-a}
\end{equation}

 By analyzing the solution set of \eqref{decoupledAB-a} for small $\eta$, we shall prove the following:\medskip
 
 \begin{proposition}\label{deformed-Dirac}
  For each real $\eta$ in some small neighborhood of zero, there exists  a unique
  $\bK^1=\bK^{1,\eta}=\left(K^{1,\eta}_1,K^{1,\eta}_2\right)$ and $\mu^{(1)}=\mu^{(1,\eta)}$ such that
  $\mu^{(\eta)}=\mu^{(0)}+\eta\mu^{(1,\eta)}$ (see \eqref{mueta-expand}) is a geometric multiplicity two $L^2(\mathbb{R}^2/\Lambda_h)$- eigenvalue of $H(\eta;\bK^{1,\eta})$.
  \end{proposition}
\medskip

\noindent{\it Proof of Proposition \ref{deformed-Dirac}:}\  Consider  \eqref{decoupledAB-a}
 for $\mu^{(1,\eta)}$ and $\bK^{1,\eta}$ for $\eta=0$. We have
\begin{align}\mu^{(1,0)}-A\left(\mu^{(1,0)},\bK^1,0\right)=0\ &\leftrightarrow\ \mu^{(1,0)}\ =\ \langle\Phi_1,W \Phi_1\rangle\ \ \ \ \ \ {\rm and}
 \label{mu1-1st}\\
 B\left(\mu^{(1,0)},\bK^1,0\right) = 0\ &\leftrightarrow\ 
 \overline{ \lambda_\sharp } \left(K_1^{1,0}+i K_2^{1,0}\right)
 \ =\ -\langle \Phi_1,W \Phi_2 \rangle\ .
 \label{K1K2eta0}
 \end{align}
Equation \eqref{K1K2eta0} is equivalent to the two equations:
  \begin{align}
 K_1^{1,0} &=\ -\Re\left(\ (\ \overline{ \lambda_\sharp }\ )^{-1}\ \langle \Phi_1,W \Phi_2 \rangle\ \right),\ \  
 K_2^{1,0} =\ -\Im\left(\ (\ \overline{ \lambda_\sharp }\ )^{-1}\ \langle \Phi_1,W \Phi_2 \rangle\ \right)
 \label{K12-1st}\end{align}

We next consider the case $\eta\ne0$, real and sufficiently small.\medskip

\noindent{\bf Claim 2:} $A$ and $B$, defined via \eqref{Mdef1}-\eqref{Mstructure},  are smooth functions of  $(\mu^{(1)},\eta,  \bK_1^{1,0})$. Moreover, there exist constants $c_1>0$, $d_1>0$ and $\eta_0>0$ such that for all $|\bK^1-\bK^{1,0}|<c_1$,
 $|\mu^{(1)}-\mu^{(1,0)}|<d_1$ and $|\eta|<\eta_0$, we have
 \begin{enumerate}
 \item $A = \langle\Phi_1,W\Phi_1\rangle\ +\ \eta\ f_A\left(\mu^{(1)},\bK^1,\eta\right)$\smallskip
 
 \item $B = \overline{\lambda_\sharp}\ (K_1^1 + iK_2^1 )\ +\ \langle\Phi_1,W\Phi_2\rangle\ +\ 
 \eta\ f_B\left(\mu^{(1)},\bK^1,\eta\right)$\smallskip
 
 \item $f_A,\ f_B=\mathcal{O}(1)$\ .\smallskip
 
  \item $\partial_\mu B=\mathcal{O}(\eta),\ \ |\partial_\mu A|\le1/2$.\
 \end{enumerate}
We leave the verification to the reader.
\medskip

%
%
%
 
 An immediate consequence of Claims 1 and 2 is:\medskip
 
 \noindent{\bf Claim 3:} Assume $\lambda_\sharp\ne0$. Then, for $|\eta|<\eta_0$, $|\bK^1-\bK^{1,0}|<c_1$ and  $|\mu^{(1)}-\mu^{(1,0)}|<d_1$ equations \eqref{decoupledAB-a}  are equivalent to the system:
   \begin{align}
 K_1^{(1)}\ +\ i K_2^{(1)}\ &=\ \left(\overline{\lambda_\sharp}\right)^{-1}\langle\Phi_1,W\Phi_1\rangle
   + \eta F(\mu^{(1)},\eta,  K_1^{(1)}, K_2^{(1)}),\label{K1K2}\\
   \mu^{(1)}\ &=\ \langle\Phi_1,W\Phi_1\rangle\ +\ \eta f_A\left(\mu^{(1)},\bK^1,\eta\right)\ 
   .\label{mu1a}
   \end{align}
         \medskip
     
By the above, we have\medskip

\noindent{\bf Claim 3:}\      $\mu^{(\eta)}=\mu^{(0)}+\eta \mu^{(1,\eta)}$ is an $L^2(\mathbb{R}^2/\Lambda_h)$- eigenvalue of $H(\bK^{(\eta)},\eta)$ of geometric multiplicity two (see \eqref{peta-expand}-\eqref{mueta-expand}  ) if and only if 
     $\mu^{(1,\eta)}$,  $K_1^{1,\eta}$ and  $K_2^{1,\eta}$ satisfy:
\begin{align}
&K_1^{1,\eta}\ -\ \Re\left(\ \left(\overline{\lambda_\sharp}\right)^{-1}\langle\Phi_1,W\Phi_1\rangle
   + \eta F(\mu^{(1,\eta)},\eta,  K_1^{1,\eta}, K_2^{1,\eta})\ \right)\ =\ 0\nn\\
&K_2^{1,\eta}\ -\ \Im\left(\ \left(\overline{\lambda_\sharp}\right)^{-1}\langle\Phi_1,W\Phi_1\rangle
   + \eta F(\mu^{(1,\eta)},\eta,  K_1^{1,\eta}, K_2^{1,\eta})\ \right)\ =\ 0\nn\\
&\mu^{(1,\eta)}\ -\ \left(\ \langle\Phi_1,W\Phi_1\rangle\ +\ \eta\ f_A(\mu^{(1,\eta)},\eta,  K_1^{1,\eta}, K_2^{1,\eta})\  \right)\ =\ 0
\label{3system}\end{align}

 So in order to prove Theorem \ref{deformed} we seek a solution of \eqref{3system} in a neighborhood of  its solution for $\eta=0$, given by \eqref{mu1-1st} and \eqref{K12-1st}: 
 \begin{equation}
 \mu^{(1)}=\mu^{(1,0)},\ K^1=K^{1,0},\ K^2=K^{2,0}.\label{solution-1st}
 \end{equation}
Note that the right hand side of \eqref{3system} defines a smooth map
 from a neighborhood of $(\mu^{(1,0)}, K_1^{1,0} , K_2^{1,0})$ to $\mathbb{R}^3$ with Jacobian at \eqref{solution-1st}, for $\eta=0$, equal to the identity. Hence, by the implicit function theorem, there exist a positive number, $\eta_1$, and smooth functions:
\begin{align}
&\eta\mapsto\mu^{(1,\eta)},\ \ \ \ \eta\mapsto \bK^{1,\eta}=\left(K_1^{1,\eta},K_2^{1,\eta}\right),
\end{align}
defined for $|\eta|<\eta_1$, such that $\mu^{(1,\eta)},\bK^{1,\eta}$ is the unique solution of \eqref{3system} for all $|\eta|<\eta_1$ in an open set about the point \eqref{solution-1st}.  This completes the proof of part (1) of Theorem \ref{deformed}.
\bigskip

To prove part (2) of Theorem  \ref{deformed}, we need to display a conical singularity in the dispersion surface about the point $(\bK^{(\eta)},\mu^{(\eta)})$, \eqref{muKeta}. For this we make strong use of the 
calculations in the proof of Theorem \ref{prop:2impliescone}. In particular,   $-\Delta+V$, $\mu^{(0)}$, $\bK$ and $\phi_j,\ j=1,2$ of the proof of Theorem \ref{prop:2impliescone} are replaced by $H(\eta), \mu^{(\eta)}, \bK^{(\eta)}$ from the proof of part (1) and $\{\phi^{(\eta)}_1,\phi^{(\eta)}_2\}$, now denotes an orthonormal spanning set for the $L^2(\mathbb{R}^2/\Lambda_h)$ nullspace of $H(\eta;\bK^{(\eta)})-\mu^{(\eta)}I$. Then,  $\{\Phi^{(\eta)}_1,\Phi^{(\eta)}_2\}=\{e^{i\bK\cdot\bx}\phi^{(\eta)}_1,e^{i\bK\cdot\bx}\phi^{(\eta)}_2\}$  is an orthonormal spanning set for the 
$L^2_{\bK^{(\eta)}}$ nullspace of $H(\eta)-\mu^{(\eta)}I$. Note also that $\Phi^{(\eta=0)}_j=\Phi_j$,
 the Floquet-Bloch states associated with the unperturbed honeycomb lattice potential, $V$.\medskip

We must study the Floquet-Bloch eigenvalue problem (compare with \eqref{Hk-evp}-\eqref{psi-per}) 
\begin{align}
&\left(\ -\left(\nabla_\bx + i\left(\bK^{(\eta)}+\bkappa\right)\right)^2\ +\ V(\bx)\ +\ \eta W(\bx)\right) 
\psi^{(\eta)}(\bx;\bK^{(\eta)}+\bkappa)\nn\\
&\ \ \ \ \ \ \ \ \ \ \ \ \ \ \ \ \ \ \ \ \ \ \ \ =\ \mu(\bK^{(\eta)}+\bkappa)\ \psi^{(\eta)}(\bx;\bK^{(\eta)}+\bkappa)\ ,\label{Hk-eta-evp}\\
&\psi^{(\eta)}(\bx+\bv;\bK^{(\eta)}+\bkappa)=\psi^{(\eta)}(\bx;\bK^{(\eta)}+\bkappa),\ \ \textrm{for all}\ \bv\in \Lambda\ .
\label{psi-eta-per}\end{align}

We express  $\mu=\mu^{(\eta)}(\bK^{(\eta)}+\bkappa)$ and $\psi^{(\eta)}(\bx;\bK^{(\eta)}+\bkappa)$ as:
\begin{align}
\mu &= \mu^{(\eta)} +\ \mu^{(1,\eta)}(\bkappa),\nn\\
  \psi^{(\eta)}(\bx;\bK^{(\eta)}+\bkappa) &= \sum_{j=1}^2\alpha_j\phi_j^{(\eta)}(\bx) +\ \psi^{(1,\eta)}(\bx), \label{mu-psi-eta-exp}
\end{align}
where $\mu^{(\eta)}=\mu^{(\eta)}(\bK^{(\eta)})$ denotes the perturbed  double-eigenvalue constructed above with corresponding orthonormal eigenfunctions $\phi_j^{(\eta)}(\bx),\ j=1,2$.

Precisely along the lines of the derivation of \eqref{det0} in the proof of Theorem \ref{prop:2impliescone}, we now find that for $|\bkappa|$ small
$\mu^{(\eta)}(\bK^{(\eta)}+\bkappa)$ is an eigenvalue of the spectral problem \eqref{Hk-eta-evp}-\eqref{psi-eta-per} if $\mu^{(1,\eta)}=\mu^{(1,\eta)}(\bkappa)$ solves
\begin{equation}
\det\mathcal{M}(\mu^{(1,\eta)},\bkappa;\eta)=0,
\label{det0-eta}\end{equation}
where 
\begin{equation}
\mathcal{M}(\mu^{(1,\eta)},\bkappa;\eta)=\mathcal{M}_0(\mu^{(1,\eta)},\bkappa;\eta)\ +\ \mathcal{M}_1( \mu^{(1,\eta)},\kappa;\eta),
 \label{M1eta-expand}
 \end{equation}
 where 
 \begin{equation}
\left|\mathcal{M}_{1,ij}( \mu^{(1,\eta)},\kappa;\eta)\right|\ =\
 C\left(\ |\kappa|\ |\mu^{(1,\eta)}| + |\kappa|^2\ \right).
  \nn  \end{equation}
  and
  \begin{align}
\mathcal{M}_0(\mu^{(1,\eta)},\bkappa;\eta)&= 
\left(\begin{array}{cc}
 \mu^{(1,\eta)} + 2i\langle\Phi_1^{(\eta)},\kappa\cdot\nabla\Phi_1^{(\eta)}\rangle & 
 2i\langle\Phi_1^{(\eta)},\kappa\cdot\nabla\Phi_2^{(\eta)}\rangle \\
 2i\langle\Phi_2^{(\eta)},\kappa\cdot\nabla\Phi_1^{(\eta)}\rangle &
  \mu^{(1,\eta)} + 2i\langle\Phi_2^{(\eta)},\kappa\cdot\nabla\Phi_2^{(\eta)}\rangle
  \end{array}\right)\nn\\
 &=  \left(\begin{array}{cc}
\mu^{(1,\eta)} + M_{11}^0(\eta)\cdot\bkappa & M_{12}^0(\eta)\cdot\bkappa\\
 M_{21}^0(\eta) \cdot\bkappa & \mu^{(1,\eta)} + M_{22}^0(\eta)\cdot\bkappa,
 \end{array}\right)
  \label{clf4}\end{align}
where $M^0_{jk}(\eta)$ are smooth complex-valued functions of $\eta$.  Note that
\begin{equation}
M_{11}^0(\eta),\ M_{22}^0(\eta)\ \textrm{are real and}\ M_{21}^0(\eta) =  \overline{M_{12}^0(\eta)}
\nn\end{equation}
and furthermore
\begin{align}
\left. \mathcal{M}_0(\mu^{(1,\eta)},\bkappa;\eta)\right|_{\eta=0}\ =\ 
 \left(\begin{array}{cc}
 \mu^{(1,\eta)}  & -\overline{\lambda_\sharp}\ (\kappa_1+i\kappa_2) \\
 -\lambda_\sharp\ (\kappa_1-i\kappa_2) &
 \mu^{(1,\eta)}
  \end{array}\right)\ .
  \label{clf6}\end{align}

Thanks to \eqref{clf4}, the equation $\det\mathcal{M}_0(\nu,\bkappa;\eta)=0$ is equivalent to
\begin{equation}
\nu^2\ +\ \left\{\ \left[M_{11}^0(\eta) + M_{22}^0(\eta)\right]\cdot\bkappa\ \right\} \nu\ +\ 
 \det\left[ \left( M_{jl}^0(\eta)\cdot\bkappa \right)_{j,l=1,2} \right]\ =\ 0
 \label{clf7}
 \end{equation}
 
 The solutions have the form
 \begin{equation}
 \nu = -\frac{\left[M_{11}^0(\eta) + M_{22}^0(\eta)\right]}{2}\cdot\bkappa\ \pm\ \sqrt{\mathcal{Q}^{(\eta)}(\bkappa)},
\label{clf8} \end{equation}
where $\mathcal{Q}^{(\eta)}(\bkappa)$ is a quadratic form in $\bkappa$ with coefficients depending smoothly on $\eta$. For $\eta=0$, \eqref{clf6} shows that the quadratic equation \eqref{clf7} takes the form:
\[\nu^2 - |\lambda_\sharp|^2\left(\kappa_1^2+\kappa_2^2\right) = 0;\]
hence in \eqref{clf8} we have
\begin{equation}
\left.\left[\ M_{11}^0(\eta) + M_{11}^0(\eta)\  \right]\right|_{\eta=0}=0\ \textrm{and}\ 
\left.\mathcal{Q}^{(\eta)}(\bkappa)\right|_{\eta=0}\ =\ |\lambda_\sharp|^2\left(\kappa_1^2+\kappa_2^2\right)
\label{clf9}
\end{equation}
Therefore, for $|\eta|<\eta_1$ (small) \eqref{clf8} takes the form
\begin{equation}
\label{clf10}
\nu = \eta{\bf b}^{(\eta)}\cdot\bkappa \pm \sqrt{\mathcal{Q}^{(\eta)}(\bkappa)},
\end{equation}
where ${\bf b}^{(\eta)}\in\mathbb{R}^2$ depends smoothly on $\eta$, $\mathcal{Q}^{(\eta)}(\bkappa)$ is a quadratic form in $\bkappa$ (depending smoothly on $\eta$) and 
\begin{equation}\label{clf11}
\left( |\lambda_\sharp|^2-C|\eta| \right)(\kappa_1^2+\kappa_2^2)\
 \le\ \mathcal{Q}^{(\eta)}(\bkappa)\ \le\ \left( |\lambda_\sharp|^2+C|\eta| \right) (\kappa_1^2+\kappa_2^2),
 \end{equation}
 for $|\eta|\le\eta_1$ and $\bkappa=(\kappa_1,\kappa_2)\in\mathbb{R}^2$. Thus, the solutions of 
 $\det\mathcal{M}_0(\nu,\bkappa;\eta)=0$ are given by \eqref{clf10} and \eqref{clf11}.
 
 We may now pass from solutions of $\det\mathcal{M}_0(\nu,\bkappa;\eta)=0$ to solutions of\\ 
  $\det\mathcal{M}(\nu,\bkappa;\eta)=0$ as in the proof of Proposition \ref{M0pM1} . Thus the 
  $L^2_{\bK^{(\eta)}+\bkappa}$- eigenvalues of $H^{(\eta)}$ are given by 
  \begin{equation}
  \mu_\pm(\bkappa) = \eta {\bf b}^{(\eta)}\cdot\bkappa\ \pm\ 
  \sqrt{\mathcal{Q}^{(\eta)}(\bkappa)}\cdot\left(\ 1+E_\pm(\bkappa;\eta)\ \right),
  \nn\end{equation}
  where $E_\pm(\bkappa;\eta)\le C|\bkappa|$ for $|\eta|\le\eta_1,\ |\bkappa|\le\tilde{\kappa}$. The proof of Theorem \ref{deformed} is complete.

   \bigskip

\begin{remark}[Instability of the Dirac Point and smooth dispersion surfaces]\label{instability}
We here note a class of perturbing potentials, $W$, such that although $-\Delta+V$ has Dirac (conical) points,
the operator $-\Delta+V+\eta W$ has a locally smooth dispersion surface near the vertices of $\brill_h$. Assume that $V$ is a honeycomb lattice potential, which is inversion-symmetric with respect to $\bx=0$, {\it i.e.} $\bx_0=0$ in Definition \ref{honeyV}, {\it i.e.} $V(-\bx)=V(\bx)$.  Let $W\in C^\infty(\mathbb{R})$, $\Lambda_h$- periodic, but  without the requirement that $W(\bx)= W(-\bx)$ for all $\bx$. Then,  typically 
$\langle \Phi_1,W\Phi_1\rangle\ne\langle \Phi_2,W\Phi_2\rangle$.  In this case, 
\[ A_{11}(\mu^{(1)},\bK^1,\eta)\ne A_{22}(\mu^{(1)},\bK^1,\eta);\]
 see \eqref{Adef} . For $\mu^{(\eta)}=\mu^{(0)}+\eta \mu^{1,\eta}$ to be an $L^2_{\bK^{(\eta)}}$ eigenvalue, we found that it is necessary and sufficient that:
 \[ \det\mathcal{M}\left(\mu^{(1,\eta)},\bK^1,\eta\right) = 0.\]
or equivalently
  \begin{equation}
\ \left(\ \mu^{(1,\eta)} - \frac{A_{11} + A_{22}}{2} \right)^2\ -\ \left(\frac{A_{11} - A_{22}}{2} \right)^2\ -\ |B|^2\ =\ 0\ .
\label{eigWnot}
\end{equation}

Thus, our eigenvalue equation becomes:
\begin{equation}
\mu^{(1,\eta)} = \frac{A_{11} + A_{22}}{2}\ \pm\ 
\sqrt{ \left(\frac{A_{11} - A_{22}}{2} \right)^2 + |B|^2 }
\label{mu1pm}
\end{equation}
When, $A_{11}-A_{22}\ne0$, each sign in \eqref{mu1pm} gives rise to an equation to which we may apply the implicit function theorem to obtain a smooth function $(\bK^1,\eta)\mapsto\mu^{(1,\eta)}(\bK^1,\eta)$. In particular, at $\eta=0$, equation \eqref{mu1pm} gives
{\small{
\begin{align}
&\ \left(\ \mu^{(1,0)} -  \int |\Phi_1(\bx)|^2\ \frac{W(\bx)+W(-\bx)}{2}\ d\bx \right)^2\nn\\ 
& =\ \left( \int |\Phi_1(\bx)|^2\ \frac{W(\bx)-W(-\bx)}{2}\ d\bx \right)^2\ +\  \left|\ \overline{ \lambda_\sharp } \left(K_1^{1,0}+i K_2^{1,0}\right)
 +\langle \Phi_1,W \Phi_2 \rangle\ \right|^2\ .
\label{eigWnot1}
\end{align} 
}}
Therefore, for small $\eta$, the two signs in \eqref{mu1pm} give rise to two distinct solutions $\mu^{(1,\eta)}_\pm$. Thus, for small non-zero $\eta$,  the double eigenvalue disappears and the dispersion
surface is smooth.

\end{remark}


\bigskip\bigskip\bigskip

\appendix

\section{Topological obstruction}\label{sec:topology}

In section \ref{sec:continuation} there arises the situation of an $N\times N$ complex matrix $A$ varying within the space of rank $N-1$ matrices. It was of interest to know whether one can construct 
 a \underline{non-zero} nullvector which is an analytic function of the entries of $A$.
In this section we provide a $2\times2$ matrix counterexample
that exhibits a topological obstruction.\medskip

Let $\mathbb{M}\subset {\rm Mat}(2)$ denote the space of $2\times2$ complex matrices of rank $1$. We prove the following
\begin{proposition}
There is no continuous map $\phi:\mathbb{M}\to\mathbb{C}^2\setminus\{0\}$ such that $\phi(A)\in {\rm Nullspace}$(A) for each $A\in\mathbb{M}$.
\end{proposition}

\begin{proof} Let $\phi$ denote such a map. We proceed to derive a contradiction. For vectors $v=\left(\begin{array}{c} v_1\\ v_2\end{array}\right)\in\mathbb{C}^2\setminus\{0\}$, define the  $2\times2$ complex rank $1$ matrix:
\begin{equation}
A(v)\  =\ \bar{v}\otimes Jv\  =\ \bar{v}\left(Jv\right)^T\ ,
\label{Av-def}
\end{equation}
where $J$ is skew symmetric and non-singular.
Note:
$v\mapsto A(v)$ is a continuous map from $\mathbb{C}^2\setminus\{0\}$ to $\mathbb{M}$.  By skew-symmetry of $J$,  $A(v)v=0$ and therefore
\[ {\rm Nullspace}(A(v))\ =\ \mathbb{C}\times v\ \textrm{for each}\ v\in \mathbb{C}^2\setminus\{0\}\ .\]
Hence, for each $v\in \mathbb{C}^2\setminus\{0\}$ there is one and only one \underline{non-zero} complex number $\lambda(v)$ such that 
\begin{equation}
\phi(A(v))=\lambda(v) v\ .\label{formofA}
\end{equation}
Since $\phi$ is assumed continuous, the map $v\mapsto\lambda(v)$ is continuous from $\mathbb{C}^2\setminus\{0\}$ to $\mathbb{C}\setminus\{0\}$. Moreover, for all $v\in\mathbb{C}^2\setminus\{0\}$ and $\theta\in\mathbb{R}$: $ A(e^{i\theta}v)=A(v)$. Hence, by \eqref{formofA}
\begin{align}
&\lambda\left(e^{i\theta}v \right) \cdot e^{i\theta}v \ =\ \phi\left( A(e^{i\theta}v) \right)\ =\ \phi\left( A(v) \right)\ 
= \lambda(v) v \nn\\ &\textrm{and therefore}\ \lambda \left( e^{i\theta}v \right) =\ e^{-i\theta}\ \lambda(v)\ .
\label{lambdatheta}
\end{align}
\medskip

Now for every $\theta\in S^1=\mathbb{R}/2\pi\mathbb{Z}$ and $t\in[0,1]$, let 
\begin{align}
v(\theta;t)\ &\equiv\ t\ e^{i\theta}\ \hat{e}_1\ + (1-t)\ \hat{e}_2, {\rm where}\ 
\hat{e}_1= \left(\begin{array}{c} 1 \\ 0\end{array}\right),\ \ \hat{e}_2= \left(\begin{array}{c} 0 \\ 1\end{array}\right) \ .\label{v-theta-t}
\end{align}
Note $v(\theta;t)\ \in\ \mathbb{C}^2\setminus\{0\}$ and introduce, for $(\theta,t)\in S^1\times[0,1]$
\begin{align}
\zeta(\theta;t)\ &\equiv\ \lambda(v(\theta;t)) . \label{zeta-def}
\end{align}

\noindent We think of $\theta\mapsto\zeta(\theta;t)$ as a 1-parameter ($t\in[0,1]$) family of closed curves in $\mathbb{C}^1\setminus\{0\}$. 
\begin{align*}
&\textrm{Taking}\ t=0, \textrm{we have}\  \zeta(\theta;0) = \lambda\left(\hat{e}_2\right),\ \textrm{for all}
\ \theta\in S^1;\ {\rm and}\\
&\textrm{taking}\  t=1,\ \textrm{we have}\ \zeta(\theta;1) = \lambda\left(e^{i\theta}\hat{e}_1\right)\ =\ e^{-i\theta}\lambda\left(\hat{e}_1\right),\ \textrm{for all}\ \theta\in S^1\ ,
\end{align*}
by \eqref{lambdatheta}.
Thus by varying $t$ between $0$ and $1$ we obtain a continuous deformation of the unit circle to a point, remaining in $\mathbb{C}\setminus\{0\}$. This is impossible. 
 \end{proof}

\bibliographystyle{siam}
\bibliography{honey}
\end{document}

%% file: abbrevs3.tex
\newcommand{\nc}{\newcommand}

\nc{\fig}[4]
{
  \begin{figure}[ht!]  
    \centering{\scalebox{#1}{\includegraphics*{./figures/#2.eps}}}
    \caption{#4}
    \label{fig:#3}
  \end{figure}
}

\nc{\bXYZ}{{\bf XYZ\ }}

\nc{\nn}{\nonumber} 
\nc{\nit}{\noindent}
\nc{\marginnote}[1] {\marginpar{\tiny #1}}
\nc{\SH}{Schr\"odinger}
\nc{\NLS}{nonlinear Schr\"odinger}
\nc{\ie}{\emph{i.e.\ \mbox{}}}
\nc{\eg}{\emph{e.g.\ \mbox{}}}
\nc{\per}{{l_j}}
\nc{\perx}{q}
\nc{\vol}{{\rm Vol}}

\nc{\lapx}{\dfrac{\partial^2}{\partial x^2}}
\nc{\lapy}{\dfrac{\partial^2}{\partial y^2}}
\nc{\lapxy}{\dfrac{\partial^2}{\partial xy}}
\nc{\diff}[2]{\frac{d #1}{d #2}}
\nc{\diffn}[3]{\frac{d^{#3} #1}{d {#2}^{#3}}} 
\nc{\pdiff}[2]{\frac{\partial #1}{\partial #2}} 
\nc{\pdiffn}[3]{\frac{\partial^{#3} #1}{\partial{#2}^{#3}}} 

\def\Xint#1{\mathchoice
  {\XXint\displaystyle\textstyle{#1}}%
  {\XXint\textstyle\scriptstyle{#1}}%
  {\XXint\scriptstyle\scriptscriptstyle{#1}}%
  {\XXint\scriptscriptstyle\scriptscriptstyle{#1}}%
  \!\int}
\def\XXint#1#2#3{{\setbox0=\hbox{$#1{#2#3}{\int}$} \vcenter{\hbox{$#2#3$}}\kern-.5\wd0}}
\def\dashint{\Xint-}
\nc{\Av}{\dashint_\cell}
\nc{\avg}[1]{\mbox{$\left \langle\ \!#1\!\!\ \right \rangle$}}
\nc{\intRd}{\int_{{\mathbb R}^d}}

\nc{\abs}[1] {\lvert #1 \rvert} 
\nc{\norm}[2] {{\lVert #1 \rVert}_{#2}} 
\nc{\normH}[1]{\|#1\|_{H^1}^2}
\nc{\normL}[1]{\|#1\|_2^2}
\nc{\nl}[1]{|#1|^{2\sigma}}
\nc{\nlb}[1]{#1^{2\sigma+1}}
\nc{\Linvd}{L_\delta^{-1}}
\nc{\Linvz}{L_0^{-1}}
\nc{\Linvzs}{L_0^{-2}}

\nc{\st}[1]{\stackrel{(\ref{eq:#1})}{=}}
\nc{\stt}[2]{\stackrel{(\ref{eq:#1}),(\ref{eq:#2})}{=}}

\nc{\veps}{\epsilon}
\nc{\Ue}{U_\eps}
\nc{\koe}{\frac{k}{\veps}} 

\nc{\wrat}{w_{\rm ratio}}
\nc{\aij}{A^{ij}}                  
\nc{\minv}{m_*^{-1}} 
\nc{\sqminv}{m_*^{-\frac{1}{2}}}
\nc{\G}{\gamma_{\rm ef\,\!f}}
\nc{\zetap}{{\zeta_*}}
\nc{\zetas}{{\zeta_{1*}}}
\nc{\Pe}{P_{\rm edge}}
\nc{\slope}{\dfrac{d\cP[u(\cdot;\mu)] }{d\mu}}
\nc{\moot}{\mu_2}
\nc{\mum}{\mu_{\rm min}}

\def\R{\mathbb{R}}
\def\C{\mathbb{C}}
\def\Z{\mathbb{Z}}

\nc{\brill}{{\mathcal{B}}}
\nc{\cell}{{\Omega}}
\nc{\bj}{{\bf j}}
\nc{\bk}{{\bf k}}
\nc{\bl}{{\bf l}}
\nc{\bm}{{\bf m}}
\nc{\bn}{{\bf n}}
\nc{\bkappa}{{\mathbf \kappa}}
\nc{\bK}{{\bf K}}
\nc{\bKp}{{\bf K'}}
\nc{\bp}{{\bf p}}
\nc{\bpi}{{\mathbf{\pi}}}
\nc{\bq}{{\bf q}}
\nc{\br}{{\bf r}}
\nc{\bv}{{\bf v}}
\nc{\bx}{{\bf x}}
\nc{\bX}{{\bf X}}
\nc{\by}{{\bf y}}
\nc{\bz}{{\bf z}}
\nc{\bnu}{{\bf \nu}}
\nc{\bxc}{{\bf x}_c}
\nc{\bxi}{{\bold\xi}}

\nc{\cP}{{\cal P}}
\nc{\cPedge}{{\cal P}_{edge}}
\nc{\cPc}{{\cal P}_{cr}} 
\nc{\cDD}{\Lambda} 
\nc{\cF}{{\cal F}}
\nc{\cG}{{\cal G}}
\nc{\cO}{{\cal O}}  
\nc{\cQ}{{\cal Q}}  
\nc{\cR}{{\cal R}}
\nc{\cI}{{\cal I}}
\nc{\cK}{{\cal K}}
\nc{\cL}{{\cal L}} 
\nc{\cM}{{\cal M}}
\nc{\cN}{{\cal N}} 
\nc{\cE}{{\cal E}}
\nc{\cH}{{\cal H}} 
\nc{\cX}{{\cal X}}
\nc{\cZ}{{\cal Z}} 
\nc{\cT}{{\cal T}} 
\nc{\order}{{\cal O}}

\newtheorem{remark}{Remark}[section]

\newcommand{\kv}{\mathbf{k}}

\newcommand{\x}{\mathbf{x}}

\newcommand{\eps}{\varepsilon}

\newtheorem{thm}{Theorem}[section]
\newtheorem{rem}[thm]{Remark}

\nc{\pulse}{%
  \begin{picture}(60,50)(0,-20)
  \qbezier(0,0)(5,0)(10,15)
  \qbezier(10,15)(15,30)(20,30)
  \qbezier(20,30)(22,30)(25,18)
  \qbezier(25,18)(27,10)(30,10)
  \qbezier(30,10)(34,10)(38,15)
  \qbezier(38,15)(42,20)(45,20)
  \qbezier(45,20)(49,20)(53,10)
  \qbezier(53,10)(57,0)(60,0)
  \multiput(0,0)(0,-2){10}{\line(0,-1){1}}
  \multiput(60,0)(0,-2){10}{\line(0,-1){1}}
  \put(30,-20){\vector(-1,0){30}}
  \put(30,-20){\vector(1,0){30}}
  \put(0,-15){\makebox(60,10){{\tiny $T$}}}
  \end{picture}}